\documentclass{article}

\usepackage{arxiv}

\usepackage[utf8]{inputenc} 
\usepackage[T1]{fontenc}    
\usepackage{url}            
\usepackage{booktabs}       
\usepackage{amsfonts}       
\usepackage{nicefrac}       
\usepackage{microtype}      
\usepackage[colorlinks = TRUE, linkcolor=blue,citecolor=blue,urlcolor=blue]{hyperref}
\usepackage{lipsum}         
\usepackage{doi}

\usepackage{amsthm,amsmath,amssymb,enumitem,multirow,colortbl,caption,subcaption,float,appendix}
\usepackage[authoryear]{natbib}
\usepackage{graphicx}

\usepackage{cleveref}       
\usepackage{comment}

\newcommand{\revise}{\color{black}}
\newcommand{\revisetwo}{\color{black}}
\newcommand{\revisethree}{\color{black}}

\usepackage{indentfirst}

\title{Bayesian modal regression based on mixture distributions}

\author{ \href{https://orcid.org/0000-0003-3265-6330}{\includegraphics[scale=0.06]{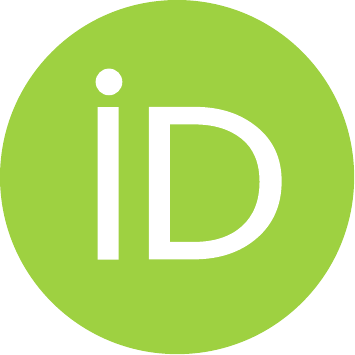}\hspace{1mm}Qingyang Liu} \\
	Department of Statistics\\
	University of South Carolina\\
	Columbia, SC 29201 \\
	\texttt{qingyang@email.sc.edu} \\
	\And
	\href{https://orcid.org/0000-0001-7077-0869}{\includegraphics[scale=0.06]{orcid.pdf}\hspace{1mm}Xianzheng Huang} \\
	Department of Statistics\\
	University of South Carolina\\
	Columbia, SC 29201 \\
	\texttt{huang@stat.sc.edu} \\
	\And
	\href{https://orcid.org/0000-0002-7190-7844}{\includegraphics[scale=0.06]{orcid.pdf}\hspace{1mm}Ray Bai} \\
	Department of Statistics\\
	University of South Carolina\\
	Columbia, SC 29201 \\
	\texttt{rbai@mailbox.sc.edu} \\
}

\renewcommand{\shorttitle}{Bayesian Modal Regression}

\hypersetup{
	pdftitle={Bayesian Modal Regression based on Mixture Distributions},
	pdfsubject={statistics},
	pdfauthor={Qingyang Liu, Xianzheng Huang, Ray Bai},
	pdfkeywords={Mode, Heavy-tailed distribution, Outlier, Unimodal distribution},
}

\numberwithin{equation}{section}
\theoremstyle{plain}

\newtheorem{theorem}{Theorem}[section]
\newtheorem{lemma}{Lemma}[section]
\newtheorem{proposition}{Proposition}[section]

\begin{document}
\maketitle

\begin{abstract}
Compared to mean regression and quantile regression, the literature on modal regression is very sparse. {\revisethree A unifying framework for Bayesian modal regression is proposed, based on a family of unimodal distributions indexed by the mode, along with other parameters that allow for flexible shapes and tail behaviors.} Sufficient conditions for posterior propriety under an improper prior on the mode parameter are derived. Following prior elicitation, regression analysis of simulated data and datasets from several real-life applications are conducted. Besides drawing inference for covariate effects that are easy to interpret, prediction and model selection under the proposed Bayesian modal regression framework are also considered. Evidence from these analyses suggest that the proposed inference procedures are very robust to outliers, enabling one to discover interesting covariate effects missed by mean or median regression,  and to construct much tighter prediction intervals than those from mean or median regression. Computer programs for implementing the proposed Bayesian modal regression are available at \href{https://github.com/rh8liuqy/Bayesian_modal_regression}{https://github.com/rh8liuqy/Bayesian\_modal\_regression}.
\end{abstract}

\keywords{Mode \and Fat-tailed distribution \and Outlier \and Unimodal distribution \and Robust statistics}

\section{Introduction}

There is an abundance of literature on mean regression models which model the conditional mean of a response variable $Y$ given a set of covariates $\boldsymbol{X}$. However, it is no secret that the mean is sensitive to outliers. Median regression -- or more generally, quantile regression -- is robust to outliers and is thus an appealing alternative to mean regression \citep{koenker2017handbook}. Besides the mean and median, the mode is yet another commonly used measure of central tendency. Compared with mean or median regression, modal regression concerns the conditional \textit{mode} of $Y$ given $\boldsymbol{X}$ and is much less explored  \citep{sager1982maximum,lee1989mode,lee1993quadratic}, especially in the parametric framework.

But why are modal regression models useful additions to the well-established mean and median regression models? For unimodal and asymmetric distributions, intervals around the conditional mode typically have higher coverage probability than intervals of the same length around the conditional mean or median \citep{yao2014new,xiang2022modal}. Consequently, prediction intervals from modal regression tend to be narrower than those for mean or median regression when data arise from a unimodal and skewed distribution. Thanks to the nature of the mode, modal regression is extremely robust to outliers that can obscure some inherent covariate effect suggested by the majority of observations, making it a worthy rival of median regression as an alternative to mean regression in regard to feature discovery. By construction, modal regression explores the relationship between the ``most probable'' value of $Y$ {\revisetwo and} $\boldsymbol{X}$, and thus offers a highly interpretable representative value of the response. {\revise For example, in precipitation forecasting, it is easier and of greater public interest to apprehend the most likely predicted rainfall (i.e. the mode of precipitation amount) rather than the predicted mean or median rainfall \citep{Dalenius1965}.}

A major challenge in building parametric modal regression models is constructing an appropriate distribution family that subsumes asymmetric, symmetric, light-tailed, and fat-tailed {\revise mode-zero error} distributions {\revise to flexibly model the distribution of $Y$ given $\boldsymbol{X}$, with the location parameter being the $\boldsymbol{X}$-dependent mode}. In this paper, we propose the general unimodal distribution (GUD) family, which is a subfamily of the general two-component mixture distribution family described in Section \ref{sec:GUD}. Members of the GUD family have a location parameter as the mode, in addition to shape and scale parameters that control the skewness and tail behaviors. Thus, our framework is appropriate for both asymmetric \textit{and} symmetric conditional distributions, as well as both light-tailed \textit{and} fat-tailed distributions. In the extreme case, our framework can also model data from distributions without any finite moments, which we introduce in Section \ref{sec:TPSD}. {\revise Such flexibility can provide a higher level of protection from model misspecification than mean regression, which requires at least the first moment of the error distribution to exist.} We propose to estimate the conditional mode and the shape/scale parameters using a Bayesian approach. By placing appropriate prior distributions on model parameters, our modal regression models can be implemented straightforwardly using Markov chain Monte Carlo (MCMC) and provide natural uncertainty quantification through the posterior distributions.

\subsection{Existing work on modal regression}
Frequentist nonparametric modal regression has been the mainstream in the limited existing literature on modal regression \citep{yao2014new, chen2016nonparametric, 10.1214/19-EJS1607}. For readers interested in frequentist nonparametric modal regression, we refer to \citet{chen2018modal} for a comprehensive review. The higher statistical efficiency and greater interpretability of covariate effects under a parametric framework motivate some recent development in frequentist parametric modal regression. For example, \citet{aristodemou2014new} and \citet{bourguignon2020parametric} proposed a parametric modal regression model based on a  gamma distribution for a positive response; \citet{zhou2020parametric} proposed two parametric modal regression models for a bounded response. \citet{menezes2021collection} give a nice review on these and other parametric modal regression models for a bounded response.  In contrast to these existing parametric modal regression models for {\revise positive or bounded} data, the modal regression models in the present manuscript are based on a \textit{new} GUD family whose support is the \textit{entire} real line. Furthermore, our work deals with \textit{Bayesian} inference for modal regression.

The literature on Bayesian modal regression is even more sparse. \citet{yu2012bayesian} proposed a nonparametric Bayesian modal regression model using Dirichlet process mixtures of uniform distributions. \citet{zhou2022bayesian} proposed a parametric Bayesian modal regression model based on a four-parameter beta distribution whose support is bounded yet unknown.  \citet{damien2017bayesian} introduced a more flexible parametric form of Bayesian modal regression using mixtures of triangular densities for a  response with an unknown bounded support. Remaining in the parametric framework, a major strength of our proposed GUD family is that it naturally facilitates data-driven learning of the skewness and tails of the underlying distribution supported on the entire real line, while signifying the mode as the central tendency measure of the response. 

\subsection{{\revise Our contributions}}

This paper aims to widen the scope of Bayesian modal regression models and highlight the advantages of these models through analyses of datasets from real-life applications in several disciplines. 
The main contributions of this paper can be summarized as follows: 
\begin{enumerate}
	\item We propose the GUD family that is suitable for Bayesian modal regression. The GUD family contains distributions that are symmetric or asymmetric, (non)normal, and/or fat-tailed.
	\item We formulate rules of prior elicitation for the GUD family. In particular, we place a flat prior on regression coefficients, weakly informative priors on all other model parameters, and establish sufficient conditions under which the posterior distribution is proper.
	\item We provide strategies for constructing prediction intervals and selecting {\revise an appropriate member of the
		GUD family} for Bayesian modal regression analysis.
	\item We illustrate the following benefits of our proposed Bayesian modal regression framework through simulation studies and data applications in economics, criminology, real estate, and molecular biology: a) robustness to outliers, b) {\revise precise} prediction, and c) high interpratability of covariate effects.  
\end{enumerate}

{\revise By accomplishing the first three items above, we provide a comprehensive toolkit for modal regression analysis based on a flexible family containing a large variety of unimodal distributions. {\revisetwo Our fully Bayesian approach allows researchers to easily infer covariate effects through their posterior distributions instead of relying on asymptotic approximations or resampling methods like the bootstrap \citep{boos2013essential}.} In a similar spirit, constructing prediction intervals based on the posterior predictive distribution also becomes straightforward, especially with a simple algorithm for sampling from the GUD family (discussed in Section~\ref{sec:GUD}).}

The structure of this paper is as follows. In Section \ref{sec:MotivatingExamples}, we motivate our proposed Bayesian modal regression framework with two data applications. In Section \ref{sec:GUD}, we formally define the GUD family and zoom in on several important members in the family. Section \ref{sec:Bayesian_inference} introduces Bayesian inference for these modal regression models, including prior elicitation, posterior propriety, uncertainty quantification, and model selection. Section \ref{sec:Simulation} provides simulation studies that illustrate the strengths of our methodology. In Section \ref{sec:data_application}, we provide two additional data applications from real estate and molecular biology. Section \ref{sec:Discussion} concludes the paper with some remarks about our Bayesian modal regression framework and several directions for future research. 

\section{Motivating applications} \label{sec:MotivatingExamples}

As a prelude to introducing our Bayesian modal regression framework, we first present results from applying the proposed methodology (to be elaborated in Sections \ref{sec:GUD} and \ref{sec:Bayesian_inference}) to datasets from the economics and criminology literature. 

\subsection{Modeling highly right-skewed bank deposits}\label{sec:bank}
It is common knowledge to economists that wealth distributions are highly skewed to the right \citep{benhabib2018skewed}. The cumulative nature of wealth not only has impact on individuals' net worth, but also has an influence on assets of large companies, including bank holding companies. In this example, we analyzed the deposits data of 50 banks and savings institutions in the United States on July 2, 2010 (Table 3.4.1 in \citet{siegel2016practical}). 

Figure \ref{fig:bank_deposits} presents the estimated density plot that results from fitting an intercept-only regression model based on the Double Two-Piece-Student-$t$, or DTP-Student-$t$, distribution (to be introduced in Section~\ref{sec:GUD}) to the dataset, along with the histogram of the observed data. From this figure, we can see that the estimated mode using the DTP-Student-$t$ distribution is close to the nonparametric mode estimate based on the histogram. This similarity and the close resemblance of the fitted density to the shape of the histogram indicate that the DTP-Student-$t$ distribution is an adequate choice for the bank deposits data. 
\begin{figure}[ht]
	\centering
	\includegraphics[width=\linewidth]{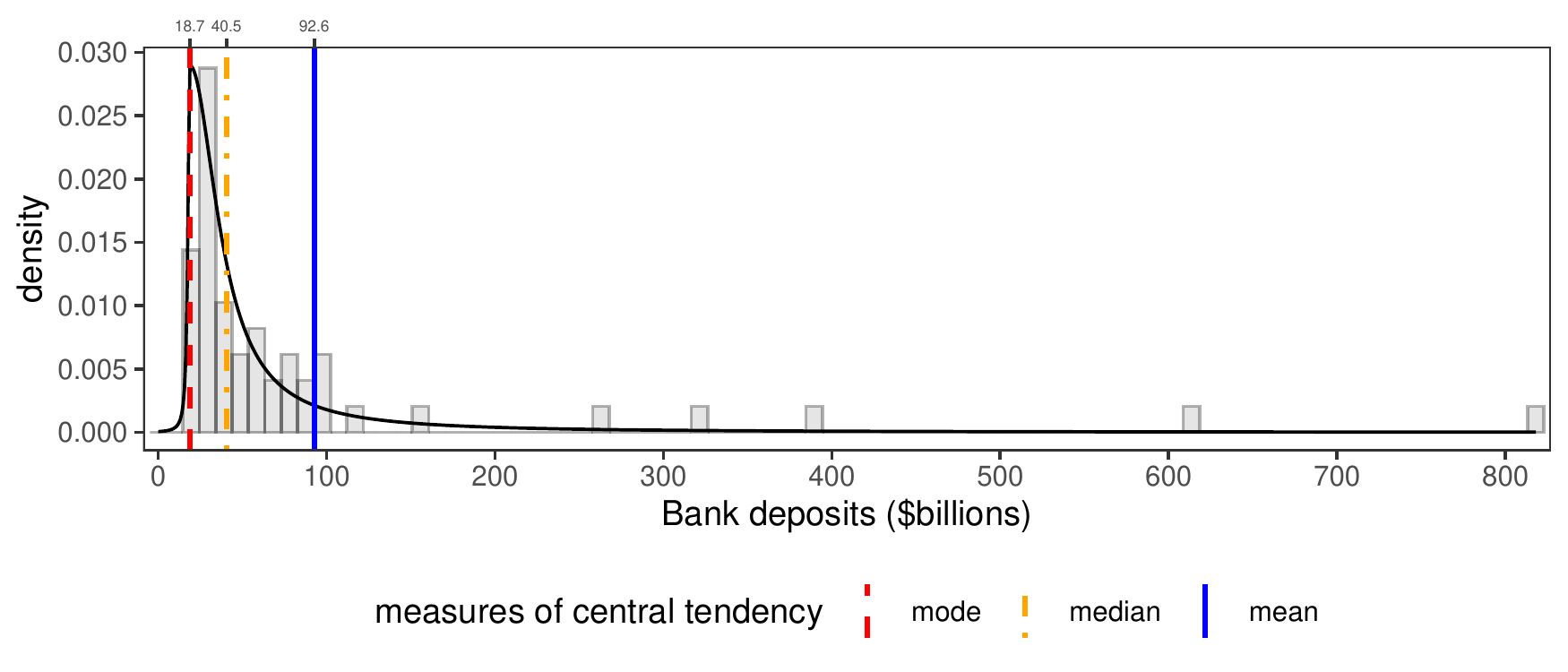}
	\caption{\small {\it Deposits (in billions of dollars) of 50 banks and savings institutions in the United States on July 2, 2010. The solid black curve is the estimated density of the DTP-Student-$t$ distribution. The three vertical lines mark locations of the sample mean (blue solid line), the sample median (orange dash-dotted line), and the estimated mode (red dashed line), respectively.}}
	\label{fig:bank_deposits}
\end{figure}

The other two measures of central tendency, i.e. the sample mean and median, are both shown to be larger than the estimated parametric mode in Figure \ref{fig:bank_deposits}. The sample mean, which equals 92.6 billion dollars, is obviously not a good measure of central tendency for most large banks and savings institutions in the United States. In particular, 40 of the 50 banks and savings institutions in our dataset had deposits \textit{less} than 92.6 billion dollars on July 2, 2010. The sample median for this data is 40.5 billion dollars, indicating that  $50\%$ of banks in the dataset had deposits larger than 40.5 billion dollars while the other half had deposits smaller than 40.5 billion dollars. In spite of its high interpretability, the (sample) median is usually difficult to visualize either from a  density plot or a histogram. In contrast, it is much easier for data analysts to locate and interpret the mode than the mean or median in Figure \ref{fig:bank_deposits}. The estimated mode using the DTP-Student-$t$ distribution is where the density plot {\revise reaches its peak} and is close to where the histogram {\revise reaches its peak}. More specifically, the posterior mean of the mode is around $20$ billion dollars, suggesting that banks in the United States are \textit{most likely} to have deposits of around $20$ billion dollars during that time.

\subsection{Modal versus mean and median regression for analyzing murder rates}\label{sec:murder}

As a second motivating example, we analyze a dataset from \citet{Agresti2021} containing the murder rate, percentage of college education, poverty percentage, and metropolitan rate for the 50 states in the United States and the District of Columbia (D.C.) from 2003. The murder rate is defined as the annual number of murders per $100{,}000$ people in the population. The poverty percentage is the percentage of residents with income below the poverty level, and the metropolitan rate is defined as the percentage of population living in the metropolitan area. 

{\revise At the stage of exploratory data analysis, we first created the conditional scatter plot matrix of the complete data (see the left panel in Figure \ref{fig:UScrime}), in which D.C. stands out as an outlier. This outlier is extreme enough to visually obscure potential association between the considered variables. We thus created a second conditional scatter plot matrix after removing D.C. (see the right panel in Figure~\ref{fig:UScrime}). Now one can more easily perceive a positive association between the poverty percentage and the murder rate, as well as a positive association between the metropolitan rate and the murder rate. Although a visual inspection seems to suggest a negative association between the college percentage and the murder rate, a more thorough regression analysis of the murder rates data is needed to confirm this.}

\begin{figure}[ht]
	\centering
	\includegraphics[width=0.49\linewidth]{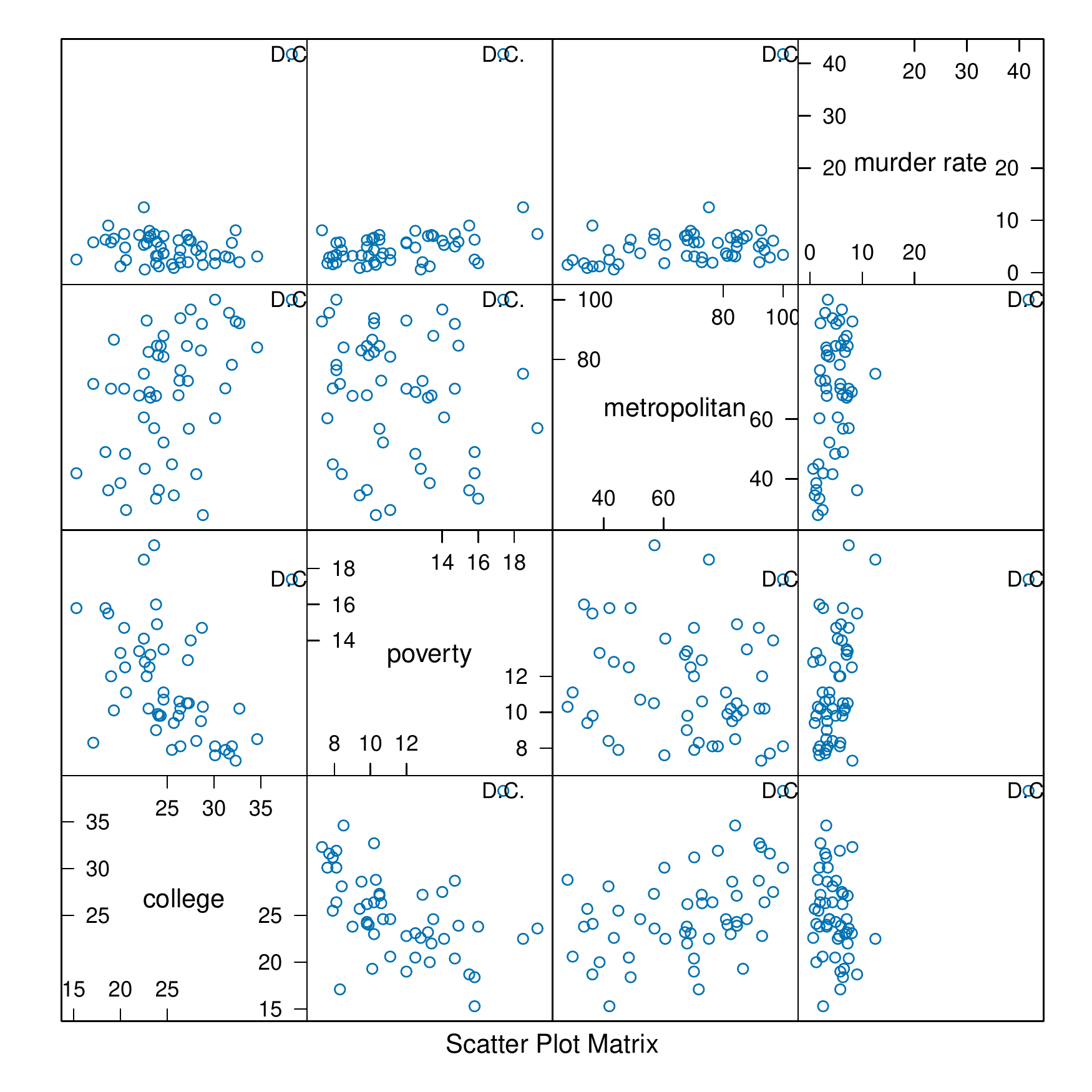}
	\includegraphics[width=0.49\linewidth]{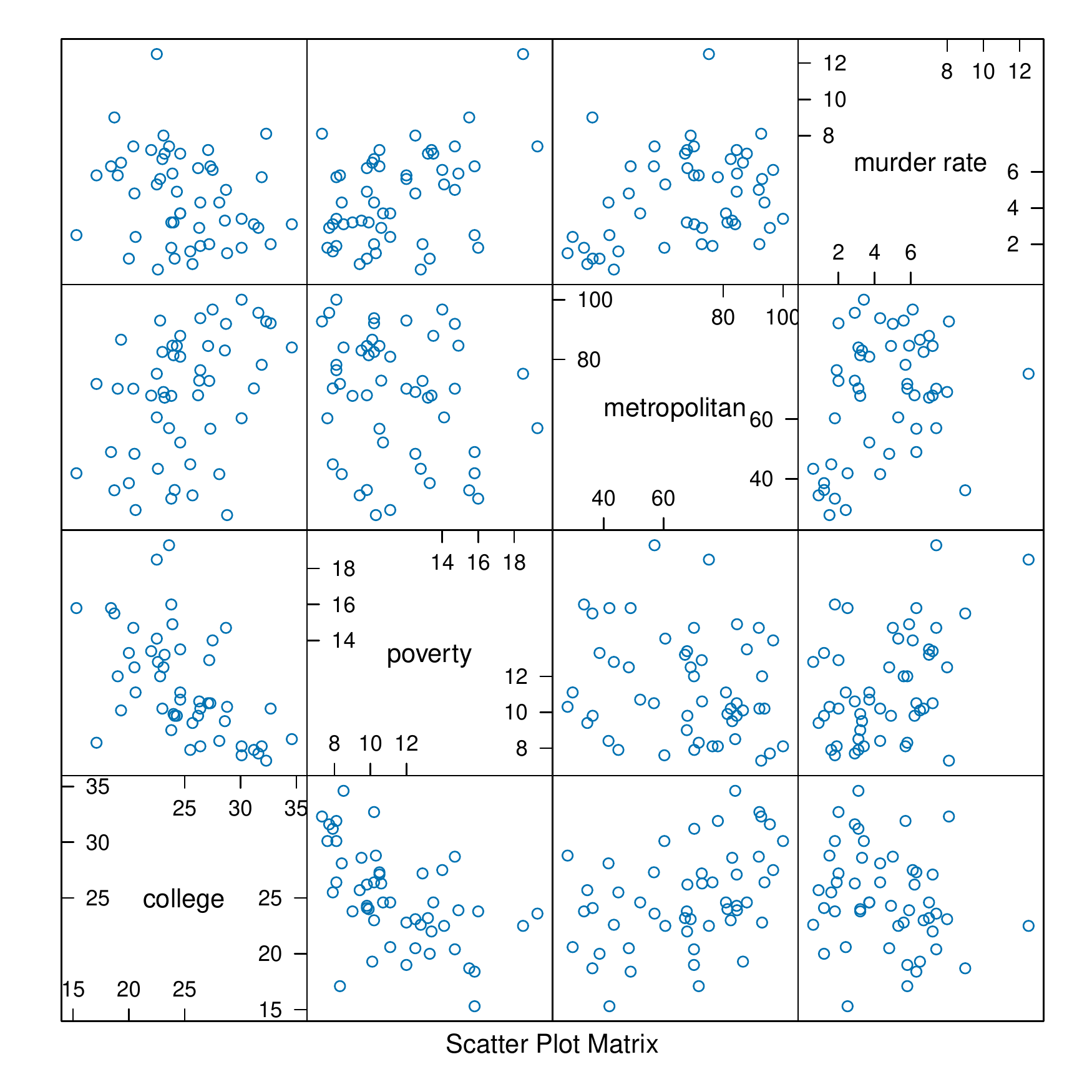}
	\caption{\small {\it The conditional scatter plot matrices of U.S. crime data, with the left matrix including D.C. {\revisetwo that is labeled in the plot}, and the right matrix excluding D.C.}}
	\label{fig:UScrime}
\end{figure}

To formally investigate the association between the murder rate ($Y$) and the aforementioned variables, we fit the following models to the U.S. crime dataset:
$$
\mathbb{M}(Y\mid \boldsymbol{\beta})=\beta_0+\beta_1 \times \operatorname{college}+\beta_2 \times \operatorname{poverty} +\beta_3 \times \operatorname{metropolitan},
$$
where $\mathbb{M}(\cdot)$ generically refers to the conditional mean, median, or mode. {\revise Table \ref{tab:UScrime} presents the inference results from mean/median/modal regression models with D.C. included in the data (see the upper half of Table~\ref{tab:UScrime}) and the counterpart results with D.C. excluded (see the lower half of Table~\ref{tab:UScrime}).} {\revisetwo Some of the results are shared by all three models in both rounds of analyses}. Namely, all models determine that the poverty percentage and the metropolitan rate were both positively associated with the murder rate. However, {\revise the first round of analysis (using the complete data) results in different conclusions about the association between the college percentage and the murder rate.} With a {\revisetwo 90\%} posterior credible interval (CI) of $(0.20, 0.74)$, the mean regression model (specified by \eqref{eq:model.normal}) implies that there exists a \textit{positive} association between the college percentage and the murder rate, conditionally on the other covariates in the model. We believe that this inference result is difficult to justify, in light of existing results from the criminology literature that conclude a negative association between higher education attainment and crime \citep{lochner2020education,hjalmarsson2012impact}. On the other hand, with a  {\revisetwo 90\%}  CI of $(-0.27, 0.05)$, the Bayesian median regression model (formulated in \eqref{eq:model.ALD}) concludes that the college percentage is \textit{not} significantly associated with the murder rate, conditionally on the other covariates. 
Our Bayesian modal regression model with the Two-Piece scale-Student-$t$, or TPSC-Student-$t$, distribution (to be introduced in Section \ref{sec:GUD}) draws a different conclusion. With a  {\revisetwo 90\%} CI of $(-0.33, -0.06)$, our Bayesian modal regression model concludes that there is a \textit{negative} association between the college percentage and the murder rate, which is more consistent with findings from the criminology literature. Lastly, according to the model criterion referred to as the expected log predictive density (ELPD, introduced in Section \ref{sec:Model.selection}) in Table \ref{tab:UScrime}, the modal regression model based on the TPSC-Student-$t$ likelihood yields the highest value of ELPD, indicating a better fit {\revisetwo to} the data than the mean and median regression models. 

{\revise With the D.C. outlier removed from the data in the second round of analysis, the median and modal regression models do in fact suggest a significant negative association between the college percentage and the murder rate, with CIs of $(-0.29,-0.06)$ and $(-0.32,-0.06)$, respectively. Even though the point estimate for the effect of the college percentage on the murder rate in the mean regression model is now negative, the CI of this covariate effect is $(-0.27,0.02)$, and thus one would not conclude it as a significant effect according to mean regression. In fact, Table \ref{tab:UScrime} shows that under the mean regression model, both point estimates and interval estimates for most of the covariate effects change substantially in the second round of analysis. In contrast, the results from our second round of modal regression analysis mostly remain the same as those from the first round. This exercise demonstrates that modal regression based on the proposed GUD family can be even  more robust to outliers than median regression and has the potential to draw reliable inferences and {\revisetwo capture} important features of data even in the presence of extreme outliers.}

\begin{table}
	\caption{\label{tab:UScrime} Estimates of covariate effects for the mean/median/modal regression models fit to the U.S. crime dataset. The mean, $5\%$ quantile, and $95\%$ quantile of the posterior distribution of each covariate effect are listed under Mean, q5, and q95, respectively. ELPD stands for expected log  predictive density.}
	\centering
	\resizebox{\columnwidth}{!}{
		\begin{tabular}[t]{llrlrrr}
			\toprule
			D.C. inclusion & Regression Model & ELPD & Parameter(covariate) & Mean & q5 & q95\\
			\midrule
			&  &  & $\beta_{1}$(college) & 0.47 & 0.20 & 0.74\\
			\cmidrule{4-7}
			&  &  & $\beta_{2}$(poverty) & 1.14 & 0.77 & 1.52\\
			\cmidrule{4-7}
			& \multirow{-4}{*}{\raggedright\arraybackslash Mean regression} & \multirow{-4}{*}{\raggedleft\arraybackslash -161.81} & $\beta_{3}$(metropolitan) & 0.07 & 0.01 & 0.12\\
			\cmidrule{2-7}
			&  &  & $\beta_{1}$(college) & -0.12 & -0.27 & 0.05\\
			\cmidrule{4-7}
			&  &  & $\beta_{2}$(poverty) & 0.44 & 0.21 & 0.68\\
			\cmidrule{4-7}
			& \multirow{-4}{*}{\raggedright\arraybackslash Median regression} & \multirow{-4}{*}{\raggedleft\arraybackslash -133.42} & $\beta_{3}$(metropolitan) & 0.06 & 0.03 & 0.08\\
			\cmidrule{2-7}
			&  &  & $\beta_{1}$(college) & -0.20 & -0.33 & -0.06\\
			\cmidrule{4-7}
			&  &  & $\beta_{2}$(poverty) & 0.24 & 0.01 & 0.46\\
			\cmidrule{4-7}
			\multirow{-12}{*}{\raggedright\arraybackslash D.C. included} & \multirow{-4}{*}{\raggedright\arraybackslash Modal regression} & \multirow{-4}{*}{\raggedleft\arraybackslash -123.24} & $\beta_{3}$(metropolitan) & 0.06 & 0.04 & 0.09\\
			\cmidrule{1-7}
			&  &  & $\beta_{1}$(college) & -0.13 & -0.27 & 0.02\\
			\cmidrule{4-7}
			&  &  & $\beta_{2}$(poverty) & 0.35 & 0.16 & 0.55\\
			\cmidrule{4-7}
			& \multirow{-4}{*}{\raggedright\arraybackslash Mean regression} & \multirow{-4}{*}{\raggedleft\arraybackslash -109.54} & $\beta_{3}$(metropolitan) & 0.06 & 0.04 & 0.09\\
			\cmidrule{2-7}
			&  &  & $\beta_{1}$(college) & -0.18 & -0.29 & -0.06\\
			\cmidrule{4-7}
			&  &  & $\beta_{2}$(poverty) & 0.35 & 0.18 & 0.52\\
			\cmidrule{4-7}
			& \multirow{-4}{*}{\raggedright\arraybackslash Median regression} & \multirow{-4}{*}{\raggedleft\arraybackslash -109.66} & $\beta_{3}$(metropolitan) & 0.05 & 0.03 & 0.08\\
			\cmidrule{2-7}
			&  &  & $\beta_{1}$(college) & -0.20 & -0.32 & -0.06\\
			\cmidrule{4-7}
			&  &  & $\beta_{2}$(poverty) & 0.24 & 0.01 & 0.45\\
			\cmidrule{4-7}
			\multirow{-12}{*}{\raggedright\arraybackslash D.C. excluded} & \multirow{-4}{*}{\raggedright\arraybackslash Modal regression} & \multirow{-4}{*}{\raggedleft\arraybackslash -111.94} & $\beta_{3}$(metropolitan) & 0.06 & 0.04 & 0.09\\
			\bottomrule
		\end{tabular}
	}
\end{table}

\section{The family of general unimodal distributions} \label{sec:GUD}

Having motivated our Bayesian modal regression framework and demonstrated its benefits on two real-life applications in Section \ref{sec:MotivatingExamples}, we now formally introduce the GUD family for Bayesian modal regression.

\subsection{{\revisetwo The formulation as a mixture distribution}}

The probability density function (pdf) of a member belonging to the GUD family is a mixture of two pdfs, $f_1$ and $f_2$, given by
\begin{equation}
	f(y\mid w,\theta,\boldsymbol{\xi}_1,\boldsymbol{\xi}_2) = wf_1(y\mid \theta,\boldsymbol{\xi}_1) + (1-w)f_2(y\mid \theta,\boldsymbol{\xi}_2).
	\label{eq:pdf.GUD}
\end{equation}
In the mixture pdf \eqref{eq:pdf.GUD}, $w \in [0,1]$ is the weight parameter, $\theta \in \left(-\infty,+\infty\right)$ is the mode as {\revisetwo a} location parameter in (\ref{eq:pdf.GUD}), $\boldsymbol{\xi}_1$ consists of parameters other than the location parameter in the pdf $f_1(\cdot\mid \theta, \boldsymbol{\xi}_1)$, and $\boldsymbol{\xi}_2$ is defined similarly for $f_2(\cdot \mid \theta, \boldsymbol{\xi}_2)$. {\revisetwo The supports of the two mixture components, denoted respectively by $\mathcal{D}_1$ and $\mathcal{D}_2$, can be the same or different, as exemplified later in this section. Moreover, the two components are chosen to achieve unimodality and other desirable properties that we elaborate in more detail next.}
Clearly, the GUD family belongs to the more general two-component mixture distribution family.
One feature of GUD that makes it stand out from the bigger family of two-component mixture distributions is that the two component distributions of GUD share the same location parameter $\theta$ as the mode, a feature that makes GUD especially suitable for modal regression. In contrast, a two-component normal mixture for instance, as a widely referenced member in the bigger family,  can be multimodal, and it is non-trivial to impose constraints on two normal components to guarantee unimodality \citep{sitek2016modes}. Even after formulating a unimodal normal mixture, its mode may not have an analytical form \citep{behboodian1970modes}. Many other members in the more general two-component mixture distribution family have the same pitfalls. 

Besides unimodality, we reiterate and complement the following three restrictions on \eqref{eq:pdf.GUD} to make the GUD family suitable and convenient for modal regression:
\begin{enumerate}
	\item[(R1)] The pdfs $f_1(\cdot \mid \theta, \boldsymbol{\xi}_1)$ and $f_2(\cdot \mid \theta, \boldsymbol{\xi}_2)$ are unimodal at $\theta$.
	\item[(R2)] The pdfs $f_1(\cdot \mid \theta, \boldsymbol{\xi}_1)$ and $f_2(\cdot \mid \theta, \boldsymbol{\xi}_2)$ are left-skewed and right-skewed respectively. 
	\item[(R3)] The mixture pdf $f(\cdot \mid w,\theta,\boldsymbol{\xi}_1,\boldsymbol{\xi}_2)$ in \eqref{eq:pdf.GUD} is continuous in its domain.
\end{enumerate}
Restriction (R1) is already implied earlier when we stress that the two components in (\ref{eq:pdf.GUD}) share the same location parameter $\theta$ as the finite mode. In the context of modal regression, (R1) ensures that one can easily link a linear predictor $\boldsymbol{X}^{\top} \boldsymbol{\beta}$ with the conditional mode of $Y$. Because modal regression adds more value to  mean/median regression when data are skewed and contain outliers, we impose (R2) to make members in GUD exhibit a wide range of skewness and tail behaviors. This second restriction also solves the notorious label switching problem that many other two-component mixture distributions suffer from, because $f_1(\cdot\mid \theta,\boldsymbol{\xi}_1)$ and $f_2(\cdot\mid \theta,\boldsymbol{\xi}_2)$ satisfying (R2) must come from different distribution families in some strict sense, as opposed to, say, both coming from the normal family. According to Theorem 1 of \citet{teicher1963identifiability}, this guarantees identifiability of all parameters associated with GUD. Lastly, (R3) eliminates ill-constructed pdfs whose mode may occur at a jump discontinuity.

Henceforth, when a random variable $Y$ follows a distribution in the GUD family, we state that $Y\mid w, \theta, \boldsymbol{\xi}_1, \boldsymbol{\xi}_2 \sim \operatorname{GUD}\left(w, \theta, \boldsymbol{\xi}_1, \boldsymbol{\xi}_2\right)$. Like for other two-component mixture distributions, one may view $Y=ZX_1+(1-Z)X_2$, where $X_1 \mid \theta, \boldsymbol{\xi}_1 \sim f_1(\cdot \mid \theta,\boldsymbol{\xi}_1)$, $X_2 \mid \theta, \boldsymbol{\xi}_2 \sim f_2(\cdot \mid \theta,\boldsymbol{\xi}_2)$, and $Z\mid w \sim \operatorname{Bernoulli}\left(w\right)$, with $Z$, $X_1$, and $X_2$ independent. This viewpoint gives rise to a data augmentation method outlined below for generating data from a GUD effortlessly {\revise when samples from $f_{1}$ and $f_{2}$ are easy to obtain}: 
\begin{enumerate}[label=(\roman*), align=left]
	\item Sample $X_1 \mid \theta, \boldsymbol{\xi}_1 \sim f_1\left(\cdot \mid \theta, \boldsymbol{\xi}_1\right)$.
	\item Sample $X_2 \mid \theta, \boldsymbol{\xi}_2 \sim f_2\left(\cdot \mid \theta, \boldsymbol{\xi}_2\right)$.
	\item Sample $Z \mid w \sim \operatorname{Bernoulli}(w)$.
	\item $Y \leftarrow ZX_1 + (1-Z)X_2$.
\end{enumerate}
Having an efficient {\revise sampling} method is especially beneficial in constructing Bayesian prediction intervals, since the most common way to approximate the posterior predictive density is by drawing samples from the posterior predictive distribution during the MCMC iterations. We will continue our discussion about the Bayesian prediction intervals in Section \ref{sec:Model.selection}. 

Relating to existing literature, the GUD family \textit{subsumes} several previously proposed distributions, such as those introduced in \citet{fernandez1998bayesian} and \citet{rubio2015bayesian}, as special cases. In Sections \ref{sec:FG}-\ref{sec:TPSD}, we detail several examples of distributions from the GUD family. {\revisetwo All mixture components in these examples belong to some location-scale family, but $f_1$ and $f_2$ are not limited to location-scale families in general (see an example in Section \ref{sec:logNM} in the Appendix).} Apart from the GUD family, another class of distributions that can accommodate both skewed/symmetric responses and fat-tailed/thin-tailed responses is the skew-normal family that is well explored by \citet{azzalini2013skew}. {\revise There are two well-received parameterizations of the skew-normal family: the direct parameterization (SNDP) and the centered parameterization (SNCP) \citep{arellano2008centred,durante2019conjugate}. Except for the symmetric members in the SNDP family, the location parameter in an SNDP typically cannot be interpreted as the mean, median, or mode. Meanwhile, the SNCP family is indexed by three parameters, with one of them being the mean. In contrast to SNDP and SNCP, the location parameter of the GUD family is the \emph{mode}, which makes it convenient for drawing inference for the (conditional) mode.}

\subsection{The flexible Gumbel distribution} \label{sec:FG} 

For predicting extreme events, the Gumbel distribution is a popular choice in many fields such as hydrology, earthquake forecasting, and insurance \citep{smith2003statistics,vidal2014bayesian,shin2015application}. The pdf of a Gumbel distribution for the maximum is
$$
f_{\operatorname{Gumbel}}\left(y\mid \theta, \sigma\right) = \frac{1}{\sigma}\exp\left\{-\frac{y-\theta}{\sigma}-\exp\left(-\frac{y-\theta}{\sigma}\right)\right\}\mathbb{I}\left(-\infty < y < \infty \right),
$$
where $\theta \in \mathbb{R}$ is the mode as the location parameter, $\sigma > 0$ is the scale parameter, and $\mathbb{I}(\cdot)$ is the indicator function. To describe data that contains a mix of extremely large and extremely small events, \citet{QXZ} proposed the flexible Gumbel (FG) distribution specified by the pdf 
\begin{equation}
	f_{\operatorname{FG}}\left(y \mid w, \theta, \sigma_1, \sigma_2 \right)= wf_{\operatorname{Gumbel}}\left(-y\mid -\theta, \sigma_1 \right) + (1-w)f_{\operatorname{Gumbel}}\left(y\mid \theta, \sigma_2 \right).
	\label{eq:pdf.FG}
\end{equation}
By mapping to (\ref{eq:pdf.GUD}), we have $f_1(y|\theta, \boldsymbol{\xi}_1)=f_{\operatorname{Gumbel}}\left(-y\mid -\theta, \sigma_1 \right)$ as the pdf of the left-skewed Gumbel distribution for the minimum. Similarly, we have $f_2(y|\theta, \boldsymbol{\xi}_2) = f_{\operatorname{Gumbel}}\left(y\mid \theta, \sigma_2 \right)$ as the pdf of the right-skewed Gumbel distribution for the maximum. 
We illustrate Bayesian modal regression based on the FG likelihood in Section \ref{sec:serum}. The FG distribution serves as a good choice of likelihood if the data is a mixture of extreme events, such as monthly maximum/minimum water elevation changes or weekly heaviest/lightest traffic on a highway.

\subsection{The double two-piece distribution} \label{sec:DTP}

\citet{rubio2015bayesian} defined the Double Two-Piece (DTP) distribution by mixing two truncated distributions. For a pdf belonging to some location-scale family of the form $(1/\sigma){\revisetwo g}(\left(y - \theta\right)/\sigma \mid \delta)$ that is unimodal at $\theta$, with a scale parameter $\sigma>0$ and a  shape parameter $\delta$, the pdf of the corresponding left $\theta$-truncated distribution is
\begin{equation}
	f_{\operatorname{LT}}(y\mid \theta,\sigma,\delta) = \frac{2}{\sigma} {\revisetwo g}\left.\left(\frac{y-\theta}{\sigma} \right\vert \delta\right) \mathbb{I}(y<\theta),
	\label{eq:pdf.LT}
\end{equation}
and the corresponding right $\theta$-truncated distribution is specified by the following pdf,
\begin{equation}
	f_{\operatorname{RT}}(y\mid \theta,\sigma,\delta) = \frac{2}{\sigma} {\revisetwo g}\left.\left(\frac{y-\theta}{\sigma} \right\vert \delta\right) \mathbb{I}(y\ge\theta).
	\label{eq:pdf.RT}
\end{equation}
By mixing the pdfs in \eqref{eq:pdf.LT}-\eqref{eq:pdf.RT}, we obtain the DTP pdf as
\begin{equation}
	f_{\operatorname{DTP}}(y\mid \theta, \sigma_1, \sigma_2, \delta_1, \delta_2) = w f_{\operatorname{LT}}(y\mid \theta,\sigma_1,\delta_1) + (1-w) f_{\operatorname{RT}}(y\mid \theta,\sigma_2,\delta_2),
	\label{eq:pdf.DTP}
\end{equation}
where
\begin{equation}
	w=\frac{\sigma_1 {\revisetwo g}\left(0 \mid \delta_2\right)}{\sigma_1 {\revisetwo g}\left(0 \mid \delta_2\right)+\sigma_2 {\revisetwo g}\left(0 \mid \delta_1\right)},
	\label{eq:pdf.DTP.w}
\end{equation}
{\revise and ${\revisetwo g}\left(0\mid \delta\right)$ represents ${\revisetwo g}(\left(y - \theta\right)/\sigma \mid \delta)$ evaluated at $y = \theta$. The weight \eqref{eq:pdf.DTP.w} is chosen to guarantee a mixture distribution that satisfies (R3) {\revisetwo even though this particular choice makes $w\ne 0, 1$ if $g(0 \mid \delta)$ never vanishes}.} Restrictions (R1) and (R2) are trivially satisfied by the construction of the  left/right $\theta$-truncated pdfs in \eqref{eq:pdf.LT}--\eqref{eq:pdf.RT}. Thus, DTP distributions belong to the GUD family.  Note, however, that our general GUD family \eqref{eq:pdf.GUD} does not \textit{require} the two component densities to be truncated, as we demonstrated earlier with the FG distribution specified by the density in \eqref{eq:pdf.FG}. 

As a concrete example, we consider the location-scale family as  the three-parameter Student's $t$ distributions, i.e., the non-standardized Student's $t$ distributions, with location parameter $\theta$, scale parameter $\sigma > 0$, and continuous degree of freedom $\delta > 0$ \citep{geweke1993bayesian}. Following (\ref{eq:pdf.LT}) and (\ref{eq:pdf.RT}), one has the corresponding left-skewed truncated three-parameter Student's $t$ distribution and the right-skewed truncated three-parameter Student's $t$ distribution, respectively. This leads to the distribution defined according to \eqref{eq:pdf.DTP} and \eqref{eq:pdf.DTP.w} that we call the DTP-Student-$t$ distribution. The DTP distribution family contains numerous distributions, all of which are suitable for modal regression (see \citet{rubio2015bayesian} for more). In the sequel, we concentrate on the DTP-Student-$t$ distribution as a special member of the DTP distribution.

\subsection{The two-piece scale distribution} \label{sec:TPSD}

By setting $\delta_1 = \delta_2 = \delta$ in (\ref{eq:pdf.DTP}), one obtains the pdf of a subfamily of the DTP family proposed in \citet{fernandez1998bayesian}, referred to as the two-piece scale (TPSC) distribution family, 
\begin{equation}
	\begin{aligned}		             
		f_{\operatorname{TPSC}}\left(y \vert \, w, \theta, \sigma, \delta \right) = &w f_{\operatorname{LT}}\left(y \left \vert \,\theta,\sigma \sqrt{\frac{w}{1-w}},\delta\right.\right)+\\
		&(1-w)f_{\operatorname{RT}}\left(y \left \vert \,\theta,\sigma\sqrt{\frac{1-w}{w}},\delta\right.\right).
		\label{eq:pdf.TPSC}
	\end{aligned}
\end{equation}

We point out that in \citet{fernandez1998bayesian}, a shape parameter $\gamma  = w^{0.5}(1-w)^{-0.5}$ is used instead of the weight parameter $w$ when formulating the mixture pdf. We adopt the parameterization in (\ref{eq:pdf.TPSC}) because we find it more straightforward to elicit a noninformative prior for $w$ than placing a noninformative prior on $\gamma$. {\revise Interestingly, it is not difficult to show that $w$ follows a uniform distribution on the interval $\left(0,1\right)$ if and only if $\gamma$ follows a log-logistic distribution with the scale parameter as $1$ and the shape parameter as $2$ \citep{ekawati2015moments}.}

Similar to the construction of the DTP-Student-$t$ distribution, we can  construct the TPSC-Student-$t$ distribution by choosing the two component distributions to be the left and right $\theta$-truncated three-parameter Student's $t$ distributions. When $w=0.5$, the TPSC-Student-$t$ distribution converges to a normal distribution with mean $\theta$ and standard deviation $\sigma$ as $\delta\to \infty$; and it reduces to a Cauchy distribution with mode $\theta$ and scale parameter $\sigma$ when $\delta=1$. Hence, even as a special case of the DTP-Student-$t$ distribution, the TPSC-Student-$t$ distribution is flexible enough to describe normally distributed data and non-normal data with extreme outlier(s) from distributions that do not have any finite moments. Since the TPSC-Student-$t$ has fewer parameters than the DTP-Student-$t$ distribution, it is an adequate choice for small datasets. On the other hand, the DTP-Student-$t$ distribution may be preferred when there is moderate sample size. Certainly, one can conduct several rounds of modal regression analysis assuming different unimodal distributions for the response, such as the FG, DTP-Student-$t$, and  TPSC-Student-$t$  distributions, and then select the most appropriate model using the model selection criteria that we introduce in Section \ref{sec:Model.selection}. All of these models can be easily implemented using the code developed for this work. 

\subsection{{\revise Pictorial depiction of GUD and GUD subfamilies}} 

As illustrated in the preceding three subsections, the GUD family is a \textit{generalization} of several previously proposed unimodal two-component mixture distributions. Figure \ref{fig:GUD_pdfs} presents pdfs of FG, DTP-Student-$t$, and TPSC-Student-$t$ distributions with different parameter specifications, which encompass  asymmetric, symmetric, fat-tailed, \textit{and} thin-tailed densities. In particular, the first panel in Figure \ref{fig:GUD_pdfs} presents the density plot of FG distribution with varying scale parameters of the right-skewed component. As $\sigma_2$ becomes larger, the tails of FG distribution (especially the right tail) become fatter. In the second panel, we show that the FG distribution is symmetric if $w = 0.5$ and $\sigma_1 = \sigma_2$. When the weight parameter $w$ surpasses 0.5, the pdf of FG distribution puts more weight on the left-skewed part, and therefore, becomes more left-skewed. {\revise In the third panel, we see that as the scale parameter of the left-skewed component increases, the left tail of the DTP-Student-$t$ distribution becomes fatter while the right tail changes little, leading to distributions that are more left-skewed.} The fourth panel shows the drastic change in the shape of the TPSC-Student-$t$ pdf as one varies the  scale parameter $\sigma$ shared by both mixture components. The last panel presents the subtle changes in the tail behavior of TPSC-Student-$t$ distributions with different values for the degree of freedom $\delta$ that is shared by both mixture components.
\begin{figure}[p]
	\centering
	\begin{subfigure}[b]{0.80\linewidth}
		\centering
		\includegraphics[width=\linewidth]{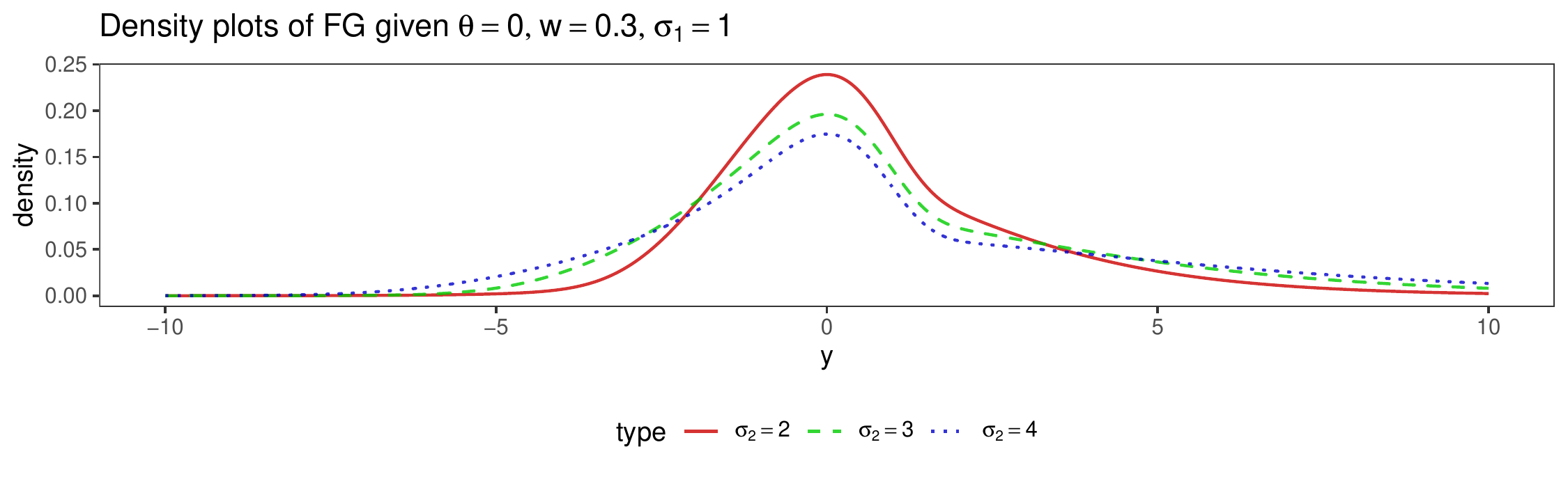}
	\end{subfigure}
	\begin{subfigure}[b]{0.80\linewidth}
		\centering
		\includegraphics[width=\linewidth]{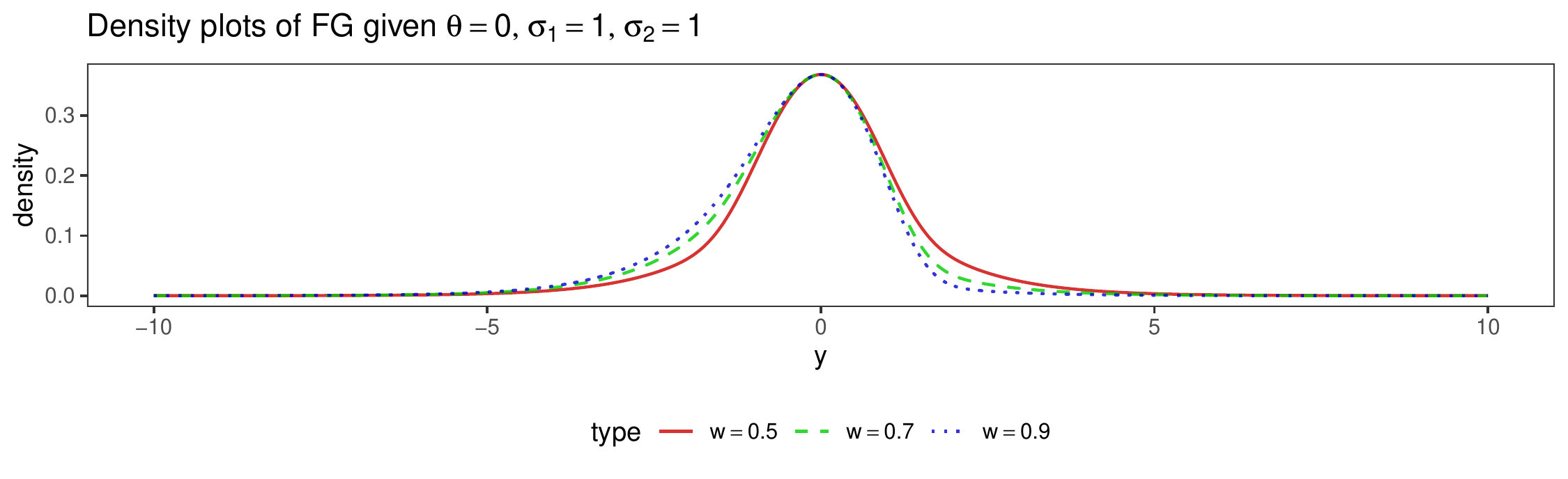}
	\end{subfigure}
	\begin{subfigure}[b]{0.80\linewidth}
		\centering
		\includegraphics[width=\linewidth]{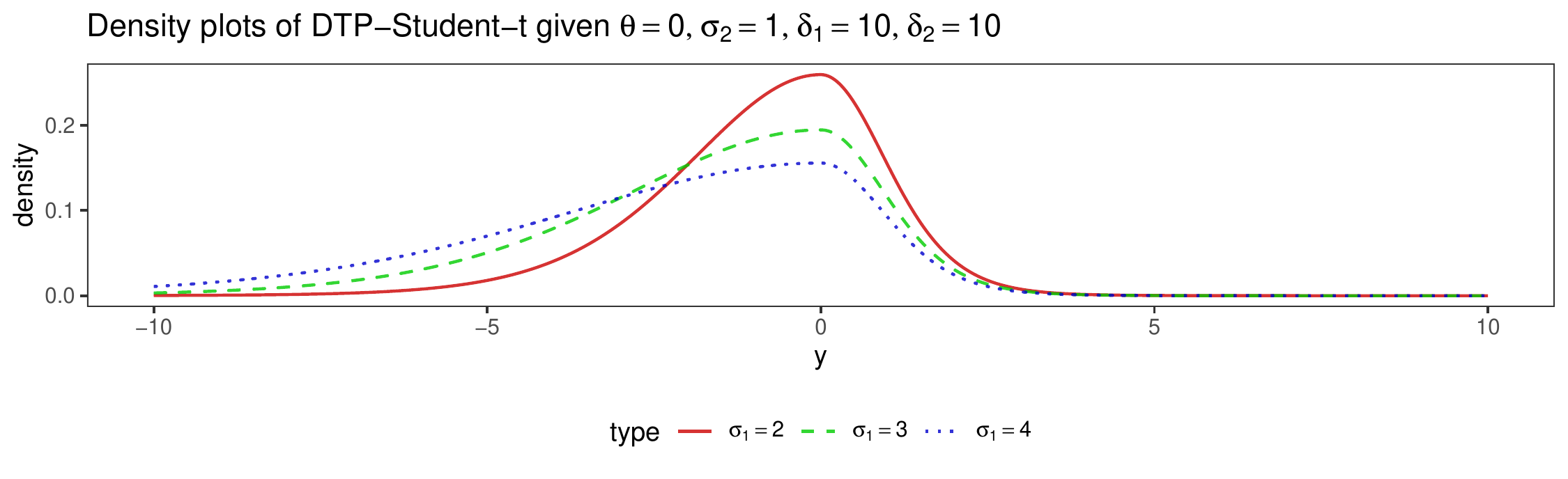}
	\end{subfigure}
	\begin{subfigure}[b]{0.80\linewidth}
		\centering
		\includegraphics[width=\linewidth]{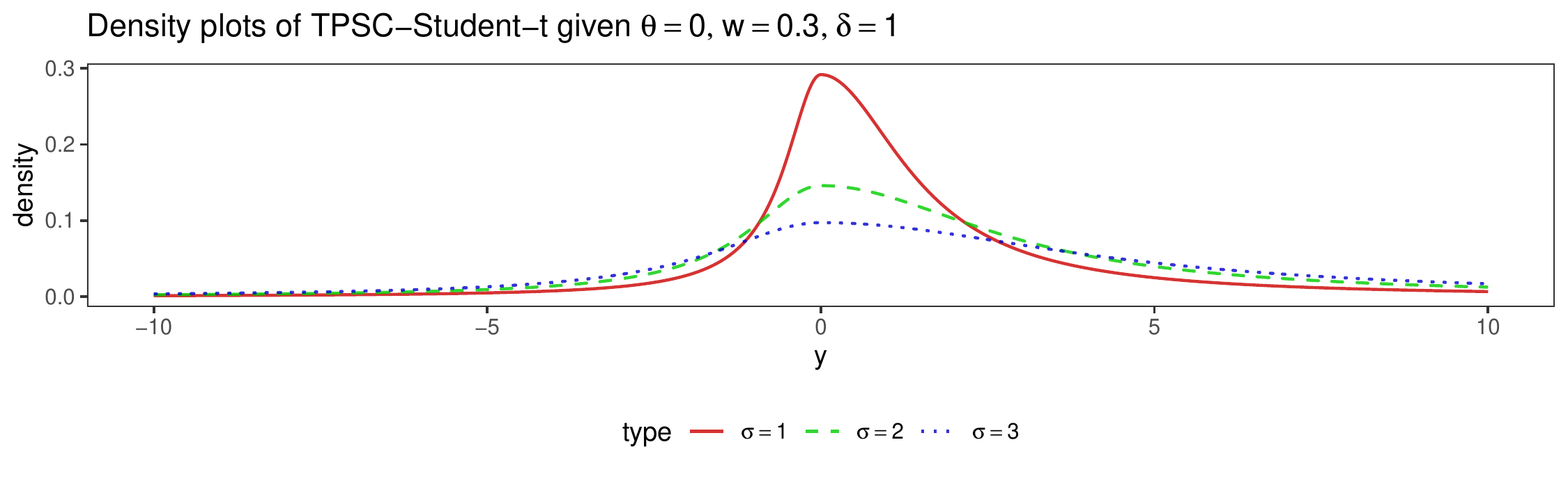}
	\end{subfigure}
	\begin{subfigure}[b]{0.80\linewidth}
		\centering
		\includegraphics[width=\linewidth]{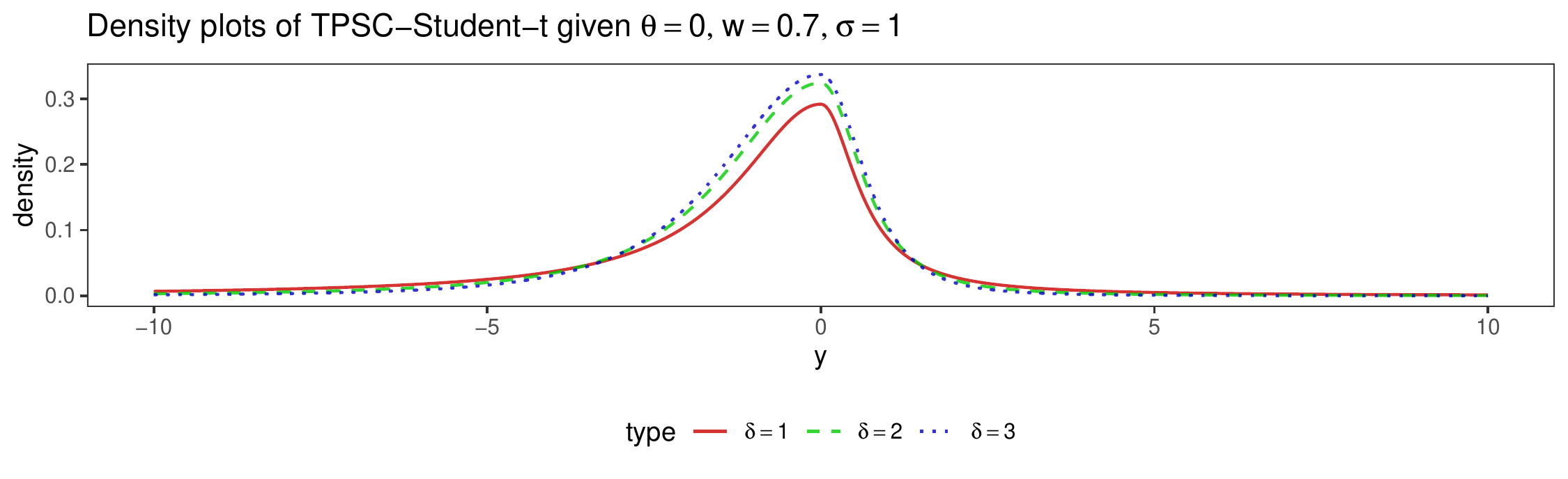}
	\end{subfigure}
	\caption{\small {\it Density plots of different distributions in the GUD family  with different parameter specifications.}}
	\label{fig:GUD_pdfs}
\end{figure}

The GUD family can be further categorized into two subfamilies. {\revisetwo Recall that $\mathcal{D}_1$ and $\mathcal{D}_2$} denote the domains of $f_1(\cdot \mid \theta, \boldsymbol{\xi}_1)$ and $f_2(\cdot \mid \theta, \boldsymbol{\xi}_2)$ in the GUD pdf \eqref{eq:pdf.GUD} respectively. If $\mathcal{D}_1 \cap \mathcal{D}_2 \ne \varnothing$, we call the mixture distribution the \textit{type I GUD}. The FG distribution is an example of type I GUD. In Section \ref{sec:logNM} of the Appendix, we present the construction of the lognormal mixture distribution (logNM), which is another example of type I GUD. On the other hand, if $\mathcal{D}_1 \cap \mathcal{D}_2 = \varnothing$, then we have the \textit{type II GUD}. The DTP distributions and the asymmetric Laplace distribution (ALD) \citep{koenker1999goodness} belong to this subfamily of type II GUD.
{\revise Lastly, all distributions in GUD considered in this section are constructed by mixing two component distributions belonging to a common distribution family or similar families (e.g. the Gumbel family for constructing the FG distribution). This is done in order to more easily satisfy restriction (R3) of the pdf being continuous everywhere. More generally, one could start with two components from different families. However, mixing two components from different families can easily lead to a mixture density that is discontinuous at the mode if not carefully constructed.}


\section{Bayesian modal regression} \label{sec:Bayesian_inference}

Having defined the GUD family in Section \ref{sec:GUD}, we are now in a position to introduce our Bayesian modal regression framework. In the remainder of the manuscript, we assume that we observe $n$ independent pairs of observations $(\boldsymbol{X}_1, Y_1), (\boldsymbol{X}_2, Y_2), \ldots, (\boldsymbol{X}_n, Y_n)$. Here, $\boldsymbol{X}_i := (X_{i1}, \ldots, X_{ip})^\top$ denotes a vector of $p$ covariates for the $i$th observation. We let $\boldsymbol{X} := [ \boldsymbol{X}_1, \ldots, \boldsymbol{X}_n ]^\top$ denote an $n \times p$ design matrix with rows $\boldsymbol{X}_i^\top, i = 1, \ldots, n$. We assume exchangeability in the sense that, given $\boldsymbol{X}$ and all parameters, $n$ observations in $\boldsymbol{Y} :=(Y_1, \ldots, Y_n)$ are independent. Our goal is to conduct inference about the conditional \textit{mode} of the response variable $Y$ given the covariates $\boldsymbol{X}$.

\subsection{Prior elicitation}\label{sec:prior_elicitation}
For all modal linear regression models in this paper, we assume that 
\begin{equation}
Y_i \mid \boldsymbol{X}_i, w, \boldsymbol{\beta}, \boldsymbol{\xi}_1, \boldsymbol{\xi}_2 \stackrel{\text{ind}}{\sim} \operatorname{GUD}\left(w, \, \boldsymbol{X}_{i}^{\top}\boldsymbol{\beta}, \, \boldsymbol{\xi}_1, \, \boldsymbol{\xi}_2\right), \text{ for $i = 1, \dots, n$,}
\label{eq:modal.linear.model}
\end{equation}
where GUD generically refers to a member of the GUD family, and ``ind'' is the acronym for ``independent.'' Recall that any member of the GUD family contains the location parameter as its mode, which is  $\boldsymbol{X}_i^\top \boldsymbol{\beta}$ as the conditional mode for $Y_i$ given $\boldsymbol{X}_i$ in \eqref{eq:modal.linear.model}.

To conduct inference for our model in \eqref{eq:modal.linear.model}, we adopt a Bayesian approach where appropriate priors are placed on the model parameters  $(w, \boldsymbol{\beta}, \boldsymbol{\xi}_1, \boldsymbol{\xi}_2)$.We endow the weight parameter $w$ with a noninformative $\textrm{Uniform}(0,1)$ prior, and use weakly informative inverse gamma priors for all positive parameters in $\boldsymbol{\xi}_1$ and $\boldsymbol{\xi}_2$. As pointed out by \citet{diebolt1994estimation}, improper priors usually lead to improper posterior distributions for mixture distributions because of identifiability problems. Therefore, {\revise if $\boldsymbol{\xi}_1$ and $\boldsymbol{\xi}_2$ do not share any common parameter}, then improper priors should \textit{not} be used for $\boldsymbol{\xi}_1$ or $\boldsymbol{\xi}_2$. 

On the other hand, a flat prior $p(\boldsymbol{\beta}) \propto 1$ on the regression coefficients $\boldsymbol{\beta}$ usually leads to a proper posterior distribution because both right and left skewed components share the same location parameter. In Section \ref{sec:properposterior}, we provide sufficient conditions under which a flat prior can be used for $\boldsymbol{\beta}$ such that the posterior distribution is proper. These sufficient conditions can be shown to hold for a variety of Bayesian modal regression models. All models going forward thus use a noninformative flat prior, $p(\boldsymbol{\beta}) \propto 1$, for $\boldsymbol{\beta}$.

Revisiting the three members of the GUD family discussed in Section~\ref{sec:GUD}, we have the Bayesian modal linear regression model based on the FG likelihood \eqref{eq:pdf.FG}  formulated as follows, 
\begin{equation}
\begin{aligned}
	Y_i \mid \boldsymbol{X}_i, w,\boldsymbol{\beta},\sigma_1,\sigma_2 &\stackrel{\text{ind}}{\sim} \operatorname{FG}\left(w,\, \boldsymbol{X}_i^{\top} \boldsymbol{\beta},\,\sigma_1,\,\sigma_2\right), \text{ for $i = 1, \dots, n$,}\\
	w &\sim \operatorname{Uniform}(0,1), \\
	\sigma_1,\sigma_2 &\stackrel{\text{i.i.d}}{\sim} \text{InverseGamma}(1,1),\\
	p(\boldsymbol{\beta}) &\propto 1,
	\label{eq:FG.model}
\end{aligned}
\end{equation}
where ``i.i.d'' refers to ``independent and identically distributed.'' 
Meanwhile, the Bayesian modal linear regression associated with the DTP-Student-$t$ likelihood \eqref{eq:pdf.DTP} is specified by 
\begin{equation}
\label{eq:DTP.model}
\begin{aligned}
	Y_i \mid \boldsymbol{X}_i, \boldsymbol{\beta}, \sigma_1, \sigma_2, \delta_1, \delta_2 & \stackrel{\text{ind}}{\sim} \operatorname{DTP-Student-}t\left(\boldsymbol{X}_i^{\top} \boldsymbol{\beta}, \, \sigma_1,\,\sigma_2, \,\delta_1,\,\delta_2 \right), \text{ for $i = 1, \dots, n$,}\\
	\sigma_1,\sigma_2,\delta_1,\delta_2 &\stackrel{\text{i.i.d}}\sim \text{InverseGamma}(1,1), \\
	p(\boldsymbol{\beta}) &\propto 1.
\end{aligned}
\end{equation}
Recall that the weight parameter $w$ of a DTP distribution is fully defined by its scale and shape parameters, so in this case, there is no need to choose a prior for $w$. Finally, the Bayesian modal linear regression associated with the TPSC-Student-$t$ likelihood \eqref{eq:pdf.TPSC} is defined as
\begin{equation}
\begin{aligned}
	Y_i \mid \boldsymbol{X}_i, w, \boldsymbol{\beta}, \sigma, \delta & \stackrel{\text{ind}}{\sim} \operatorname{TPSC-Student-}t\left(w,\, \boldsymbol{X}_i^{\top} \boldsymbol{\beta}, \, \sigma, \, \delta \right), \text{ for $i = 1, \dots, n$,}\\
	w &\sim \operatorname{Uniform}(0,1), \\
	\sigma,\delta &\stackrel{\text{i.i.d}}\sim \text{InverseGamma}(1,1), \\
	p(\boldsymbol{\beta}) &\propto 1.
	\label{eq:TPSC.model}
\end{aligned}
\end{equation}
According to Proposition \ref{thm:suff2} in next subsection, all of the proposed Bayesian modal regression models \eqref{eq:FG.model}--\eqref{eq:TPSC.model} above have proper posterior distributions. Practitioners can construct various other Bayesian modal regression models using the same strategy shown above. In this paper, we concentrate on the modal regression models based on the FG, DTP-Student-$t$, and TPSC-Student-$t$ likelihoods for the sake of concreteness.

\subsection{Sufficient conditions for posterior propriety}\label{sec:properposterior}

Since we use an improper prior, $p(\boldsymbol{\beta}) \propto 1$, for the regression coefficients $\boldsymbol{\beta}$ in our Bayesian modal regression models, it is important to check that the posterior distribution is proper. Theorem \ref{thm:suff1} gives sufficient conditions under which the GUD likelihood \eqref{eq:pdf.GUD} with a flat prior on the mode/location parameter and suitably chosen priors on other model parameters lead to a proper posterior. Theorem \ref{thm:suff2} extends this result to the regression setting. Proofs for the theorems and propositions in this section can be found in Section \ref{sec:proofs} of the Appendix. We stress that our results are \emph{nonasymptotic}; that is, our results apply for any \emph{fixed} sample size $n$.

To ease the notation, let $f_{Z}\left(y\mid w,\boldsymbol{\xi}_{1},\boldsymbol{\xi}_{2}\right) := f\left(y \mid w, \theta=0, \boldsymbol{\xi}_1, \boldsymbol{\xi}_2\right)$ be the pdf of GUD family with the mode at 0. We can rewrite the pdf \eqref{eq:pdf.GUD} as 
$f_Z(y-\theta \mid w, \boldsymbol{\xi}_1, \boldsymbol{\xi}_2)=f\left(y \mid w, \theta, \boldsymbol{\xi}_1, \boldsymbol{\xi}_2\right)$.

\begin{theorem}\label{thm:suff1}
Let $\Theta_{w, \boldsymbol{\xi}_1, \boldsymbol{\xi}_2}$ denote the parameter space of $w, \boldsymbol{\xi}_1$ and $\boldsymbol{\xi}_2$, with respective independent priors $p(w)$, $p(\boldsymbol{\xi}_1)$, and $p(\boldsymbol{\xi}_2)$. For any $n \geq 1$, if
$$
\iiint_{\Theta_{w, \boldsymbol{\xi}_1, \boldsymbol{\xi}_2}} \left\{f_Z\left(0 \mid w, \boldsymbol{\xi}_1, \boldsymbol{\xi}_2\right)\right\}^{n-1} p(w) p\left(\boldsymbol{\xi}_1\right) p\left(\boldsymbol{\xi}_2\right) d w d \boldsymbol{\xi}_1 d \boldsymbol{\xi}_2<\infty,
$$
then the posterior distribution $p\left(w,\theta,\boldsymbol{\xi}_{1},\boldsymbol{\xi}_{2} \mid Y_1,\dots,Y_n\right)$ is proper under a flat prior $p(\theta) \propto 1$.
\end{theorem}

Theorem \ref{thm:suff1} applies to the case where there is a single location parameter $\theta$ (as in the bank deposits application in Section \ref{sec:bank}). Next, we extend this result to the regression setting. Theorem \ref{thm:suff2} enables us to use the noninformative flat prior $p(\boldsymbol{\beta}) \propto 1$ for the regression coefficients $\boldsymbol{\beta}$ in Bayesian modal regression based on the GUD likelihood \eqref{eq:pdf.GUD}.

\begin{theorem}\label{thm:suff2}
Let $\boldsymbol{X}$ be a full rank design matrix with $p \le n$ and finite entries. If 
$$
\iiint_{\Theta_{w, \boldsymbol{\xi}_1, \boldsymbol{\xi}_2}} \left\{f_Z\left(0 \mid w, \boldsymbol{\xi}_1, \boldsymbol{\xi}_2\right)\right\}^{n-p}\ p(w) p\left(\boldsymbol{\xi}_1\right) p\left(\boldsymbol{\xi}_2\right) d w d \boldsymbol{\xi}_1 d \boldsymbol{\xi}_2<\infty,
$$
then the posterior distribution $p\left(w, \boldsymbol{\beta}, \boldsymbol{\xi}_1, \boldsymbol{\xi}_2 \mid \boldsymbol{X}, \boldsymbol{Y}\right)$ is proper under a flat prior $p(\boldsymbol{\beta}) \propto 1$.
\end{theorem}
The sufficient conditions in Theorems \ref{thm:suff1} and \ref{thm:suff2} may seem abstract, and checking such conditions amounts to testing convergence of multiple integrals. The intuition behind these theorems is that, if the GUD likelihood with a mode of zero has a proper posterior distribution under suitably chosen priors on the scale/shape parameters, then the use of a flat prior $p(\boldsymbol{\beta}) \propto 1$ is acceptable. 

\begin{proposition}\label{thm:proper.posterior.reg}
Suppose that $\boldsymbol{X}$ is full rank with $p \leq n$ and finite entries. Then the Bayesian modal regression models \eqref{eq:FG.model}, \eqref{eq:DTP.model}, and \eqref{eq:TPSC.model} based on the FG, DTP-Student-$t$, and TPSC-Student-$t$ likelihoods, respectively, have proper posterior distributions.
\end{proposition}
Proposition \ref{thm:proper.posterior.reg} confirms that under suitable regularity conditions on the design matrix $\boldsymbol{X}$, all of the regression models proposed in this paper have proper posterior distributions. The proof of Proposition \ref{thm:proper.posterior.reg} relies on verifying the sufficient condition given in Theorem \ref{thm:suff2}. Our proof provides a template for verifying posterior propriety for other Bayesian modal regression models \eqref{eq:modal.linear.model} under the general GUD family.

\citet{diebolt1994estimation} have argued that improper priors should in general not be used for Bayesian modeling of mixture distributions. We note, however, that the reasoning of \citet{diebolt1994estimation} does not necessarily apply to the \textit{location} parameter $\theta$ (or the mode). This is because the mode $\theta$ is shared by \textit{both} left- and right-skewed components in our proposed GUD family of distributions. Therefore, we are able to derive  sufficient conditions under which a totally noninformative flat prior $p(\theta) \propto 1$ or $p(\boldsymbol{\beta}) \propto 1$ can still be used to infer the conditional \textit{mode}. 

On the other hand, we recommend against using improper priors for any of the \textit{non}-location parameters (i.e. the shape/scale parameters) in Bayesian modal regression based on the GUD family. We formalize this in Proposition \ref{thm:improper} below. This proposition states that, for the GUD family, using an improper prior for \emph{any} shape/scale parameter that is not shared by both components leads to an \textit{improper} posterior distribution. 
\begin{proposition}\label{thm:improper}
If $\tau \in \left(\boldsymbol{\xi}_1 \cup \boldsymbol{\xi}_2\right) \backslash\left(\boldsymbol{\xi}_1 \cap \boldsymbol{\xi}_2\right)$, then using an improper prior for $\tau$ will lead to an improper posterior distribution.
\end{proposition}

In Section \ref{sec:logNM} of the Appendix, we provide a specific example of Proposition \ref{thm:improper} for the logNM distribution (also introduced in the same section). 

\subsection{Uncertainty quantification and model selection} \label{sec:Model.selection} 
{\revisetwo
Let $\boldsymbol{\Omega} = \left(w, \boldsymbol{\beta}, \boldsymbol{\xi}_1, \boldsymbol{\xi}_2\right)$ represent the collection of model parameters. We are interested in the distribution of a new response $Y_{\text{new}}$ given new covariates $\boldsymbol{X}_{\text{new}}$, and observed data $(\boldsymbol{Y}, \boldsymbol{X})$.

The posterior predictive distribution is defined as
$$
p\left(Y_{\text{new}} \mid \boldsymbol{Y}, \boldsymbol{X}, \boldsymbol{X}_{\text{new}}\right) = \int_{\Theta} p\left(Y_{\text{new}} \mid \boldsymbol{\Omega}, \boldsymbol{Y}, \boldsymbol{X}, \boldsymbol{X}_{\text{new}}\right) p\left(\boldsymbol{\Omega} \mid \boldsymbol{Y}, \boldsymbol{X}, \boldsymbol{X}_{\text{new}}\right) d \boldsymbol{\Omega},
$$
where $\Theta$ denotes the parameter space. Because the unobserved data is conditionally independent of the observed data $(\boldsymbol{Y}, \boldsymbol{X})$ given the parameters $\boldsymbol{\Omega}$, we have $
p(Y_{\text{new}} \mid \boldsymbol{\Omega}, \boldsymbol{Y}, \boldsymbol{X}, \boldsymbol{X}_{\text{new}}) = p(Y_{\text{new}} \mid \boldsymbol{\Omega}, \boldsymbol{X}_{\text{new}})$.
Additionally, because the new set of covariates $\boldsymbol{X}_{\text{new}}$ is independent of the posterior distribution of the parameters $\boldsymbol{\Omega}$, 
$
p(\boldsymbol{\Omega} \mid \boldsymbol{Y}, \boldsymbol{X}, \boldsymbol{X}_{\text{new}}) = p(\boldsymbol{\Omega} \mid \boldsymbol{Y}, \boldsymbol{X})$.
Therefore, the posterior predictive distribution reduces to
\begin{align}
	p\left(Y_{\text{new}} \mid \boldsymbol{Y}, \boldsymbol{X}, \boldsymbol{X}_{\text{new}}\right) & = \int_{\Theta} p\left(Y_{\text{new}} \mid \boldsymbol{\Omega}, \boldsymbol{X}_{\text{new}}\right) p\left(\boldsymbol{\Omega} \mid \boldsymbol{Y}, \boldsymbol{X}\right) d \boldsymbol{\Omega}.\label{eq:ppd}
	\end{align}}
	Obtaining an approximation of the posterior predictive distribution specified by \eqref{eq:ppd} is computationally inexpensive. With the {\revise sampling method} outlined in Section \ref{sec:GUD} for the GUD family, one can easily draw samples from {\revisetwo $p(Y_{\text{new}}\mid \boldsymbol{\Omega}, \boldsymbol{X}_{\operatorname{new}})$} during each iteration in our MCMC algorithm, and then obtain samples from the posterior predictive distribution {\revisetwo $p\left(Y_{\operatorname{new}} \mid \boldsymbol{Y}, \boldsymbol{X}, \boldsymbol{X}_{\operatorname{new}} \right)$}. 
	In this paper, we use the \texttt{hdi} function in the \textsf{R} package \texttt{HDInterval} \citep{Rmanual,meredith2018package}, whose inputs are random samples generated from the posterior predictive distributions,  to calculate the highest probability density (HPD) intervals. We use 90\% HPD prediction intervals as the posterior prediction intervals for all mean/median/modal regression models that we consider in Sections \ref{sec:Simulation} and \ref{sec:data_application}.
	
	Due to the inherent nature of the conditional mode, the HPD intervals from modal regression models will usually be \textit{narrower} than those constructed under mean or median regression models, while having the \textit{same} amount of coverage \citep{yao2014new}. 
	Prediction intervals from mean or median regression can sometimes be very conservative and contain many implausible values. We illustrate the benefits of more efficient inference from modal regression in Sections \ref{sec:Simulation} and \ref{sec:data_application}. 
	
	As mentioned in Section \ref{sec:GUD}, there are many different GUD likelihoods that a practitioner can choose from in order to conduct Bayesian inference for modal regression. We propose to use the Bayesian leave-one-out expected log posterior density as a model selection criterion for selecting the ``best'' GUD likelihood to use. The Bayesian leave-one-out expected log predictive density is defined as
	\begin{equation}
\text{ELPD}=\sum_{i=1}^n \log p\left(Y_i \mid Y_{-i}\right),
\label{eq:loo}
\end{equation}
where $Y_{-i}$ represents all observations except the $i$-th observation. In \eqref{eq:loo}, ``ELPD'' stands for the theoretical expected log predictive density. Intuitively, if a model fits the data well, its predicted value of $Y_i$ given $Y_{-i}$ should be close to the observed $Y_i$ and $p\left(Y_i \mid Y_{-i}\right)$ should be large, for all $i = 1, \dots, n$. Therefore, an adequate model tends to yield a high ELPD.

We apply the Pareto-smoothed importance sampling method (PSIS) of \citet{vehtari2017practical} to obtain an estimate of ELPD. 
The PSIS estimation of ELPD has been implemented in the \textsf{R} package \texttt{loo}, which is compatible with the \texttt{Stan} programming language \citep{carpenter2017stan}. When fitting multiple competing models to the same dataset, the model with the highest estimated ELPD is preferred. By a slight abuse of notation, we use ELPD to refer to the estimated ELPD in all empirical study presented in this paper. 

Other model selection criteria, such as the deviance information criterion (DIC) proposed by \citet{spiegelhalter2002bayesian} and the widely applicable information criterion (WAIC) introduced by \citet{watanabe2010asymptotic}, are also applicable to regression models with GUD likelihoods. In fact, DIC and WAIC have been shown to be asymptotically equal to ELPD \citep{gelman2013bayesian}. However, \citet{vehtari2017practical} recommended ELPD and WAIC over DIC because DIC is not a fully Bayesian information criterion and is based on a point estimate. Additionally, \citet{vehtari2017practical} demonstrated that ELPD is more robust than WAIC in finite samples with weak priors or influential observations. Therefore, we decided to use ELPD for all data applications and numerical studies in this paper.

\subsection{{\revise Implementation}}

We utilized the \texttt{Stan} programming language interfaced with \textsf{R}  \citep{carpenter2017stan} to implement all data analyses in the empirical study presented in this article. {\revise  \texttt{Stan} uses Hamiltonian Monte Carlo \citep{neal2011mcmc} and leverages the No-U-Turn sampler (NUTS) proposed by \citet{hoffman2014no}. The implementation of Bayesian linear modal regression in \eqref{eq:modal.linear.model} in \texttt{Stan} involves defining data log-likelihood functions first, followed by specifying priors as outlined in \eqref{eq:FG.model} through \eqref{eq:TPSC.model}. Computer programs for reproducing all numerical results in our study are available at the following link:  \href{https://github.com/rh8liuqy/Bayesian_modal_regression}{https://github.com/rh8liuqy/Bayesian\_modal\_regression}.}

Even though it may seem plausible to apply the sampling method as outlined in Section~\ref{sec:GUD} where one introduces a latent variable to separate the inference for the left-skewed and right-skewed parts, dealing with type II GUD distributions introduces challenges. Introducing such a latent variable in a data augmentation MCMC algorithm might result in a degenerate random variable, as demonstrated in Section \ref{sec:note_on_MCMC} of the Appendix. {\revise In such cases, the data augmented MCMC sampler is \emph{reducible}. {\revisetwo A Markov chain is reducible if it is impossible to eventually get from one state to any other states in a finite number of steps. In our context, the data augmentation algorithm is reducible since a poor initialization leads to the latent variable always taking the same value, regardless of the number of MCMC iterations. A reducibile MCMC sampler is less practical, since it requires ``acceptable'' initial conditions in order to be able to explore the entire parameter space \citep{Robert2004}. Acceptable initial conditions are typically unknown.}

{\revisetwo In contrast to many widely used data augmentation algorithms}, NUTS does not require the introduction of discrete latent variables to fit Bayesian mixture models. Instead, NUTS simply requires one to be able to calculate the gradient of the log-likelihood with respect to the model parameters, and this can be done using automatic differentiation. We observed satisfactory performance of NUTS in the empirical study, as evidenced by promising convergence results for all data applications and simulations studies (see the convergence diagnostics presented in Section \ref{Sec:Convergence_Diagnostics}) of the Appendix.

\section{\revise Simulation studies} \label{sec:Simulation}

\subsection{Left-skewed data}\label{sec:left_skewed_data}

We now present one simulation study which show that our Bayesian modal regression model is an excellent choice for modeling data that is heavily skewed. Under our simulation settings, simulated data was left-skewed; and, in addition to the pronounced global conditional mode, there was also a small local mode. We compared our Bayesian modal regression models to {\revise classic/robust} Bayesian mean and median regression models. The classic Bayesian mean regression used a normal likelihood, i.e., 
\begin{equation}
	\begin{aligned}
		Y_i \mid \boldsymbol{\beta},\sigma,\boldsymbol{X}_i & \stackrel{\text{ind}}{\sim} \mathcal{N}(\boldsymbol{X}_i^{\top} \boldsymbol{\beta},\,\sigma^2), \text{ for $i = 1,\dots, n$,}\\
		\sigma &\sim \text{InverseGamma}(1,1),\\
		p(\boldsymbol{\beta}) &\propto 1.
	\end{aligned}
	\label{eq:model.normal}
\end{equation}
{\revise The robust Bayesian mean regression model used the SNCP likelihood \citep{arellano2008centred}, i.e.,
	\begin{equation}
		\begin{aligned}
			Y_i \mid \boldsymbol{\beta},\sigma,\boldsymbol{X}_i & \stackrel{\text{ind}}{\sim} \operatorname{SNCP}(\boldsymbol{X}_i^{\top} \boldsymbol{\beta},\,\sigma^2,\gamma_{1}), \text{ for $i = 1,\dots, n$,}\\
			\sigma &\sim \text{InverseGamma}(1,1),\\
			\gamma_{1} &\sim \text{Uniform}(-1,1),\\
			p(\boldsymbol{\beta}) &\propto 1.
		\end{aligned}
		\label{eq:model.SNCP}
\end{equation}}
In line with the literature on parametric Bayesian quantile regression \citep{yu2001bayesian,yu2005three}, we also implemented Bayesian median regression using the asymmetric Laplace distribution (ALD), with quantile parameter $p=0.5$. That is, our Bayesian median regression model was
\begin{equation}
	\begin{aligned}
		Y_i \mid \boldsymbol{\beta},\sigma,\boldsymbol{X}_i & \stackrel{\text{ind}}{\sim} \operatorname{ALD}(\boldsymbol{X}_i^{\top} \boldsymbol{\beta},\,\sigma,\,p = 0.5), \text{ for $i = 1,\dots, n$,} \\
		\sigma &\sim \text{InverseGamma}(1,1),\\
		p(\boldsymbol{\beta}) &\propto 1.
	\end{aligned}
	\label{eq:model.ALD}
\end{equation}
We stress that in this simulation study, \textit{none} of the likelihoods used for mean, median, or modal regression was exactly the same as the data generating mechanism. Therefore, all considered regression models are ``wrong,'' creating particularly realistic yet  challenging scenarios under which we could more fairly compare the performance across these competing methods. 

{\revisetwo With two different sample sizes $n = 30$ and $n=300$}, we generated observations from the model,
$$Y_i = \beta_0 + \beta_1 X_i + \epsilon_i,$$
where $\beta_0=\beta_1=1$ and, for $i=1, \ldots, n$, $\epsilon_i \overset{\text{i.i.d}}{\sim} 0.05\mathcal{N}(-50,1^2)+0.95\mathcal{N}(0,1^2)$ and  $X_i\overset{\text{i.i.d}}{\sim}\text{Uniform}(0,1)$. We then fit the mean/median/modal regression models to the simulated data. For modal regression, we fit the FG model \eqref{eq:FG.model}, the DTP-Student-$t$ model \eqref{eq:DTP.model}, and the TPSC-Student-$t$ model \eqref{eq:TPSC.model}. Among the modal regression models, we found that the  TPSC-Student-$t$ model had the highest ELPD. For the sake of brevity, we present only the results from the models fit with the normal, {\revise SNCP}, ALD, and TPSC-Student-$t$ likelihoods.

In Figure \ref{fig:left_skewed_data}, we provide the empirical coverage rate and the average width of the posterior prediction intervals across $n=30$ observations under each of mean/median/modal regression model. With the narrowest prediction interval for the same amount of coverage, results from the modal regression model clearly stand out in Figure \ref{fig:left_skewed_data}. In addition, the modal regression model with the TPSC-Student-$t$ likelihood had the largest ELPD. Therefore, it was the most appropriate model for the simulated data among the three candidate models in this replicate.

{\revise Figure \ref{fig:left_skewed_data_residual} depicts the true mode-zero model error distribution contrasted with the four estimated mode-zero error distributions based on the regression analyses considered in Figure~\ref{fig:left_skewed_data}, where we used posterior means as point estimates for all parameters. For mean regression analysis, the estimated normal distribution (with $\hat{\sigma} = 12.7$) and the estimated SNCP distribution (with $\hat{\sigma} = 13.3$ and $\hat{\gamma}_{1} = 0.08$) both failed to capture the shape of the true distribution and indicated much greater variability around the mode than there truly was. The estimated error distribution from the median regression model (with $\hat{\sigma} = 2.13$ in the ALD likelihood) was much improved over the former two estimated densities. However, the estimation continued to improve notably when the TPSC-Student-$t$ distribution was assumed in the proposed modal regression (with $\hat{w} = 0.55, \hat{\sigma} = 0.80, \hat{\delta} = 1.08$). In particular, the height of this last estimated density at the mode was the closest to that of the true error distribution. Compared with the three competitors, the tail behavior of the estimated TPSC-Student-$t$ density was also strikingly similar to that of the ground truth.} 

{\revisetwo We repeated this experiment comparing the four regression models for 300 times, with sample sizes of $n=30$ and $n=300$ in each replication. Table \ref{tab:simu_left_skewed} provides Monte Carlo averages of the coverage rate of the 90\% HPD prediction interval, width of the prediction interval, and ELPD for each regression model. When the sample size is as small as 30, the two mean regression models yielded lower coverage rates yet wider prediction intervals on average than those from the median and modal regression models. Even though the latter two regression models enjoyed similar (higher) coverage rates, the modal regression model tended to give much tighter prediction intervals. 
	
	Note that on average 95\% of the data generated in this simulation study were non-outliers. Coverage rates close to 95\% thus imply that these models have adequate prediction intervals as they can predict the values of non-outliers reasonably well. Within the methods with similar coverage rates, we prefer the modal regression model with the narrowest prediction interval since these predictions have the least amount of uncertainty. Furthermore, the modal regression model with the TPSC-Student-$t$ likelihood had the largest average Expected Log Predictive Density (ELPD), reinforcing that the modal regression model based on the TPSC-Student-$t$ provided the best overall model fit.
	
	After we increased the sample size from $n = 30$ to $n = 300$, all regression methods had around 95\% coverage rate. However, the modal regression model still had the narrowest prediction interval and highest ELPD, on average. This again confirms the supremacy of the modal regression model in this simulation study.}

\begin{figure}
	\centering
	\includegraphics[width=1.0\linewidth]{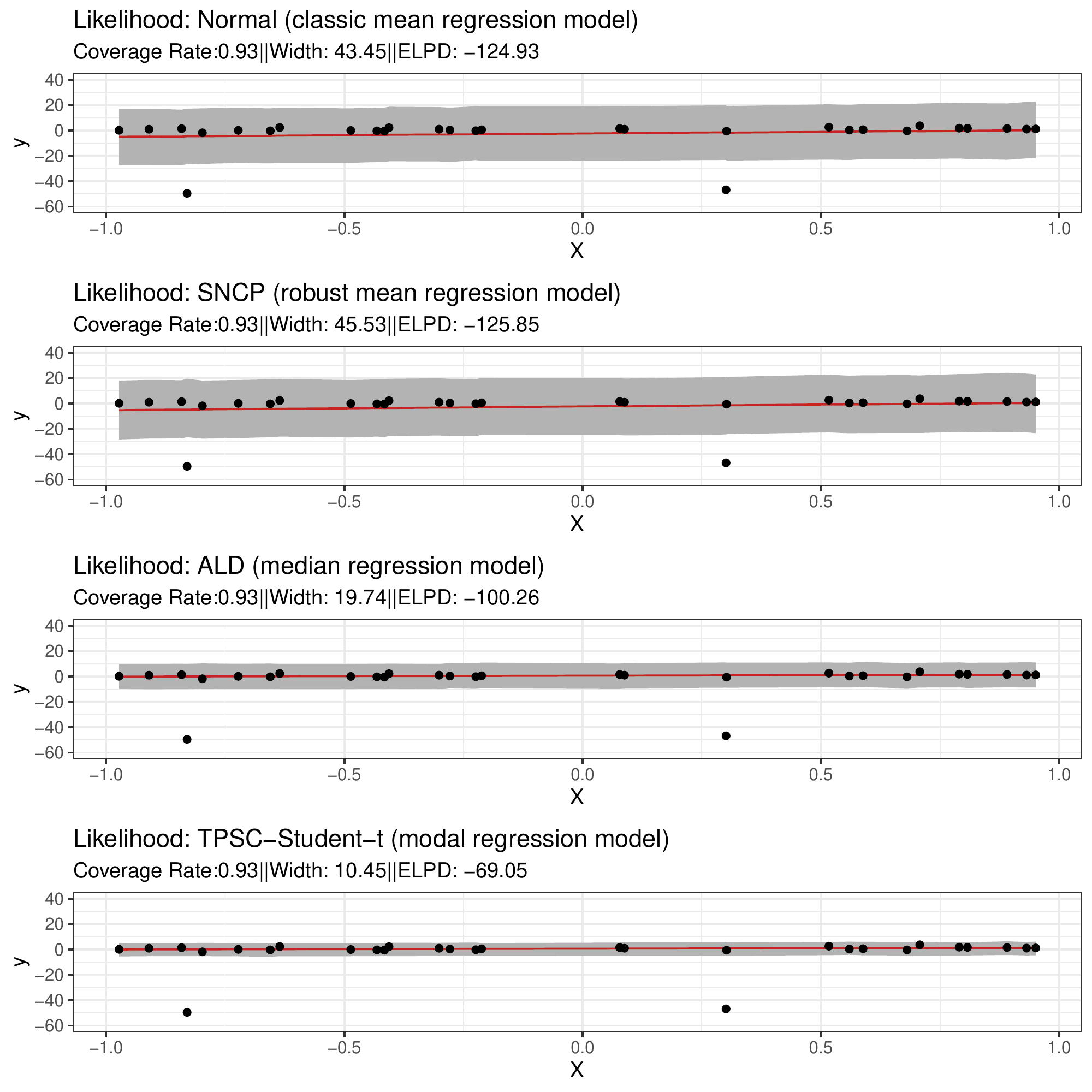}
	\caption{\small {\it The gray shaded areas show the $90\%$ posterior prediction intervals for the simulated left-skewed data. The solid red line is the estimated median from the posterior predictive distribution. The prediction intervals are narrower for Bayesian modal regression.}}
	\label{fig:left_skewed_data}
\end{figure}
\begin{figure}[t]
	\centering
	\includegraphics[width=\linewidth]{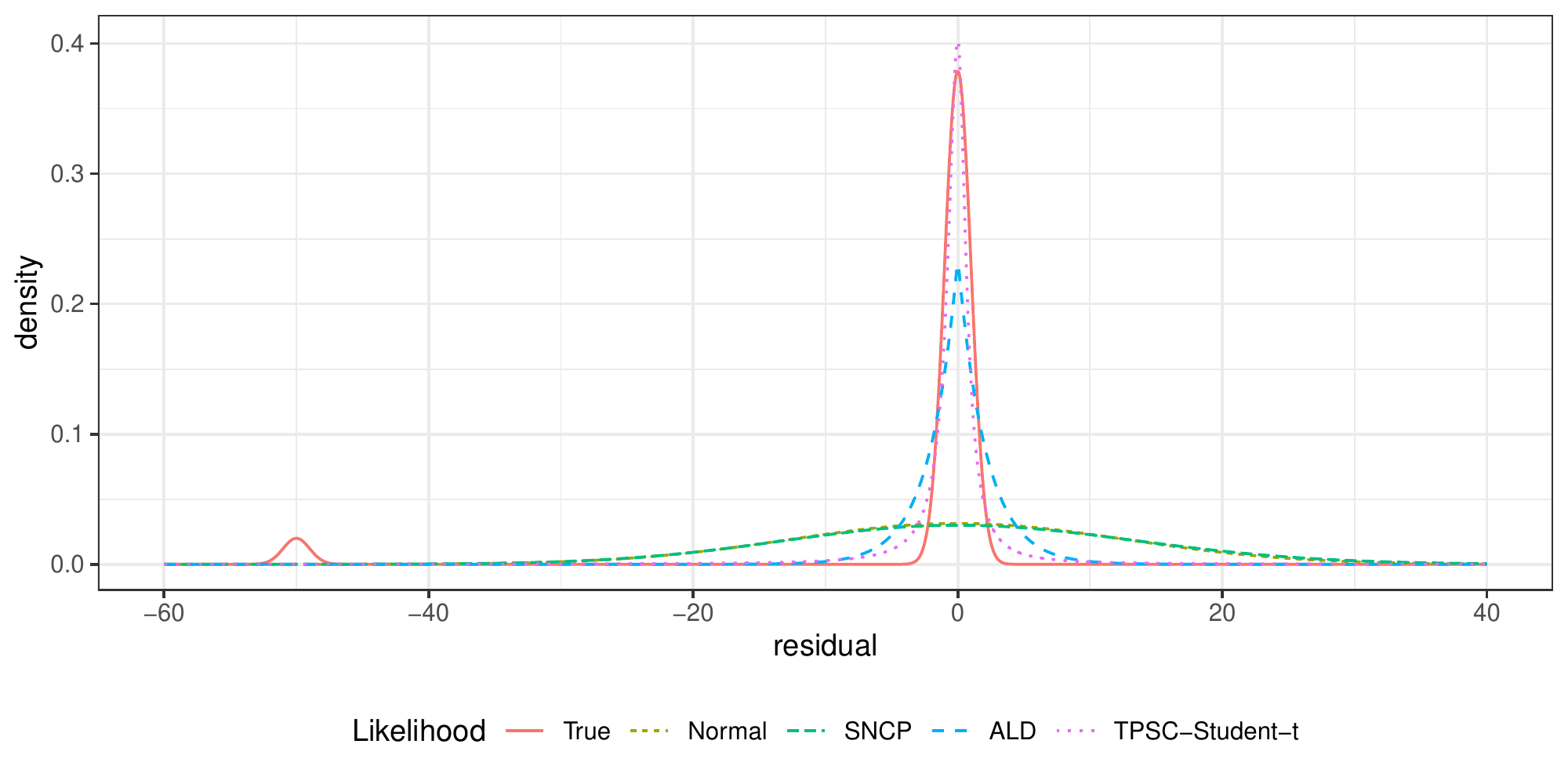}
	\caption{\small {\it The estimated density plots of residuals with fixed mode at $0$ from regression models utilizing normal, SNCP, ALD and TPSC-Student-$t$ likelihoods.}}
	\label{fig:left_skewed_data_residual}
\end{figure}
\begin{table}[t]
	\caption{\label{tab:simu_left_skewed}Comparison of Bayesian mean, median, and modal regression models fitted to left-skewed data. Results were averaged across 300 Monte-Carlo replicates of left-skewed datasets. The empirical standard error associated with each Monte-Carlo average is provided in parenthesis following the average. {\revisetwo The 90\% HPD intervals were used to calculate the coverage rate.} }
	\centering
	\resizebox{\columnwidth}{!}{
		\begin{tabular}[t]{llrrr}
			\toprule
			Sample Size & Likelihood (regression model) &  Coverage Rate (\%) & Width & ELPD\\
			\midrule
			& normal & 93.47(0.20) & 32.26(1.08) & -104.70(2.00)\\
			\cmidrule{2-5}
			& SNCP & 93.56(0.19) & 33.78(1.13) & -105.60(2.02)\\
			\cmidrule{2-5}
			& ALD & 94.11(0.21) & 15.47(0.56) & -85.11(1.40)\\
			\cmidrule{2-5}
			\multirow{-4}{*}{\raggedright\arraybackslash n = 30} & TPSC-Student-$t$ & \textbf{94.69(0.22)} & \textbf{8.31(0.21)} & \textbf{-59.96(0.64)}\\
			\cmidrule{1-5}
			& normal & 95.01(0.07) & 35.67(0.26) & -1142.65(3.05)\\
			\cmidrule{2-5}
			& SNCP & 95.01(0.07) & 35.85(0.27) & -1144.74(3.02)\\
			\cmidrule{2-5}
			& ALD & 95.03(0.07) & 14.94(0.16) & -861.39(3.45)\\
			\cmidrule{2-5}
			\multirow{-4}{*}{\raggedright\arraybackslash n = 300} & TPSC-Student-$t$ & \textbf{94.64(0.06)} & \textbf{6.69(0.06)} & \textbf{-591.81(1.87)}\\
			\bottomrule
		\end{tabular}
	}
\end{table}

\subsection{Right-skewed data}\label{sec:right_skewed_data}

In Section \ref{sec:left_skewed_data}, we demonstrated the advantages of our Bayesian modal regression models {\revise over Bayesian mean and median regression} when the data was left-skewed. {\revise In this section, we compare the performance of our model on right-skewed data against an alternative modal linear regression (MODLR) method proposed by \citet{yao2014new}. Whereas our Bayesian modal regression models are fully parametric, the mode-zero model error distribution under MODLR is left unspecified. However, just like our models, MODLR assumes a linear function of covariates for the conditional mode of $Y$, i.e. $\text{Mode}(Y \mid \boldsymbol{X} ) = \boldsymbol{X}^\top \boldsymbol{\beta}$.}

We generated $n = 30$ observations from the model,
$$
Y_{i} = \beta_{0} + \beta_{1} X_{1,i} + \beta_{2} X_{2,i} + \epsilon_{i},
$$
where $\beta_{0} = \beta_{1} = \beta_{2} = 1$, $X_{1,i}$ and $X_{2,i}$ come from the standard normal distribution independently, and $\epsilon_i$ follows SNDP$(-0.3754,1,5)$, which is an SNDP distribution with the location parameter as $-0.3754$, the scale parameter as $1$, and the skewness parameter as $5$ \citep{azzalini2013skew}. This distribution is right-skewed and with mode at $0$. We fit a Bayesian modal regression that assumes the TPSC-Student-$t$ distribution of $Y$ given $X_1$ and $X_2$ to the simulated data. 

{\revise We also implemented MODLR to infer covariate effects. MODLR is based on kernel density estimation and requires bandwidth selection. To ensure a (more than) fair comparison, we selected the bandwidth that minimizes the average sum-of-square bias, $\sum_{j=0}^{2} (\hat{\beta}_{j} - \beta_{j})^{2}$, across 300 repeated experiments. This bandwidth selection procedure is biased towards MODLR as it uses the true values of regression coefficients that are unknown in practice. After the optimal bandwidth was chosen, we used MODLR to estimate regression coefficients and constructed confidence intervals using the nonparametric bootstrap with 500 bootstrap samples. We repeated this process for 300 total experiments and compared the results to our Bayesian modal regression model. As in Section~\ref{sec:murder}, drawing inference for regression coefficients based on our parametric modal regression model was more straightforward, with the posterior means serving as the point estimates and the corresponding credible intervals as interval estimates.}

{\revise Table \ref{tab:simu_right_skewed} presents summary statistics for inferences of the regression coefficents under the two considered methods averaged across 300 Monte Carlo replicates. Both Bayesian modal regression based on TPSC-Student-$t$ and MODLR produced point estimates for the covariate effects $\beta_1$ and $\beta_2$ close to their true values, with the TPSC-Student-$t$ being more precise than MODLR. Both methods appeared to overestimate the intercept, possibly more so when MODLR was used. Table \ref{tab:simu_right_skewed} shows that the interval estimates from the proposed Bayesian modal regression model were clearly more reliable than those from MODLR, especially those for $\beta_1$ and $\beta_2$ where our Bayesian model yielded tighter interval estimates with similar empirical coverage rates. In this example, Bayesian modal regression based on TPSC-Student-$t$ exhibited superior performance over MODLR in terms of inference for regression coefficients. This is despite the fact that the likelihood in our parametric Bayesian model is misspecified and despite the fact that MODLR is ostensibly more flexible by allowing the residual error distribution to be unspecified.}

\begin{table}[t]
	\caption{\label{tab:simu_right_skewed} \small {\it Comparison of the Bayesian modal regression based on the TPSC-Student-$t$ and MODLR. Results were averaged across 300 Monte-Carlo replicates of right-skewed datasets. The empirical standard error associated with each Monte-Carlo average is provided in parenthesis following the average. {\revisetwo CI denotes 90\% credible interval}.}}
	\centering
	\resizebox{\columnwidth}{!}{
		\begin{tabular}[t]{lcccc}
			\toprule
			Regression Model & Parameter & Point Estimation & Coverage Rate ($\%$) & Width of CI\\
			\midrule
			& $\beta_{0}$ & 1.18 (0.17) & 91.67 & 0.76 (0.18)\\
			\cmidrule{2-5}
			& $\beta_{1}$ & 1.00 (0.12) & 91.33 & 0.41 (0.09)\\
			\cmidrule{2-5}
			\multirow{-3}{*}{\raggedright\arraybackslash TPSC-Student-$t$} & $\beta_{2}$ & 1.00 (0.12) & 92.00 & 0.41 (0.09)\\
			\cmidrule{1-5}
			& $\beta_{0}$ & 1.25 (0.15) & 42.00 & 0.54 (0.18)\\
			\cmidrule{2-5}
			& $\beta_{1}$ & 1.00 (0.16) & 89.67 & 0.55 (0.22)\\
			\cmidrule{2-5}
			\multirow{-3}{*}{\raggedright\arraybackslash MODLR} & $\beta_{2}$ & 1.01 (0.16) & 91.33 & 0.56 (0.23)\\
			\bottomrule
		\end{tabular}
	}
\end{table}

\section{More data applications of Bayesian modal regression}\label{sec:data_application}

\subsection{\revise Boston housing prices} \label{sec:boston_house}
{\revise To demonstrate the differences in interpretation of mean/median/modal regression, we analyzed a dataset containing $n=506$ house prices in the area of Boston, Massachusetts, in the year 1970. {\revisetwo This dataset is
		available in the \textsc{R} package \texttt{MASS} \citep{venables2013modern} under the name \texttt{Boston}}. We are interested in the intricate relationship between house prices and thirteen covariates recorded in the dataset. The models are formulated as follows:
	$$
	\mathbb{M} \left(Y_i \mid \boldsymbol{\beta }\right) = \boldsymbol{X}_i^\top \boldsymbol{\beta}, \text{ for }i=1, \ldots, 506,
	$$
	where the response $Y_i$ is the median value of owner-occupied homes in thousands dollars in the $i$th area, and $\boldsymbol{X}_i$ contains a 1 (for the intercept $\beta_0$) and values for the following 13 covariates associated with the same $i$th area: per capita crime rate by town ($X_{1,i}$), proportion of residential land zoned for lots over 25,000 square feet ($X_{2,i}$), proportion of non-retail business acres per town ($X_{3,i}$), Charles River dummy variable ($X_{4,i}$ = 1 if tract bounds river; 0 otherwise), Nitrogen oxides concentration in parts per 10 million ($X_{5,i}$), average number of rooms per dwelling ($X_{6,i}$), proportion of owner-occupied units built prior to 1940 ($X_{7,i}$), weighted mean of distances to five Boston employment centres ($X_{8,i}$), an index of accessibility to radial highways ($X_{9,i}$), full-value property-tax rate per $\$10{,}000$ ($X_{10,i}$), pupil-teacher ratio by town ($X_{11,i}$), the modified proportion of blacks by town ($X_{12,i}$), and percentage of the population that was lower status ($X_{13,i}$).
	
	Parallel to the study design in Section~\ref{sec:left_skewed_data}, we carried out four different regression analyses of this dataset, including the usual mean regression with normal likelihood, the more robust mean regression assuming a skewed normal model error (SNCP), median regression based on the ALD likelihood, and our proposed modal regression assuming a TPSC-Student-$t$ distribution for the model error distribution. Table \ref{tab:boston_house_point_estimation} reports estimates of the fourteen regression coefficients resulting from these four different analyses. 
	
	Because the interpretation of regression coefficients depends on whether the regression function is the conditional mean,  median, or mode of the response, we focus on comparing conclusions from different regression models with respect to statistical significance of a covariate effect on the house price. According to Table~\ref{tab:boston_house_point_estimation}, all four regression analyses reached the same conclusions for all covariates in this regard except for one covariate, the proportion of owner-occupied units built prior to 1940 (see results for $\beta_{7}$ in Table~\ref{tab:boston_house_point_estimation}). In particular, the $90\%$ credible intervals from the two mean regression models both contained zero, suggesting that, on average, the house price in Boston is not significantly influenced by the proportion of owner-occupied units built prior to 1940. Yet the counterpart results from the median regression model and the modal regression indicate that, had one looked into the association of the considered covariates with the median or the mode of house prices, one would conclude that the proportion of owner-occupied units built prior to 1940 has a significant, negative effect on house price.} 
\begin{table}[t]
	\centering
	\caption{\label{tab:boston_house_point_estimation}Parameter estimates obtained from the mean (normal and SNCP)/median (ALD)/modal regression models (TPSC-Student-$t$) fitted to the Boston house price data. The mean, $5\%$ quantile, and $95\%$ quantile of the posterior distribution of each regression coefficient are listed under Mean, q5, and q95, respectively.}
	\resizebox{\columnwidth}{!}{
		\begin{tabular}[t]{lrrrrrrrrrrrr}
			\toprule
			\multicolumn{1}{c}{ } & \multicolumn{3}{c}{Normal} & \multicolumn{3}{c}{SNCP} & \multicolumn{3}{c}{ALD} & \multicolumn{3}{c}{TPSC-Student-$t$} \\
			\cmidrule(l{3pt}r{3pt}){2-4} \cmidrule(l{3pt}r{3pt}){5-7} \cmidrule(l{3pt}r{3pt}){8-10} \cmidrule(l{3pt}r{3pt}){11-13}
			\multicolumn{1}{c}{Parameter} & \multicolumn{1}{c}{Mean} & \multicolumn{1}{c}{q5} & \multicolumn{1}{c}{q95} & \multicolumn{1}{c}{Mean} & \multicolumn{1}{c}{q5} & \multicolumn{1}{c}{q95} & \multicolumn{1}{c}{Mean} & \multicolumn{1}{c}{q5} & \multicolumn{1}{c}{q95} & \multicolumn{1}{c}{Mean} & \multicolumn{1}{c}{q5} & \multicolumn{1}{c}{q95} \\
			\cmidrule(l{3pt}r{3pt}){1-1} \cmidrule(l{3pt}r{3pt}){2-2} \cmidrule(l{3pt}r{3pt}){3-3} \cmidrule(l{3pt}r{3pt}){4-4} \cmidrule(l{3pt}r{3pt}){5-5} \cmidrule(l{3pt}r{3pt}){6-6} \cmidrule(l{3pt}r{3pt}){7-7} \cmidrule(l{3pt}r{3pt}){8-8} \cmidrule(l{3pt}r{3pt}){9-9} \cmidrule(l{3pt}r{3pt}){10-10} \cmidrule(l{3pt}r{3pt}){11-11} \cmidrule(l{3pt}r{3pt}){12-12} \cmidrule(l{3pt}r{3pt}){13-13}
			$\beta_{0}$ & 36.46 & 28.12 & 44.81 & 35.73 & 28.60 & 42.90 & 14.91 & 7.32 & 22.48 & 13.02 & 6.48 & 19.66\\
			$\beta_{1}$ & -0.11 & -0.16 & -0.05 & -0.12 & -0.17 & -0.07 & -0.12 & -0.16 & -0.06 & -0.13 & -0.17 & -0.09\\
			$\beta_{2}$ & 0.05 & 0.02 & 0.07 & 0.03 & 0.01 & 0.05 & 0.04 & 0.02 & 0.05 & 0.02 & 0.01 & 0.04\\
			$\beta_{3}$ & 0.02 & -0.08 & 0.12 & 0.03 & -0.05 & 0.11 & 0.01 & -0.05 & 0.07 & 0.01 & -0.05 & 0.06\\
			$\beta_{4}$ & 2.69 & 1.27 & 4.11 & 1.38 & 0.17 & 2.54 & 1.49 & 0.51 & 2.53 & 1.44 & 0.50 & 2.39\\
			\addlinespace
			$\beta_{5}$ & -17.80 & -24.06 & -11.53 & -14.02 & -19.48 & -8.69 & -8.99 & -13.64 & -4.36 & -6.71 & -10.73 & -2.79\\
			$\beta_{6}$ & 3.81 & 3.12 & 4.49 & 3.42 & 2.77 & 4.08 & 5.26 & 4.50 & 6.01 & 4.87 & 4.11 & 5.63\\
			$\beta_{7}$ & 0.00 & -0.02 & 0.02 & -0.02 & -0.03 & 0.00 & -0.03 & -0.04 & -0.01 & -0.04 & -0.05 & -0.02\\
			$\beta_{8}$ & -1.48 & -1.81 & -1.15 & -1.12 & -1.42 & -0.83 & -0.99 & -1.23 & -0.75 & -0.89 & -1.11 & -0.66\\
			$\beta_{9}$ & 0.31 & 0.20 & 0.42 & 0.24 & 0.14 & 0.34 & 0.18 & 0.09 & 0.26 & 0.14 & 0.07 & 0.21\\
			\addlinespace
			$\beta_{10}$ & -0.01 & -0.02 & -0.01 & -0.01 & -0.02 & -0.01 & -0.01 & -0.01 & -0.01 & -0.01 & -0.01 & -0.01\\
			$\beta_{11}$ & -0.95 & -1.17 & -0.74 & -0.82 & -1.00 & -0.64 & -0.75 & -0.90 & -0.59 & -0.61 & -0.73 & -0.48\\
			$\beta_{12}$ & 0.01 & 0.00 & 0.01 & 0.01 & 0.00 & 0.01 & 0.01 & 0.01 & 0.02 & 0.01 & 0.01 & 0.01\\
			$\beta_{13}$ & -0.52 & -0.61 & -0.44 & -0.43 & -0.51 & -0.35 & -0.31 & -0.39 & -0.23 & -0.28 & -0.35 & -0.21\\
			\bottomrule
		\end{tabular}
	}
\end{table}

{\revise To provide further context and insights, we present the posterior prediction coverage rate of the Boston house prices, the width of prediction interval, and ELPD in Table \ref{tab:boston_house_prediction}. Using our proposed modal regression, the prediction intervals attained similar empirical coverage rates, while also giving the tightest prediction intervals on average among the four models considered. Our modal regression model also provided the overall best fit for this dataset according to ELPD.} 
\begin{table}[ht]
	\caption{\label{tab:boston_house_prediction}Comparison of Bayesian mean/median/modal regression models fitted to the Boston house price data.}
	\centering
	\resizebox{\columnwidth}{!}{
		\begin{tabular}[t]{lrrr}
			\toprule
			Likelihood (regression model)  & Coverage Rate (\%) & Width & ELPD\\
			\midrule
			Normal (mean regression) & 93.28 & 15.82 & -1518.66\\
			SNCP (mean regression) & 91.30 & 14.38 & -1464.62\\
			ALD (median regression) & 90.51 & 14.55 & -1444.57\\
			TPSC-Student-$t$ (modal regression) & 89.33 & 13.12 & -1408.28\\
			\bottomrule
		\end{tabular}
	}
\end{table}

\subsection{Detecting a quadratic relationship in serum data}\label{sec:serum}
\citet{isaacs1983serum} analyzed the relationship between serum concentration (grams per litre) of immunoglobulin-G (IgG) in 298 children aged from 6 months to 6 years. IgG is an antibody that plays an important role in humoral and protective immunity \citep{van2016emerging}. There are ethical difficulties in taking repeated blood samples from healthy subjects. Therefore, researchers often use age as a proxy for determining the reference ranges for IgG in childhood. Previously, \citet{yu2001bayesian} analyzed serum data and modeled IgG concentration with a quadratic model in age. In the spirit of \citet{yu2001bayesian}, we fit the following mean/median/modal regression models to this dataset:
$$
\mathbb{M}(Y\mid \boldsymbol{\beta})=\beta_0+\beta_1 \times \operatorname{Age}+\beta_2 \times \operatorname{Age}^2.
$$
Table \ref{tab:tab3} shows the parameter estimates from the models that we fit to this data. Based on the CIs for $\beta_2$, we see that only the modal regression model was able to detect the quadratic term (CI of $(-0.18, -0.03)$ for modal regression). This finding is somewhat consistent with \citet{royston1994regression} who concluded that a simple linear regression model was inadequate for this same dataset. \citet{isaacs1983serum} also suggested that there was a quadratic relationship between the square root of IgG concentration and children's age. 
\begin{table}[t]
	\caption{\label{tab:tab3} Parameter estimates from the mean/median/modal regression models fitted to the serum data. The mean, $5\%$ quantile, and $95\%$ quantile of the posterior distribution of each regression coefficient are listed under Mean, q5, and q95, respectively.}
	\centering
	\resizebox{\columnwidth}{!}{
		\begin{tabular}{lrlrrr}
			\toprule
			Likelihood (regression model) & ELPD & Parameter & Mean & q5 & q95\\
			\midrule
			&  & $\beta_0\left(\operatorname{intercept}\right)$ & 3.09 & 2.46 & 3.73\\
			\cmidrule{3-6}
			&  & $\beta_1\left(\operatorname{Age}\right)$ & 0.96 & 0.44 & 1.47\\
			\cmidrule{3-6}
			\multirow{-3}{*}{\raggedright\arraybackslash Normal (mean regression)} & \multirow{-3}{*}{\raggedleft\arraybackslash -627.12} & $\beta_2\left(\operatorname{Age}^2\right)$ & -0.05 & -0.13 & 0.04\\
			\cmidrule{1-6}
			&  & $\beta_0\left(\operatorname{intercept}\right)$ & 2.81 & 2.12 & 3.55\\
			\cmidrule{3-6}
			&  & $\beta_1\left(\operatorname{Age}\right)$ & 1.12 & 0.54 & 1.69\\
			\cmidrule{3-6}
			\multirow{-3}{*}{\raggedright\arraybackslash ALD (median regression)} & \multirow{-3}{*}{\raggedleft\arraybackslash -638.12} & $\beta_2\left(\operatorname{Age}^2\right)$ & -0.07 & -0.16 & 0.03\\
			\cmidrule{1-6}
			&  & $\beta_0\left(\operatorname{intercept}\right)$ & 2.37 & 1.85 & 2.89\\
			\cmidrule{3-6}
			&  & $\beta_1\left(\operatorname{Age}\right)$ & 1.15 & 0.72 & 1.59\\
			\cmidrule{3-6}
			\multirow{-3}{*}{\raggedright\arraybackslash FG (modal regression)} & \multirow{-3}{*}{\raggedleft\arraybackslash -623.18} & $\beta_2\left(\operatorname{Age}^2\right)$ & -0.11 & -0.18 & -0.03\\
			\bottomrule
		\end{tabular}
	}
\end{table}

Table \ref{tab:tab3} shows that the ELPD of the modal regression model based on the FG likelihood \eqref{eq:FG.model} was larger than the ELPD for both the mean or median regression models. In this example, the modal regression model not only provides a different viewpoint (i.e. that there exists a significant quadratic relationship between IgG and age), but it also fits the dataset better according to our model selection criterion.

\section{Discussion} \label{sec:Discussion}

In this paper, we have introduced a unifying Bayesian modal regression framework. Namely, we proposed a simple and flexible unimodal distribution family called the GUD family that is suitable for Bayesian modal regression. Members of this family can be either symmetric or asymmetric, either thin-tailed or fat-tailed, depending on values of the shape and scale parameters. Our framework adds to the existing literature on likelihood-based robust regression \citep{ronchetti2009robust, box1968bayesian,lange1989robust,da2020bayesian,gagnon2020new}. In particular, the GUD family exhibits robustness to outliers, skewness, and model misspecification, as demonstrated in Section \ref{sec:Simulation}. In contrast to other parametric families designed for robust regression \citep{azzalini2013skew}, however, a notable feature of the GUD family is that all members of this family have a location parameter that is \emph{also} the conditional mode. This makes the GUD family suitable for inference and prediction of the conditional mode specifically. 

Compared to mean and quantile regression, work on Bayesian modal regression analysis is quite scarce. Our paper aims to promote Bayesian modal regression as a complement to these other analyses. We demonstrated that our modeling framework based on the GUD family is very versatile and has wide applications in many fields such as economics (the bank deposit data in Section \ref{sec:bank} and the Boston house prices data in Section \ref{sec:boston_house}), criminology (the murder rate data in Section \ref{sec:murder}), and molecular biology (the serum data in Section \ref{sec:serum}). In particular, we showed that Bayesian modal regression can reveal structures and detect potentially significant covariate effects that are missed by other Bayesian regression models.

To conduct Bayesian inference of the conditional mode, we provided prior elicitation procedures, along with the sufficient conditions under which a flat prior $p(\boldsymbol{\beta})$ on the regression coefficients $\boldsymbol{\beta}$ can be used. We proposed a method for constructing posterior prediction intervals and a model selection criterion based on the posterior predictive distribution. We demonstrated that our modal regression models provide very tight prediction intervals with high coverage, are robust to outliers, and have excellent interpretability. 

We stress that it is important not to fit only one type of regression model. In practice, researchers should fit several models to the data and utilize regression diagnostics to evaluate model assumptions and determine whether there are any influential observations. Our modal regression model framework is an especially appealing choice when the data is skewed and(or) contains (extreme) outliers. Moreover, model selection criteria such as ELPD can be used to select a suitable (mean, quantile, or modal) regression model for final analyses. Other posterior predictive checks, e.g. those described in Chapter 6 of \citet{gelman2013bayesian}, can also be used to assess the appropriateness of using a GUD likelihood for Bayesian inference.

The modal regression models in this paper contain parametric assumptions, both about the data likelihood and the linear relationship between the covariates and the conditional mode. Instead of using the fully parametric models presented in this manuscript, one may prefer to use  Bayesian semiparametric modal regression models instead. A Bayesian semiparametric modal regression model can be constructed either by modeling the conditional mode with a Gaussian process (i.e. we can relax the linearity assumption) and/or by replacing the GUD likelihood with a carefully constructed infinite mixture model that is indexed by the mode (i.e. we can relax the assumption of a known residual error distribution). These exciting extensions to Bayesian modal regression are the topics of ongoing work.  

Another interesting future direction to explore is Bayesian modal regression in high dimensions. When the number of covariates $p$ is large or even exceeds the sample size $n$, some form of regularization is typically desired. In this case, we can replace the flat prior on $\boldsymbol{\beta}$ with a spike-and-slab prior \citep{MitchellBeauchamp1988, GeorgeMcCulloch1993, RockovaGeorge2018} or a global-local shrinkage prior \citep{Bhadra2019, GriffinBrown2021}. These priors shrink most of the regression coefficients in $\boldsymbol{\beta}$ towards zero and allow for variable selection. We anticipate that these types of priors would work well in high-dimensional Bayesian modal regression with the GUD likelihood, especially if $p>n$. 

\section*{Declaration of generative AI in scientific writing}

During the preparation of this work the authors used ChatGPT in order to check grammar. After using this tool/service, the authors reviewed and edited the content as needed and take(s) full responsibility for the content of the publication.

\section*{Acknowledgments}

We are grateful to the Associate Editor and two anonymous reviewers for their thoughtful comments and suggestions which helped to greatly improve our article. The last listed author was partially support by National Science Foundation grant DMS-2015528.

\appendix
\section{Proofs of main results} \label{sec:proofs}

\subsection{Preliminary lemmas}

Before proving the main theorems and propositions, we first prove the following two lemmas.

\begin{lemma}\label{thm:student.t}
	Let $p\left(x\right)$ be the pdf of an inverse gamma distribution with shape and scale parameters $a \in (0, \infty)$ and $b \in (0, \infty)$ respectively. Then for $p \leq n$,
	$$
	\int_{0}^{\infty} \left\{\frac{\Gamma \left(0.5x+0.5\right)}{\Gamma \left(0.5x\right)}\right\}^{n-p} x^{0.5p-0.5n} p\left(x\right) dx < \infty.
	$$
\end{lemma}
\begin{proof}
	First, let us split the integral into two parts,
	\begin{equation}
		\label{LemmaA1proof:pt1}
		\begin{aligned}
			&\ \int_{0}^{\infty} \left\{\frac{\Gamma \left(0.5x+0.5\right)}{\Gamma \left(0.5x\right)}\right\}^{n-p} x^{0.5p-0.5n} p\left(x\right) dx  \\
			= &\ \int_{0}^{1} \left\{\frac{\Gamma \left(0.5x+0.5\right)}{\Gamma \left(0.5x\right)}\right\}^{n-p} x^{0.5p-0.5n} p\left(x\right) dx \\
			&\ + \int_{1}^{\infty} \left\{\frac{\Gamma \left(0.5x+0.5\right)}{\Gamma \left(0.5x\right)}\right\}^{n-p} x^{0.5p-0.5n} p\left(x\right) dx \\
			:= &\ I_1 + I_2.
		\end{aligned}
	\end{equation}
	We next consider the integrals $I_1$ and $I_2$ separately. 
	
	Because $\Gamma(x)$ is strictly decreasing for $x \in (0,1)$, we have 
	$$
	\Gamma\left(0.5x+0.5\right) < \Gamma(0.5x), \quad \forall x \in (0,1).
	$$
	Therefore,
	\begin{equation}
		\label{LemmaA1proof:pt2}
		\begin{aligned}
			I_1 = & \int_{0}^{1} \left\{\frac{\Gamma \left(0.5x+0.5\right)}{\Gamma \left(0.5x\right)}\right\}^{n-p} x^{0.5p-0.5n} p\left(x\right) dx\\
			< &\ \int_{0}^{1} x^{0.5p-0.5n} p\left(x\right) dx \\
			\propto &\ \int_{0}^{1} x^{-(a+0.5n-0.5p)-1} \exp\left(-b/x\right) dx \\
			< &\ \infty,
		\end{aligned}
	\end{equation}
	where the last line of the display is because $x^{-(a+0.5n-0.5p)-1} \exp\left(-b/x\right)$ is the kernel of an inverse gamma distribution with $a+0.5n-0.5p$ and $b$ as the shape and scale parameter respectively.
	
	Finally, we consider $I_2$. By Gautschi's inequality, 
	$$
	\frac{\Gamma(x+1)}{\Gamma(x+s)}<(x+1)^{1-s}, \quad \forall x >0, \, 0 < s < 1, 
	$$
	thus
	$$
	\frac{\Gamma(0.5x+0.5)}{\Gamma(0.5x)}<0.5^{0.5}(x+1)^{0.5}, \quad \forall x > 1.
	$$
	Hence,
	\begin{equation}
		\label{LemmaA1proof:pt3}
		\begin{aligned}
			I_2 = &\ \int_{1}^{\infty} \left\{\frac{\Gamma \left(0.5x+0.5\right)}{\Gamma \left(0.5x\right)}\right\}^{n-p} x^{0.5p-0.5n} p\left(x\right) dx \\
			< &\ \int_{1}^{\infty} 0.5^{0.5n-0.5p}\left(x+1\right)^{0.5n-0.5p} x^{0.5p-0.5n} p\left(x\right) dx \\
			\propto &\ \int_{1}^{\infty} \left(1+1/x\right)^{0.5n-0.5p} p\left(x\right) dx \\
			< &\ \int_{1}^{\infty} 2^{0.5n-0.5p} p\left(x\right) dx \\
			< &\ 2^{0.5n-0.5p} \\
			< &\ \infty.
		\end{aligned}
	\end{equation}
	Combining \eqref{LemmaA1proof:pt1}-\eqref{LemmaA1proof:pt3}, one proves the assertion. 
\end{proof}

\begin{lemma}\label{thm:locp}
	Let $\boldsymbol{U}_{p \times p}= [ \boldsymbol{U}_1, \dots, \boldsymbol{U}_p ]^{\top}$ be a nonsingular design matrix with finite entries. If $f_Z(y-\theta)$ is the pdf of a distribution from the location family with $\theta \in \mathbb{R}$ as the location parameter, then
	$$
	\int_{\mathbb{R}^p} \prod_{i=1}^p f_Z\left(y_i-\boldsymbol{U}_i^{\top} \boldsymbol{\beta}\right) d \boldsymbol{\beta}=1 /|\operatorname{det}(\boldsymbol{U})| .
	$$
\end{lemma}
\begin{proof}
	For $\boldsymbol{\theta}=\left[\theta_1, \ldots, \theta_p\right]^{\top}$, let 
	$$\boldsymbol{\theta}=\boldsymbol{U} \boldsymbol{\beta}.$$ 
	Since $\boldsymbol{U}$ is a nonsingular design matrix with finite entries,
	$$
	\boldsymbol{U}^{-1} \boldsymbol{\theta}=\boldsymbol{\beta}.
	$$
	Hence, the corresponding Jacobian matrix of the one-to-one transformation is
	$$
	\frac{\partial \boldsymbol{\beta}}{\partial \boldsymbol{\theta}}=\boldsymbol{U}^{-1} .
	$$
	Using a change of variables with $\theta_i=\boldsymbol{U}_i^{\top} \boldsymbol{\beta}$, we have that
	$$
	\begin{aligned}
		\int_{\mathbb{R}^p} \prod_{i=1}^p f_Z\left(y_i-\boldsymbol{U}_i^{\top} \boldsymbol{\beta}\right) d \boldsymbol{\beta} & =\int_{\mathbb{R}^p} \prod_{i=1}^p f_Z\left(y_i-\theta_i\right)\left|\operatorname{det}\left(\boldsymbol{U}^{-1}\right)\right| d \boldsymbol{\theta} \\
		& =\left|\operatorname{det}\left(\boldsymbol{U}^{-1}\right)\right| \int_{\mathbb{R}^p} \prod_{i=1}^p f_Z\left(y_i-\theta_i\right) d \boldsymbol{\theta} \\
		& =\left|\operatorname{det}\left(\boldsymbol{U}^{-1}\right)\right| \\
		& =1 /|\operatorname{det}(\boldsymbol{U})| .
	\end{aligned}
	$$ 
	This completes the proof.
\end{proof}

\subsection{Proofs of Theorems \ref{thm:suff1} and \ref{thm:suff2} and Propositions \ref{thm:proper.posterior.reg} and \ref{thm:improper} in the main article}

\subsection*{Proof of Theorem \ref{thm:suff1} in the main article}

\begin{proof}
	Recall that $f_Z\left(y \mid w, \boldsymbol{\xi}_1, \boldsymbol{\xi}_2\right)$ has $y=0$ as the global mode such that
	$$
	f_Z\left(0 \mid w, \boldsymbol{\xi}_1, \boldsymbol{\xi}_2\right) \geq f_Z\left(y \mid w, \boldsymbol{\xi}_1, \boldsymbol{\xi}_2\right), \quad \forall y, w, \boldsymbol{\xi}_1, \boldsymbol{\xi}_2.
	$$
	Therefore, 
	$$
	\begin{aligned}
		\prod_{i=1}^n f\left(y_i \mid w, \theta, \boldsymbol{\xi}_1, \boldsymbol{\xi}_2\right) & =\prod_{i=1}^n f_Z\left(y_i-\theta \mid w, \boldsymbol{\xi}_1, \boldsymbol{\xi}_2\right) \\
		& \leq f_Z\left(y_1-\theta \mid w, \boldsymbol{\xi}_1, \boldsymbol{\xi}_2\right) \{ f_Z\left(0 \mid w, \boldsymbol{\xi}_1, \boldsymbol{\xi}_2\right) \}^{n-1}.
	\end{aligned}
	$$
	Using this upper bound for the likelihood, we have, for the posterior distribution, 
	\begin{align*}
		&\ p(w, \theta, \boldsymbol{\xi}_1, \boldsymbol{\xi}_2|\boldsymbol{X}, \boldsymbol{Y}) \\
		\propto &\ \left\{\prod_{i=1}^n f\left(y_i \mid w, \theta, \boldsymbol{\xi}_1, \boldsymbol{\xi}_2\right) \right\}p(w) p(\theta)p(\boldsymbol{\xi}_1)p(\boldsymbol{\xi}_2) \\
		\le &\ f_Z\left(y_1-\theta \mid w, \boldsymbol{\xi}_1, \boldsymbol{\xi}_2\right) \{ f_Z\left(0 \mid w, \boldsymbol{\xi}_1, \boldsymbol{\xi}_2\right) \}^{n-1} p(w) p(\theta)p(\boldsymbol{\xi}_1)p(\boldsymbol{\xi}_2).
	\end{align*}
	We next integrate the preceding expression with respect to $\theta$, then with respect to $(w, \boldsymbol{\xi}_1,\boldsymbol{\xi}_2)$ to check the propriety of $p(w, \theta, \boldsymbol{\xi}_1, \boldsymbol{\xi}_2|\boldsymbol{X}, \boldsymbol{Y})$.
	
	Taking integration with respect to $\theta$ and using a change of variables, we have that 
	$$
	\int_{-\infty}^{+\infty} f_Z\left(y_1-\theta \mid w, \boldsymbol{\xi}_1, \boldsymbol{\xi}_2\right) \{ f_Z\left(0 \mid w, \boldsymbol{\xi}_1, \boldsymbol{\xi}_2\right) \}^{n-1} d \theta= \{ f_Z\left(0 \mid w, \boldsymbol{\xi}_1, \boldsymbol{\xi}_2\right) \}^{n-1}.
	$$
	Finally, by the sufficient condition given in Theorem \ref{thm:suff1} in the main article, we have
	$$
	\iiint_{\Theta_{w, \boldsymbol{\xi}_1, \boldsymbol{\xi}_2}} \{ f_Z\left(0 \mid w, \boldsymbol{\xi}_1, \boldsymbol{\xi}_2\right) \}^{n-1} p(w) p\left(\boldsymbol{\xi}_1\right) p\left(\boldsymbol{\xi}_2\right)\, d w \, d \boldsymbol{\xi}_1 \,  d \boldsymbol{\xi}_2<\infty.
	$$
	It follows that $$
	\iiint_{\Theta_{w, \boldsymbol{\xi}_1, \boldsymbol{\xi}_2}} \int_{-\infty}^{+\infty}p(w, \theta, \boldsymbol{\xi}_1, \boldsymbol{\xi}_2|\boldsymbol{X}, \boldsymbol{Y}) \, d\theta \, d w \, d \boldsymbol{\xi}_1 \, d \boldsymbol{\xi}_2<\infty.
	$$
	This shows that the posterior distribution is proper. 
\end{proof}

\subsection*{Proof of Theorem \ref{thm:suff2} in the main article}

\begin{proof}
	By assumption, the $n \times p$ design matrix $\boldsymbol{X} := [ \boldsymbol{X}_1, \ldots, \boldsymbol{X}_n ]^\top$ is full rank. Without loss of generality, we assume that the first $p$ rows of $\boldsymbol{X}$ are linearly independent. Define the submatrix $\boldsymbol{U}_{p \times p}$ consisting of the first $p$ rows of $\boldsymbol{X}$.
	
	Using the fact that the GUD family is a unimodal location family, we have that
	$$
	\begin{aligned}
		\prod_{i=1}^n f\left(y_i \mid w, \boldsymbol{\beta}, \boldsymbol{\xi}_1, \boldsymbol{\xi}_2\right) & =\prod_{i=1}^n f_Z\left(y_i-\boldsymbol{X}_i^{\top} \boldsymbol{\beta} \mid w, \boldsymbol{\xi}_1, \boldsymbol{\xi}_2\right) \\
		& \leq \prod_{i=1}^p f_Z\left(y_i-\boldsymbol{X}_i^{\top} \boldsymbol{\beta} \mid w, \boldsymbol{\xi}_1, \boldsymbol{\xi}_2\right) \{f_Z\left(0 \mid w, \boldsymbol{\xi}_1, \boldsymbol{\xi}_2\right) \}^{n-p}.
	\end{aligned}
	$$
	Using this upper bound for the likelihood, we have, for the posterior distribution, 
	\begin{align*}
		&\ p(w, \boldsymbol{\beta}, \boldsymbol{\xi}_1, \boldsymbol{\xi}_2|\boldsymbol{X}, \boldsymbol{Y}) \\
		\propto &\ \left\{\prod_{i=1}^n f\left(y_i \mid w, \boldsymbol{\beta}, \boldsymbol{\xi}_1, \boldsymbol{\xi}_2\right) \right\} p(w) p(\boldsymbol{\beta})p(\boldsymbol{\xi}_1)p(\boldsymbol{\xi}_2) \\
		\le &\  \prod_{i=1}^p f_Z\left(y_i-\boldsymbol{X}_i^{\top} \boldsymbol{\beta} \mid w, \boldsymbol{\xi}_1, \boldsymbol{\xi}_2\right) \{f_Z\left(0 \mid w, \boldsymbol{\xi}_1, \boldsymbol{\xi}_2\right) \}^{n-p} p(w)p(\boldsymbol{\beta})p(\boldsymbol{\xi}_1)p(\boldsymbol{\xi}_2).
	\end{align*}
	We next integrate the preceding expression with respect to $\boldsymbol{\beta}$, then with respect to $(w, \boldsymbol{\xi}_1,\boldsymbol{\xi}_2)$ to check the propriety of $p(w, \boldsymbol{\beta}, \boldsymbol{\xi}_1, \boldsymbol{\xi}_2|\boldsymbol{X}, \boldsymbol{Y})$.
	
	By Lemma \ref{thm:locp},
	$$
	\begin{aligned}
		& \int_{\mathbb{R}^p} \prod_{i=1}^p f_Z\left(y_i-\boldsymbol{X}_i^{\top} \boldsymbol{\beta} \mid w, \boldsymbol{\xi}_1, \boldsymbol{\xi}_2\right) \{ f_Z \left(0 \mid w, \boldsymbol{\xi}_1, \boldsymbol{\xi}_2\right) \}^{n-p} d \boldsymbol{\beta} \\
		= & \{ f_Z \left(0 \mid w, \boldsymbol{\xi}_1, \boldsymbol{\xi}_2\right) \}^{n-p} /|\operatorname{det}(\boldsymbol{U})| .
	\end{aligned}
	$$
	Finally, by the sufficient condition given in Theorem \ref{thm:suff2} in the main article, we have
	$$
	1/|\operatorname{det}(\boldsymbol{U})| \iiint_{\Theta_{w, \boldsymbol{\xi}_1, \boldsymbol{\xi}_2}} \{ f_Z \left(0 \mid w, \boldsymbol{\xi}_1, \boldsymbol{\xi}_2\right) \}^{n-p} p(w) p\left(\boldsymbol{\xi}_1\right) p\left(\boldsymbol{\xi}_2\right) d w d \boldsymbol{\xi}_1 d \boldsymbol{\xi}_2<\infty.
	$$
	Since $\boldsymbol{X}$ has finite entries (and thus, so does $\boldsymbol{U}$), it must be that $\operatorname{det}(\boldsymbol{U})$ is a constant. Further, $\boldsymbol{U}$ is nonsingular, so $\operatorname{det}(\boldsymbol{U}) \neq 0$. It follows that 
	$$\iiint_{\Theta_{w, \boldsymbol{\xi}_1, \boldsymbol{\xi}_2}} \int_{\mathbb{R}^p} p(w, \boldsymbol{\beta}, \boldsymbol{\xi}_1, \boldsymbol{\xi}_2\mid \boldsymbol{X}, \boldsymbol{Y}) \, d \boldsymbol{\beta} \, d w \, d \boldsymbol{\xi}_1 \, d \boldsymbol{\xi}_2<\infty.$$
	This shows that the posterior distribution is proper.
\end{proof}

\subsection*{Proof of Proposition \ref{thm:proper.posterior.reg} in the main article}

\begin{proof}
	The proof consists of verifying the sufficient condition in Theorem \ref{thm:suff2} in the main article for the three considered Bayesian modal regression models.
	
	First, we show that the modal regression model based on the FG distribution in \eqref{eq:FG.model} has a proper posterior distribution. By Theorem \ref{thm:suff2} in the main article, we need to show that 
	$$
	\int_{0}^{\infty}\int_{0}^{\infty}\int_{0}^{1} \{ f_{\operatorname{FG}}\left(y = 0 \mid w, \theta = 0, \sigma_1, \sigma_2 \right) \}^{n-p} p(w)p(\sigma_1)p(\sigma_2) dw d\sigma_1 d\sigma_2 < \infty.
	$$
	Note that
	$$
	\begin{aligned}
		f_{\operatorname{FG}}\left(y = 0 \mid w, \theta = 0, \sigma_1, \sigma_2 \right) &= \exp(-1)\left(w/\sigma_1 + (1-w)/\sigma_2\right) \\
		& \le \exp(-1)\left(1/\sigma_1 + 1/\sigma_2\right).
	\end{aligned}
	$$
	Therefore, it is sufficient to show 
	$$
	\begin{aligned}
		& \int_0^{\infty} \int_0^{\infty} \int_0^1\left(1 / \sigma_1+1 / \sigma_2\right)^{n-p} p(w) p\left(\sigma_1\right) p\left(\sigma_2\right) d w d \sigma_1 d \sigma_2 \\
		= & \int_0^{\infty} \int_0^{\infty} \left(1 / \sigma_1+1 / \sigma_2\right)^{n-p} p\left(\sigma_1\right) p\left(\sigma_2\right) d \sigma_1 d \sigma_2\\
		= & \int_0^{\infty} \int_0^{\infty} \sum_{k=0}^{n-p} { \binom{n-p}{k}}\left(1/\sigma_1\right)^{n-p-k} \left(1/\sigma_2\right)^{k} p\left(\sigma_1\right) p\left(\sigma_2\right) d \sigma_1 d \sigma_2 \\
		< & \infty.
	\end{aligned}
	$$
	The last inequality is true because we use the inverse gamma distribution as the prior for $\sigma_1$ and $\sigma_2$ and because for any inverse gamma random variable $X$, $\mathbb{E}[1/X^{k}] < \infty$ for all $k \in \mathbb{N}$.
	
	Second, we want to show that the linear modal regression model based on the DTP-Student-$t$ distribution \eqref{eq:DTP.model} has a proper posterior distribution. We have
	$$
	\begin{aligned}
		f_{\operatorname{DTP-Student-}t}\left(y = 0 \mid \theta = 0, \sigma_1, \sigma_2, \delta_1, \delta_2 \right) &= 2\left(1-w\right)\frac{\Gamma\left(0.5\delta_2 + 0.5\right)}{\Gamma(0.5 \delta_2)} \frac{1}{\sqrt{\delta_2 \pi} \sigma_2} \\
		&\le 2 \frac{\Gamma\left(0.5\delta_2 + 0.5\right)}{\Gamma(0.5 \delta_2)} \frac{1}{\sqrt{\delta_2 \pi} \sigma_2} \\
		&< 2 \frac{\Gamma\left(0.5\delta_2 + 0.5\right)}{\Gamma(0.5 \delta_2)} \frac{1}{\sqrt{\delta_2} \sigma_2},
	\end{aligned}
	$$
	where $w \in [0,1]$ is defined in \eqref{eq:pdf.DTP.w}. Applying Lemma \ref{thm:student.t} and the fact that for any inverse gamma random variable $X$, $\mathbb{E}[1/X^{k}] < \infty$ for all $k \in \mathbb{N}$, we have
	$$
	\begin{aligned}
		& \int_{0}^{\infty}\int_{0}^{\infty}\left\{\frac{\Gamma\left(0.5 \delta_2+0.5\right)}{\Gamma\left(0.5 \delta_2\right)} \frac{1}{\sqrt{\delta_2} \sigma_2}\right\}^{n-p}p(\sigma_2)p(\delta_2) d\sigma_2 d\delta_2 \\
		= &\ \int_{0}^{\infty} \left[ \int_{0}^{\infty} \left\{ \frac{\Gamma \left( 0.5 \delta_2 + 0.5 \right)}{\Gamma \left( 0.5 \delta_2 \right)} \right\}^{n-p} \delta_2^{0.5p-0.5n} p(\delta_2) d \delta_2 \right] \left( \frac{1}{\sigma_2} \right)^{n-p} p(\sigma_2) d \sigma_2  \\
		< &\ \infty.
	\end{aligned}
	$$
	Therefore, by Theorem \ref{thm:suff2} in the main article, the posterior distribution for regression model in \eqref{eq:DTP.model} is proper. 
	
	Lastly, one can show that the linear modal regression model based on the TPSC-Student-$t$ distribution \eqref{eq:TPSC.model} also has a proper posterior distribution. The proof is almost identical to the proof of posterior propriety for the DTP-Student-$t$ distribution and is therefore omitted.
\end{proof}

\subsection*{Proof of Proposition \ref{thm:improper} in the main article}

\begin{proof}
	Recall that the pdf of the GUD family is
	$$
	f\left(y \mid w, \theta, \boldsymbol{\xi}_1, \boldsymbol{\xi}_2\right)=w f_1\left(y \mid \theta, \boldsymbol{\xi}_1\right)+(1-w) f_2\left(y \mid \theta, \boldsymbol{\xi}_2\right) .
	$$
	Without loss of generality, we assume that $\tau \in \boldsymbol{\xi}_1$ and $\tau \notin \boldsymbol{\xi}_2$. Suppose that there is only one observation, we have that
	$$
	\int_{\tau \in \Theta_{\tau}} p\left(w, \theta, \boldsymbol{\xi}_1, \boldsymbol{\xi}_2 \mid y\right) d \tau \ge \int_{\tau \in \Theta_{\tau}} (1-w) f_2\left(y \mid \theta, \boldsymbol{\xi}_2\right) p\left(\tau\right) d \tau 
	= \infty, 
	$$
	since $p(\tau)$ is improper. 
	
	When there is more than one observation, binomial expansion of the GUD likelihood gives $\prod_{i=1}^n f\left(y_i \mid w, \theta, \boldsymbol{\xi}_1, \boldsymbol{\xi}_2\right)\ge C[f_2\left(y \mid \theta, \boldsymbol{\xi}_2\right)]^n$, where $C>0$ is free of $\boldsymbol{\xi}_1$. Hence, for $C[f_2\left(y \mid \theta, \boldsymbol{\xi}_2\right)]^n p\left(\tau\right)$, the integration with respect to $\tau$ is still divergent when $p(\tau)$ is improper.
\end{proof}

\section{The lognormal mixture distribution} \label{sec:logNM}

To demonstrate how researchers can add new members to the GUD family, we present the construction of the lognormal mixture distribution (logNM) below. Here, we pick the lognormal distribution because the lognormal distribution is right-skewed and unimodal. We construct a location-shift lognormal distribution such that the transformed lognormal distribution is still right-skewed but has a mode at $0$. Next, we flip the location-shift lognormal distribution at $0$ to get a left-skewed unimodal lognormal distribution. Finally, we mix the left- and right-skewed lognormal distribution together to construct the logNM distribution. More details of the construction of logNM can be found in \eqref{eq:pdf.logN}-\eqref{eq:pdf.logNM}.

The pdf of the lognormal distribution is 
\begin{equation}
	f_{\operatorname{logN}}(y \mid \mu, \nu)	= \frac{1}{y \nu \sqrt{2 \pi}} \exp \left(-\frac{\left(\ln (y)-\mu\right)^2}{2 \nu^2}\right) \mathbb{I}(y > 0),
	\label{eq:pdf.logN}
\end{equation}
where $\mu\in (-\infty, +\infty)$ and $\nu>0$ are two parametrers, and the mode is given by $\exp\left(\mu-\nu^2\right)$. We define the pdf of logNM as a mixture of two lognormal pdfs formulated as follows, 
\begin{equation}
	\begin{aligned}
		f_{\operatorname{logNM}}(y\mid w,\theta,\boldsymbol{\xi}_1,\boldsymbol{\xi}_2) &= wf_1(y|\theta, \boldsymbol{\xi}_1) + \left(1-w\right) f_2(y|\theta, \boldsymbol{\xi}_2),  \\
		f_1(y|\theta,\boldsymbol{\xi}_1) & = f_{\operatorname{logN}}\left(\exp\left(\mu_1-\nu^2_1\right)-\left(y-\theta\right) \mid  \mu_1, \nu_1 \right),\\
		f_2(y|\theta, \boldsymbol{\xi}_2) & = f_{\operatorname{logN}}\left(\exp\left(\mu_2-\nu^2_2\right)+\left(y-\theta\right) \mid \mu_2, \nu_2 \right),
	\end{aligned}
	\label{eq:pdf.logNM}
\end{equation}
where $\boldsymbol{\xi}_1 = \left[\mu_1,\nu_1\right]^{\top}$ and $\boldsymbol{\xi}_2 = \left[\mu_2,\nu_2\right]^{\top}$. It can be shown that both component distributions in  \eqref{eq:pdf.logNM} are unimodal at $\theta$, with continuous densities over the real line, one left-skewed and the other  right-skewed. Moreover, the pdfs of the individual lognormal mixture components in \eqref{eq:pdf.logNM} have $0$ at the right and left boundaries of their supports. Hence, the pdf of the logNM distribution is continuous in all of $\mathbb{R}$. Having verified that all three restrictions (R1)-(R3) for the GUD family (introduced in Section 3 of the main manuscript), we can proceed to use the logNM likelihood \eqref{eq:logNM.model} for Bayesian modal regression. 

Figure \ref{fig:logNM_pdfs} demonstrates that the logNM distributions can be asymmetric or symmetric given different combinations of parameter values. The top panel shows that, with an increase of $\mu_2$, the right tail of the logNM distribution becomes heavier while its left tail remains almost the same. The bottom panel shows how $\nu_1$ influences the amount of skewness of the logNM distribution when all three logNM distributions are left-skewed.

\begin{figure}[t]
	\centering
	\begin{subfigure}[b]{0.90\columnwidth}
		\centering
		\includegraphics[width=\columnwidth]{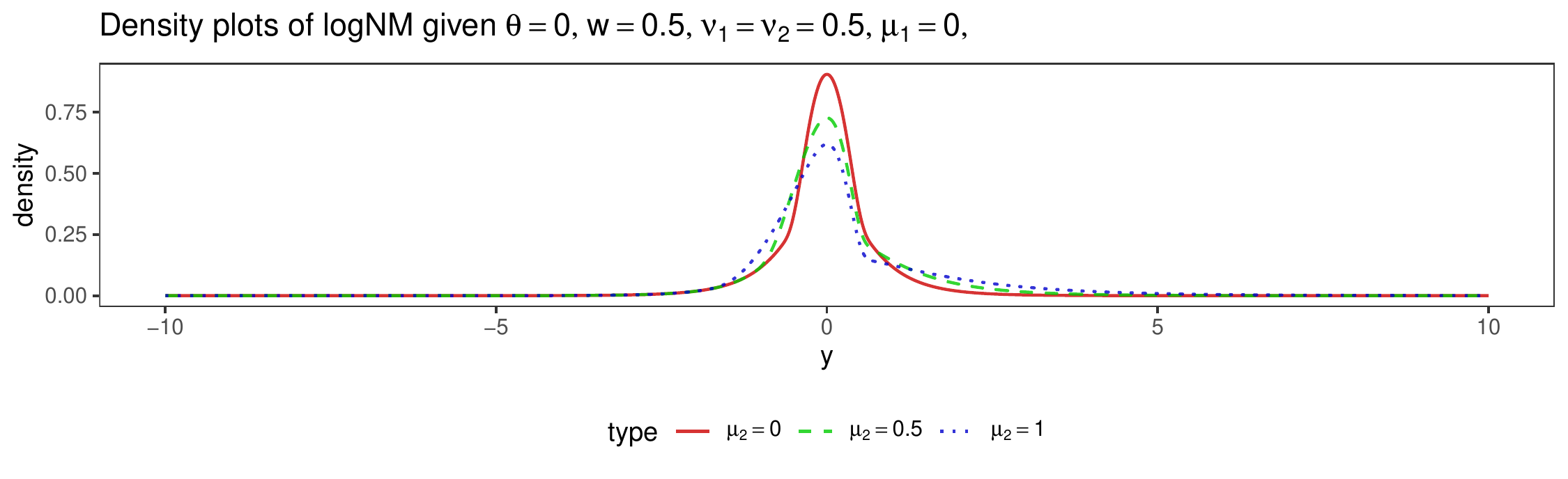}
	\end{subfigure}
	\begin{subfigure}[b]{0.90\columnwidth}
		\centering
		\includegraphics[width=\columnwidth]{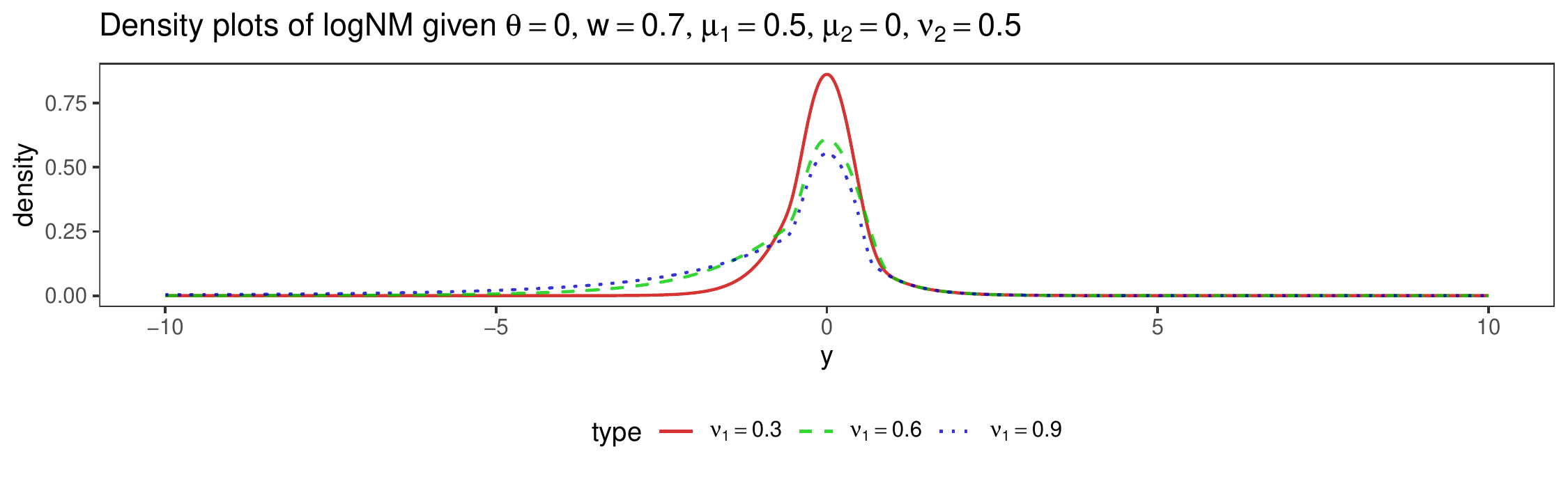}
	\end{subfigure}
	\caption{\label{fig:logNM_pdfs} \small {\it Density plots of the logNM distribution given different combinations of parameter values.}}
\end{figure}

Practitioners can build a Bayesian modal linear regression model based on the logNM likelihood \eqref{eq:logNM.model} as follows:
\begin{equation}
	\begin{aligned}
		Y_i \mid \boldsymbol{X}_i, w,\boldsymbol{\beta},\boldsymbol{\xi}_1,\boldsymbol{\xi}_2 &\stackrel{\text{ind}}\sim \operatorname{logNM}\left(w,\boldsymbol{X}^{\top}_i \boldsymbol{\beta}, \mu_1, \nu_1, \mu_2, \nu_2 \right), \\
		w &\sim \operatorname{Uniform}(0,1), \\
		\nu_1, \nu_2 &\stackrel{\text{i.i.d}}\sim \text{InverseGamma} (1,1),\\
		\mu_1, \mu_2 &\stackrel{\text{i.i.d}}\sim \mathcal{N} (0,100^2),\\
		p(\boldsymbol{\beta}) &\propto \mathcal{N}_{p}\left(\boldsymbol{0},10^2 \times \boldsymbol{I}_{p\times p}\right),
	\end{aligned}
	\label{eq:logNM.model}
\end{equation}
where $\boldsymbol{I}_{p\times p}$ stands for the $p$ by $p$ identity matrix and $\boldsymbol{\beta}$ is a $p$-dimension random vector. 
If one wishes to use a flat prior $p(\boldsymbol{\beta}) \propto 1$, then one must verify the sufficient condition in Theorem \ref{thm:suff2} of the main article. 

If researchers have another right-skewed or a left-skewed distribution to work with, then they can mimic the construction of the logNM above to propose a member of the GUD family that works for their applications. For example, one can use the reparameterized unimodal right-skewed Gamma distribution from \citet{bourguignon2020parametric_s} to construct a new type I GUD and a corresponding modal regression model. 

In the following lemma, we show why flat priors \textit{cannot} be used for $\mu_1$ and(or) $\mu_2$ in the Bayesian modal regression model based on the logNM distribution. This simple example demonstrates why we should \textit{avoid} placing improper priors on the \textit{non}-location parameters in Bayesian modal regression models based on the GUD family. 

\begin{proposition} \label{thm:improper.posterior}
	Endowing $\mu_1$ and(or) $\mu_2$ with flat priors $p(\mu_1) \propto 1$ and(or) $p(\mu_2) \propto 1$ leads to an improper posterior distribution under the logNM model \eqref{eq:logNM.model}.
\end{proposition}

\begin{proof}
	We want to show that
	$$
	\int_{-\infty}^{+\infty}\prod_{i=1}^{n}f_{\operatorname{logNM}}\left(y_i \mid w, \theta, \sigma_1, \sigma_2, \mu_1, \mu_2 \right)
	d\mu_1 = +\infty,
	$$
	and(or),
	$$
	\int_{-\infty}^{+\infty}\prod_{i=1}^{n}f_{\operatorname{logNM}}\left(y_i \mid w, \theta, \sigma_1, \sigma_2, \mu_1, \mu_2 \right)
	d\mu_2 = +\infty.
	$$
	Since 
	$$
	\begin{aligned}
		f_{\operatorname{logNM}}\left(y \mid w, \theta, \sigma_1, \sigma_2, \mu_1, \mu_2 \right) &= w f_{\operatorname{logN}}\left(\exp\left(\mu_1-\nu^2_1\right)-\left(y-\theta\right) \mid \theta, \mu_1, \nu_1 \right) + \\
		&\phantom{=} \left(1-w\right) f_{\operatorname{logN}}\left(\exp\left(\mu_2-\nu^2_2\right)+\left(y-\theta\right) \mid \theta, \mu_2, \nu_2 \right),
	\end{aligned}
	$$
	and any pdf must be nonnegative, it suffices to show that
	$$
	\int_{-\infty}^{+\infty} f_{\operatorname{logN}}\left(\exp\left(\mu_1-\nu^2_1\right)-\left(y-\theta\right) \mid \theta, \mu_1, \nu_1 \right) d\mu_2 = \infty,
	$$
	and(or)
	$$
	\int_{-\infty}^{+\infty} f_{\operatorname{logN}}\left(\exp\left(\mu_2-\nu^2_2\right)+\left(y-\theta\right) \mid \theta, \mu_2, \nu_2 \right) d\mu_1 = \infty.
	$$
	Both the integrals above are non-finite. This completes the proof. 
\end{proof}

Following the same arguments as those in Proposition \ref{thm:improper.posterior}, one can show that improper priors such as $p(\nu_1) \propto 1/\nu_1$ and(or) $p(\nu_2) \propto 1/\nu_2$ will also lead to an improper posterior distribution. As stated in Proposition \ref{thm:improper.posterior}, a general rule is that, for the Bayesian modal regression models based on the GUD family, using improper prior(s) for any parameter in $\left(\boldsymbol{\xi}_1 \cup \boldsymbol{\xi}_2\right) \backslash \left(\boldsymbol{\xi}_1 \cap \boldsymbol{\xi}_2\right)$ leads to an improper posterior distribution. Here, $A \backslash B=A \cap B^c$ denotes a collection of elements in $A$ but not in $B$. 

\section{A short note about MCMC} \label{sec:note_on_MCMC}

Readers who are familiar with Bayesian modeling of mixture distributions may wonder why we do not use the data augmentation ``trick'' to design a specific MCMC algorithm for modal regression models based on the GUD family. The problem lies in the type II distributions of the GUD family. If $\mathcal{D}_1 \cap \mathcal{D}_2 = \varnothing$, then the latent variable conditional on other parameters and observed data becomes a degenerate random variable. This degenerate random variable will not behave randomly. 

To demonstrate that we have a degenerate random variable, let us consider a simple case with a single observation. Recall that the type II GUD has the pdf
$$
f\left(y \mid w, \theta, \boldsymbol{\xi}_1, \boldsymbol{\xi}_2\right)=w f_1\left(y \mid \theta, \boldsymbol{\xi}_1\right) I\left(y < \theta\right)+(1-w) f_2\left(y \mid \theta, \boldsymbol{\xi}_2\right)I\left(y \ge \theta\right) .
$$
Introducing the latent variable $z$, we have the joint pdf as
$$
\begin{aligned}
	f\left(y, z \mid w, \theta, \boldsymbol{\xi}_1, \boldsymbol{\xi}_2\right)= & {\left[w f_1\left(y \mid \theta, \boldsymbol{\xi}_1\right) I(y<\theta)\right]^z } \\
	& {\left[(1-w) f_2\left(y \mid \theta, \boldsymbol{\xi}_2\right) I(y \geq \theta)\right]^{1-z}. }
\end{aligned}
$$
The conditional distribution of $z$ is then
$$
\begin{aligned}
	p\left(z \mid w, \theta, \boldsymbol{\xi}_1, \boldsymbol{\xi}_2, y\right) &\propto f\left(y,z \mid w, \theta, \boldsymbol{\xi}_1, \boldsymbol{\xi}_2\right) \\
	&\sim \operatorname{Bernoulli}\left(r\right),
\end{aligned}
$$
where $$r = \frac{w f_1\left(y \mid \theta, \boldsymbol{\xi}_1\right) I\left(y < \theta\right)}{w f_1\left(y \mid \theta, \boldsymbol{\xi}_1\right) I\left(y < \theta\right) + (1-w) f_2\left(y \mid \theta, \boldsymbol{\xi}_2\right)I\left(y \ge \theta\right)}.$$
Similarly, the conditional distribution of $\theta$ is
$$
p\left(\theta \mid w, \boldsymbol{\xi}_1, \boldsymbol{\xi}_2, y,z\right) \propto f\left(y,z \mid w, \theta, \boldsymbol{\xi}_1, \boldsymbol{\xi}_2\right) p\left(\theta\right).
$$
The conditional mean of the latent variable $z$ can only be 0 \textit{or} 1 since $y < \theta$ and $y \ge \theta$ cannot happen at the same time. Hence, the latent variable becomes a degenerate random variable during the MCMC iterations. Without loss of generality, let us assume $z = 1$ in the first iteration of the MCMC. It is not hard to see that, during the MCMC iterations, the updated values of $\theta$ can only be larger than $y$. This leads to $z$ being equal to 1 for the rest of the iterations in the MCMC algorithm. However, if the true value of $\theta$ is smaller than $y$, then the MCMC chain will never reach the true value. Similarly, if $z = 0$ in the first iteration of the MCMC, then all updated values of $\theta$ can only be smaller than $y$. {\revise In other words, the latent variable data augmentation algorithm will \textit{not} explore the \textit{whole} parameter space.} {\revisetwo A Markov chain is irreducible if it is possible to eventually get from one state to any other state in a finite number of steps. Conversely, a Markov chain is reducible if it is impossible to eventually get from one to any other state in a finite number of steps. Hence, the latent variable data augmentation algorithm is \textit{reducible}, rendering it much less practical. A practical MCMC algorithm must be able to explore the entire state space regardless of the initialization.}

{\revise In the case with $n > 1$ observations, to show that the the latent variable data augmentation algorithm is \textit{reducible}, it suffices to provide one legitimate counterexample. This example will illustrate that the latent variable data augmentation algorithm fails to explore the whole parameter space. Denote $\boldsymbol{\theta} = \left\{\theta_{j}, 1 \le j \le n\right\}$ as the collection of all location/mode parameters and $\theta_{\operatorname{min}}$ as the smallest one. In the first iteration of the MCMC, suppose $\boldsymbol{z} = \left[z_{1},\dots,z_{n}\right]^{\top} = \left[1,\dots,1\right]^{\top}$. Then the joint pdf becomes
	$$
	\prod_{i=1}^{n} f\left(y_{i}, z_{i} = 1 \mid w, \boldsymbol{\theta}, \boldsymbol{\xi}_1, \boldsymbol{\xi}_2\right) = w^{n} \left\{\prod_{i=1}^{n} f_1\left(y_{i} \mid \boldsymbol{\theta}, \boldsymbol{\xi}_1\right)\right\} I\left(y_{(n)}<\theta_{\operatorname{min}}\right),
	$$
	where $y_{(n)}$ represents the largest sample order statistic. As a result, the support of conditional distribution $f\left(\theta_{j} \mid \boldsymbol{y}, w, \boldsymbol{\xi}_{1},\boldsymbol{\xi}_{2},\boldsymbol{z}\right)$ is truncated to $\left(y_{(n)},+\infty\right)$ for all $j = 1,\dots,n$. During the MCMC iterations, the updated values of $\theta_{j}$ will always be larger than $y_{(n)}$. This leads to $\boldsymbol{z} = \left[1,\dots,1\right]^{\top}$ and $\theta_{j} > y_{(n)}$ for all $ j = 1,\dots,n$, for the rest of the iterations in the MCMC algorithm. Therefore, in the case with $n > 1$, the latent variable data augmentation algorithm is \textit{reducible}.} {\revisetwo This is not ideal since a poor initialization leads to a degenerate random variable. Practical MCMC algorithms should be able to explore the entire parameter space \emph{regardless} of their initial values.} 

{\revisetwo To demonstrate that the latent variable data augmentation algorithm is \textit{reducible}, we designed a simulation study. We generated 100 observations from the TPSC-Student-$t$ distribution with the mode parameter $\theta = 0$, $w = 0.4$, $\sigma = 1$, and $\delta = 5$. Without loss of generality, we assume $\theta$ is the only unknown parameter. We use a non-informative normal prior with mean 0 and variance 10,000, which is numerically equivalent to a flat prior. We choose the elliptical slice sampling algorithm for the sampling of $\theta$ since we use a normal prior on $\theta$ \citep{murray2010elliptical}. 
	
	We present the traceplots from the latent variable data augmentation algorithm in Figure \ref{fig:latent_variable_MCMC}. From these traceplots, we can see that the latent variable data augmentation algorithm is extremely sensitive to the initial values of $\theta$. In the simulated data, the maximum value is 7.4083, and the minimum value is $-4.4288$. If we choose an initial value of 8, larger than the maximum value, the updated values of $\theta$ are all larger than the maximum value of the simulated data, as shown in the top panel. If we choose an initial value of $-5$, less than the minimum value, all updated values of $\theta$ are smaller than the minimum value of the simulated data, as shown in the bottom panel. Finally, if we choose an initial value of 1, which is within the range of the simulated data, the MCMC chain never captures the true value of $\theta = 0$, as shown in the middle panel. Figure \ref{fig:latent_variable_MCMC} clearly indicates that the latent variable data augmentation algorithm is highly dependent on the initial value of $\theta$. This sensitivity leads to poor mixing and a failure to explore the parameter space adequately, thereby confirming that the algorithm is \textit{reducible}. This reducibility implies that the algorithm fails to properly converge to the true posterior distribution of $\theta$. }

In conclusion, for Bayesian modal regression models based on the type II GUD subfamily, the latent variable data augmentation algorithm is \textit{reducible}. As a consequence of this, we do \textit{not} use the data augmentation ``trick'' for mixture models in our MCMC algorithm. Instead, we use the No-U-Turn Sampler implemented in the \textsc{Stan} software.

\begin{figure}
	\centering
	\includegraphics[width = \textwidth]{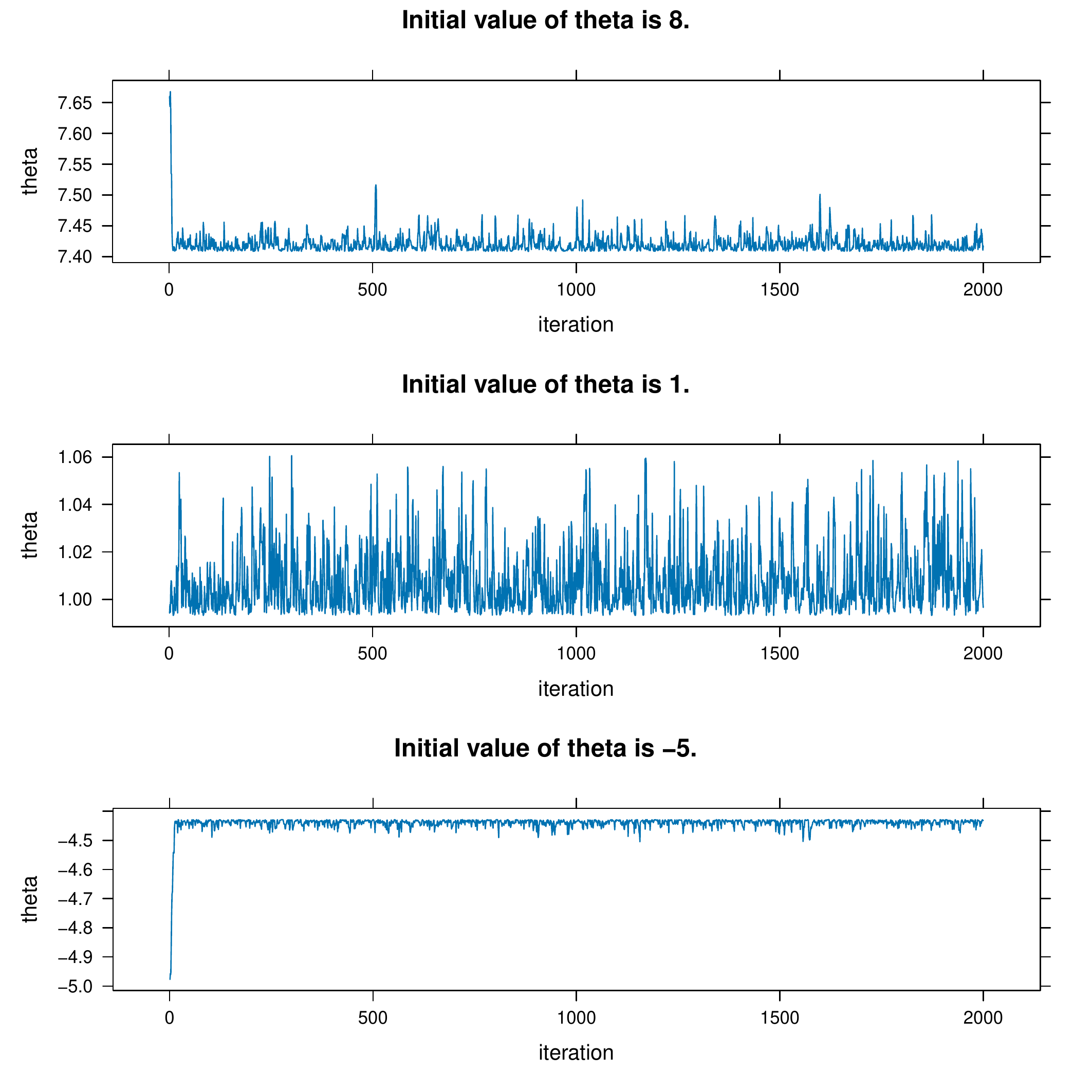}
	\caption{\label{fig:latent_variable_MCMC}Traceplots demonstrating that the latent variable data augmentation algorithm is reducible.} 
\end{figure}

\section{Convergence Diagnostics for the Real Data Applications and Simulation Studies} \label{Sec:Convergence_Diagnostics}
In this section, we include more details about posterior inference, convergence diagnostics, and traceplots for the four data application examples and two simulation studies from the main manuscript. The rhat, which has the theoretical minimum value as 1, is a statistic measuring the convergence of the MCMC chains. To obtain reliable posterior inference, it is recommended that rhat should be near 1 or at least less than 1.1 (page 287 of \cite{gelman2013bayesian_s}). The ess\_bulk and ess\_tail are the bulk and tail effective sample size respectively. The ess\_tail is defined as the minimum of the effective sample sizes for the 5\% and 95\% quantiles. The recommended lower threshold for ess\_bulk and ess\_tail is 400 \citep{vehtari2021rank}. All of the rhat, ess\_bulk and ess\_tail summary statistics in Tables \ref{tab:summary_bank_deposit} through \ref{tab:summary_simu_left_skewed_TPSC_one_data}, meet the recommended thresholds. This suggests that our analyses in the main manuscript are free from convergence issues.

All of the traceplots, from Figure \ref{fig:traceplot_bank} to Figure \ref{fig:left_simu}, presented in this section also indicate that the MCMC chains mixed well, with no convergence issues. We ran four MCMC chains for each model that was fit and used the combined MCMC samples to approximate the posterior distributions. For the intercept-only regression model fit to the bank deposits data (Section 2.1 of the main manuscript), we set the number of warmup iterations as $10{,}000$ and the number of post-warmup iterations as $20{,}000$ for each chain. For all other models in the main article, we set the number of warmup iterations as $10{,}000$ and the number of post-warmup iterations as $10{,}000$ for each of the four MCMC chains.


\begin{table}
	\centering
	\caption{\label{tab:summary_bank_deposit}Summary statistics and diagnostics for modal regression based on the DTP-Student-$t$ likelihood in bank deposits application.}
	\begin{tabular}[t]{lrrrrrrrrr}
		\toprule
		variable & mean & median & sd & mad & q5 & q95 & rhat & ess\_bulk & ess\_tail\\
		\midrule
		$\sigma_{1}$ & 1.20 & 1.08 & 0.56 & 0.44 & 0.57 & 2.26 & 1.00 & 51884.27 & 43430.38\\
		$\sigma_{2}$ & 20.52 & 19.98 & 5.36 & 5.22 & 12.72 & 30.15 & 1.00 & 48958.33 & 49223.01\\
		$\delta_{1}$ & 1.34 & 1.16 & 0.74 & 0.50 & 0.59 & 2.68 & 1.00 & 71896.15 & 47798.76\\
		$\delta_{2}$ & 0.92 & 0.89 & 0.21 & 0.20 & 0.62 & 1.30 & 1.00 & 54497.18 & 54720.92\\
		$\theta$ & 18.72 & 18.69 & 1.80 & 1.68 & 15.92 & 21.62 & 1.00 & 43742.48 & 39115.74\\
		\bottomrule
	\end{tabular}
\end{table}


\begin{table}
	\centering
	\caption{Summary statistics and diagnostics for mean regression based on the normal likelihood in the crime rate application, including D.C.}
	\begin{tabular}[t]{lrrrrrrrrr}
		\toprule
		variable & mean & median & sd & mad & q5 & q95 & rhat & ess\_bulk & ess\_tail\\
		\midrule
		$\beta_{0}$ & -24.30 & -24.32 & 5.35 & 5.22 & -33.05 & -15.46 & 1.00 & 16406.58 & 19021.49\\
		$\beta_{1}$ & 0.47 & 0.47 & 0.16 & 0.16 & 0.20 & 0.74 & 1.00 & 18216.70 & 21480.39\\
		$\beta_{2}$ & 1.14 & 1.14 & 0.23 & 0.22 & 0.77 & 1.52 & 1.00 & 19775.85 & 22853.66\\
		$\beta_{3}$ & 0.07 & 0.07 & 0.03 & 0.03 & 0.01 & 0.12 & 1.00 & 23191.42 & 22530.73\\
		$\sigma$ & 4.50 & 4.46 & 0.47 & 0.45 & 3.80 & 5.34 & 1.00 & 21872.50 & 22509.97\\
		\bottomrule
	\end{tabular}
\end{table}

\begin{table}
	\centering
	\caption{Summary statistics and diagnostics for mean regression based on the normal likelihood in the crime rate application, excluding D.C.}
	\begin{tabular}[t]{lrrrrrrrrr}
		\toprule
		variable & mean & median & sd & mad & q5 & q95 & rhat & ess\_bulk & ess\_tail\\
		\midrule
		$\beta_{0}$ & -0.46 & -0.49 & 3.01 & 3.01 & -5.41 & 4.48 & 1.00 & 15995.36 & 18893.26\\
		$\beta_{1}$ & -0.13 & -0.13 & 0.09 & 0.09 & -0.27 & 0.02 & 1.00 & 17410.41 & 20836.56\\
		$\beta_{2}$ & 0.35 & 0.36 & 0.12 & 0.12 & 0.16 & 0.55 & 1.00 & 18782.72 & 22207.49\\
		$\beta_{3}$ & 0.06 & 0.06 & 0.02 & 0.02 & 0.04 & 0.09 & 1.00 & 23461.81 & 21966.14\\
		$\sigma$ & 2.04 & 2.02 & 0.22 & 0.21 & 1.72 & 2.43 & 1.00 & 23786.27 & 22309.52\\
		\bottomrule
	\end{tabular}
\end{table}

\begin{table}
	\centering
	\caption{Summary statistics and diagnostics for median regression based on the ALD likelihood in the crime rate application, including D.C.}
	\begin{tabular}[t]{lrrrrrrrrr}
		\toprule
		variable & mean & median & sd & mad & q5 & q95 & rhat & ess\_bulk & ess\_tail\\
		\midrule
		$\beta_0$ & -1.36 & -1.14 & 3.48 & 3.38 & -7.44 & 3.97 & 1.00 & 9215.11 & 9934.78\\
		$\beta_1$ & -0.12 & -0.12 & 0.10 & 0.10 & -0.27 & 0.05 & 1.00 & 10087.62 & 12198.44\\
		$\beta_2$ & 0.44 & 0.43 & 0.14 & 0.13 & 0.21 & 0.68 & 1.00 & 10433.68 & 11002.33\\
		$\beta_3$ & 0.06 & 0.06 & 0.02 & 0.02 & 0.03 & 0.08 & 1.00 & 17077.75 & 17352.45\\
		$\sigma$ & 1.15 & 1.13 & 0.17 & 0.16 & 0.90 & 1.44 & 1.00 & 16659.49 & 20246.21\\
		\bottomrule
	\end{tabular}
\end{table}

\begin{table}
	\centering
	\caption{Summary statistics and diagnostics for median regression based on the ALD likelihood in the crime rate application, excluding D.C.}
	\begin{tabular}[t]{lrrrrrrrrr}
		\toprule
		variable & mean & median & sd & mad & q5 & q95 & rhat & ess\_bulk & ess\_tail\\
		\midrule
		$\beta_0$ & 1.22 & 1.32 & 2.49 & 2.49 & -2.96 & 5.14 & 1.00 & 9951.97 & 12600.58\\
		$\beta_1$ & -0.18 & -0.19 & 0.07 & 0.07 & -0.29 & -0.06 & 1.00 & 11125.90 & 14798.56\\
		$\beta_2$ & 0.35 & 0.36 & 0.10 & 0.10 & 0.18 & 0.52 & 1.00 & 11741.02 & 15603.53\\
		$\beta_3$ & 0.05 & 0.05 & 0.01 & 0.01 & 0.03 & 0.08 & 1.00 & 17263.28 & 18630.85\\
		$\sigma$ & 0.80 & 0.79 & 0.12 & 0.11 & 0.63 & 1.01 & 1.00 & 19247.22 & 19584.58\\
		\bottomrule
	\end{tabular}
\end{table}

\begin{table}
	\centering
	\caption{Summary statistics and diagnostics for modal regression based on the TPSC-Student-$t$ likelihood in the crime rate application, including D.C.}
	\begin{tabular}[t]{lrrrrrrrrr}
		\toprule
		variable & mean & median & sd & mad & q5 & q95 & rhat & ess\_bulk & ess\_tail\\
		\midrule
		$\beta_{0}$ & 1.16 & 1.29 & 2.69 & 2.64 & -3.49 & 5.35 & 1.00 & 13537.93 & 17014.18\\
		$\beta_{1}$ & -0.20 & -0.20 & 0.08 & 0.08 & -0.33 & -0.06 & 1.00 & 14052.58 & 17556.68\\
		$\beta_{2}$ & 0.24 & 0.25 & 0.14 & 0.14 & 0.01 & 0.46 & 1.00 & 9133.10 & 13661.34\\
		$\beta_{3}$ & 0.06 & 0.06 & 0.02 & 0.01 & 0.04 & 0.09 & 1.00 & 10599.66 & 8941.30\\
		$w$ & 0.28 & 0.28 & 0.12 & 0.13 & 0.08 & 0.47 & 1.00 & 5557.27 & 3737.66\\
		$\sigma$ & 1.17 & 1.16 & 0.26 & 0.25 & 0.76 & 1.61 & 1.00 & 7724.39 & 3906.00\\
		$\delta$ & 1.87 & 1.75 & 0.64 & 0.56 & 1.07 & 3.07 & 1.00 & 12601.42 & 15273.52\\
		\bottomrule
	\end{tabular}
\end{table}

\begin{table}
	\centering
	\caption{Summary statistics and diagnostics for modal regression based on the TPSC-Student-$t$ likelihood in the crime rate application, excluding D.C.}
	\begin{tabular}[t]{lrrrrrrrrr}
		\toprule
		variable & mean & median & sd & mad & q5 & q95 & rhat & ess\_bulk & ess\_tail\\
		\midrule
		$\beta_{0}$ & 1.16 & 1.24 & 2.63 & 2.61 & -3.32 & 5.35 & 1.00 & 14330.69 & 19114.85\\
		$\beta_{1}$ & -0.20 & -0.20 & 0.08 & 0.08 & -0.32 & -0.06 & 1.00 & 15117.09 & 19602.91\\
		$\beta_{2}$ & 0.24 & 0.25 & 0.13 & 0.14 & 0.01 & 0.45 & 1.00 & 12220.30 & 17676.54\\
		$\beta_{3}$ & 0.06 & 0.06 & 0.02 & 0.01 & 0.04 & 0.09 & 1.00 & 11499.69 & 8564.30\\
		$w$ & 0.30 & 0.31 & 0.13 & 0.13 & 0.08 & 0.50 & 1.00 & 5437.73 & 2386.75\\
		$\sigma$ & 1.27 & 1.26 & 0.27 & 0.27 & 0.83 & 1.73 & 1.00 & 6497.28 & 2682.18\\
		$\delta$ & 2.99 & 2.59 & 1.59 & 1.09 & 1.37 & 5.97 & 1.00 & 10568.04 & 5927.78\\
		\bottomrule
	\end{tabular}
\end{table}


\begin{table}
	\centering
	\caption{Summary statistics and diagnostics for mean regression based on the normal likelihood in the Boston house price application.}
	\begin{tabular}[t]{lrrrrrrrrr}
		\toprule
		variable & mean & median & sd & mad & q5 & q95 & rhat & ess\_bulk & ess\_tail\\
		\midrule
		$\beta_{0}$ & 36.46 & 36.46 & 5.07 & 5.03 & 28.12 & 44.81 & 1.00 & 18019.29 & 23938.34\\
		$\beta_{1}$ & -0.11 & -0.11 & 0.03 & 0.03 & -0.16 & -0.05 & 1.00 & 42660.60 & 29663.24\\
		$\beta_{2}$ & 0.05 & 0.05 & 0.01 & 0.01 & 0.02 & 0.07 & 1.00 & 35258.34 & 29918.50\\
		$\beta_{3}$ & 0.02 & 0.02 & 0.06 & 0.06 & -0.08 & 0.12 & 1.00 & 34526.72 & 30190.27\\
		$\beta_{4}$ & 2.69 & 2.69 & 0.86 & 0.86 & 1.27 & 4.11 & 1.00 & 48946.81 & 28847.94\\
		$\beta_{5}$ & -17.80 & -17.82 & 3.81 & 3.81 & -24.06 & -11.53 & 1.00 & 28728.15 & 29107.00\\
		$\beta_{6}$ & 3.81 & 3.81 & 0.41 & 0.41 & 3.12 & 4.49 & 1.00 & 23280.57 & 28539.14\\
		$\beta_{7}$ & 0.00 & 0.00 & 0.01 & 0.01 & -0.02 & 0.02 & 1.00 & 36062.59 & 29788.64\\
		$\beta_{8}$ & -1.48 & -1.48 & 0.20 & 0.20 & -1.81 & -1.15 & 1.00 & 30599.28 & 29810.00\\
		$\beta_{9}$ & 0.31 & 0.31 & 0.07 & 0.07 & 0.20 & 0.42 & 1.00 & 23921.39 & 27946.68\\
		$\beta_{10}$ & -0.01 & -0.01 & 0.00 & 0.00 & -0.02 & -0.01 & 1.00 & 25027.57 & 28529.62\\
		$\beta_{11}$ & -0.95 & -0.95 & 0.13 & 0.13 & -1.17 & -0.74 & 1.00 & 27126.23 & 28973.83\\
		$\beta_{12}$ & 0.01 & 0.01 & 0.00 & 0.00 & 0.00 & 0.01 & 1.00 & 50871.21 & 30534.77\\
		$\beta_{13}$ & -0.52 & -0.52 & 0.05 & 0.05 & -0.61 & -0.44 & 1.00 & 30284.46 & 31537.54\\
		$\sigma$ & 4.75 & 4.75 & 0.15 & 0.15 & 4.51 & 5.01 & 1.00 & 48680.55 & 30354.76\\
		\bottomrule
	\end{tabular}
\end{table}

\begin{table}
	\centering
	\caption{Summary statistics and diagnostics for mean regression based on the SNCP likelihood in the Boston house price application.}
	\begin{tabular}[t]{lrrrrrrrrr}
		\toprule
		variable & mean & median & sd & mad & q5 & q95 & rhat & ess\_bulk & ess\_tail\\
		\midrule
		$\beta_{0}$ & 35.73 & 35.72 & 4.34 & 4.32 & 28.60 & 42.90 & 1.00 & 16772.00 & 22788.90\\
		$\beta_{1}$ & -0.12 & -0.11 & 0.03 & 0.03 & -0.17 & -0.07 & 1.00 & 41731.91 & 27005.57\\
		$\beta_{2}$ & 0.03 & 0.03 & 0.01 & 0.01 & 0.01 & 0.05 & 1.00 & 33778.96 & 29274.33\\
		$\beta_{3}$ & 0.03 & 0.03 & 0.05 & 0.05 & -0.05 & 0.11 & 1.00 & 36229.13 & 28328.01\\
		$\beta_{4}$ & 1.38 & 1.39 & 0.72 & 0.72 & 0.17 & 2.54 & 1.00 & 40843.75 & 29814.52\\
		$\beta_{5}$ & -14.02 & -13.99 & 3.28 & 3.31 & -19.48 & -8.69 & 1.00 & 26542.15 & 28558.83\\
		$\beta_{6}$ & 3.42 & 3.42 & 0.40 & 0.40 & 2.77 & 4.08 & 1.00 & 22959.69 & 28660.09\\
		$\beta_{7}$ & -0.02 & -0.02 & 0.01 & 0.01 & -0.03 & 0.00 & 1.00 & 34025.64 & 29864.72\\
		$\beta_{8}$ & -1.12 & -1.12 & 0.18 & 0.18 & -1.42 & -0.83 & 1.00 & 28615.18 & 29609.14\\
		$\beta_{9}$ & 0.24 & 0.24 & 0.06 & 0.06 & 0.14 & 0.34 & 1.00 & 21697.75 & 27647.14\\
		$\beta_{10}$ & -0.01 & -0.01 & 0.00 & 0.00 & -0.02 & -0.01 & 1.00 & 26057.74 & 28254.44\\
		$\beta_{11}$ & -0.82 & -0.82 & 0.11 & 0.11 & -1.00 & -0.64 & 1.00 & 27037.86 & 28634.89\\
		$\beta_{12}$ & 0.01 & 0.01 & 0.00 & 0.00 & 0.00 & 0.01 & 1.00 & 49887.11 & 29362.26\\
		$\beta_{13}$ & -0.43 & -0.43 & 0.05 & 0.05 & -0.51 & -0.35 & 1.00 & 28554.45 & 28327.71\\
		$\sigma$ & 4.54 & 4.53 & 0.16 & 0.16 & 4.29 & 4.80 & 1.00 & 40386.94 & 30394.13\\
		$\gamma_{1}$ & 0.76 & 0.76 & 0.05 & 0.05 & 0.67 & 0.83 & 1.00 & 34071.14 & 31212.56\\
		\bottomrule
	\end{tabular}
\end{table}

\begin{table}
	\centering
	\caption{Summary statistics and diagnostics for median regression based on the ALD likelihood in the Boston house price application.}
	\begin{tabular}[t]{lrrrrrrrrr}
		\toprule
		variable & mean & median & sd & mad & q5 & q95 & rhat & ess\_bulk & ess\_tail\\
		\midrule
		$\beta_{0}$ & 14.91 & 14.90 & 4.61 & 4.60 & 7.32 & 22.48 & 1.00 & 11096.61 & 17659.41\\
		$\beta_{1}$ & -0.12 & -0.12 & 0.03 & 0.03 & -0.16 & -0.06 & 1.00 & 25674.22 & 23664.56\\
		$\beta_{2}$ & 0.04 & 0.04 & 0.01 & 0.01 & 0.02 & 0.05 & 1.00 & 23970.41 & 27121.52\\
		$\beta_{3}$ & 0.01 & 0.01 & 0.04 & 0.04 & -0.05 & 0.07 & 1.00 & 28424.98 & 26133.02\\
		$\beta_{4}$ & 1.49 & 1.48 & 0.62 & 0.61 & 0.51 & 2.53 & 1.00 & 31081.28 & 24766.42\\
		$\beta_{5}$ & -8.99 & -9.01 & 2.83 & 2.85 & -13.64 & -4.36 & 1.00 & 18306.15 & 24468.30\\
		$\beta_{6}$ & 5.26 & 5.26 & 0.46 & 0.46 & 4.50 & 6.01 & 1.00 & 13054.10 & 20692.65\\
		$\beta_{7}$ & -0.03 & -0.03 & 0.01 & 0.01 & -0.04 & -0.01 & 1.00 & 18567.66 & 23782.98\\
		$\beta_{8}$ & -0.99 & -0.99 & 0.15 & 0.15 & -1.23 & -0.75 & 1.00 & 22202.34 & 25884.79\\
		$\beta_{9}$ & 0.18 & 0.18 & 0.05 & 0.05 & 0.09 & 0.26 & 1.00 & 16706.88 & 23384.36\\
		$\beta_{10}$ & -0.01 & -0.01 & 0.00 & 0.00 & -0.01 & -0.01 & 1.00 & 19296.89 & 23999.50\\
		$\beta_{11}$ & -0.75 & -0.75 & 0.09 & 0.09 & -0.90 & -0.59 & 1.00 & 19821.84 & 24185.39\\
		$\beta_{12}$ & 0.01 & 0.01 & 0.00 & 0.00 & 0.01 & 0.02 & 1.00 & 31154.75 & 26507.97\\
		$\beta_{13}$ & -0.31 & -0.31 & 0.05 & 0.05 & -0.39 & -0.23 & 1.00 & 16869.67 & 23558.08\\
		$\sigma$ & 1.57 & 1.56 & 0.07 & 0.07 & 1.45 & 1.69 & 1.00 & 34414.60 & 27521.02\\
		\bottomrule
	\end{tabular}
\end{table}

\begin{table}
	\centering
	\caption{Summary statistics and diagnostics for modal regression based on the TPSC-Student-$t$ likelihood in the Boston house price application.}
	\begin{tabular}[t]{lrrrrrrrrr}
		\toprule
		variable & mean & median & sd & mad & q5 & q95 & rhat & ess\_bulk & ess\_tail\\
		\midrule
		$\beta_{0}$ & 13.02 & 13.01 & 4.01 & 4.01 & 6.48 & 19.66 & 1.00 & 17667.44 & 23088.01\\
		$\beta_{1}$ & -0.13 & -0.13 & 0.02 & 0.02 & -0.17 & -0.09 & 1.00 & 30697.31 & 23195.19\\
		$\beta_{2}$ & 0.02 & 0.02 & 0.01 & 0.01 & 0.01 & 0.04 & 1.00 & 32180.06 & 28382.14\\
		$\beta_{3}$ & 0.01 & 0.01 & 0.03 & 0.03 & -0.05 & 0.06 & 1.00 & 29251.66 & 27875.97\\
		$\beta_{4}$ & 1.44 & 1.43 & 0.57 & 0.57 & 0.50 & 2.39 & 1.00 & 37386.87 & 30230.27\\
		$\beta_{5}$ & -6.71 & -6.68 & 2.42 & 2.40 & -10.73 & -2.79 & 1.00 & 24700.28 & 26558.72\\
		$\beta_{6}$ & 4.87 & 4.87 & 0.46 & 0.47 & 4.11 & 5.63 & 1.00 & 20440.38 & 25965.38\\
		$\beta_{7}$ & -0.04 & -0.04 & 0.01 & 0.01 & -0.05 & -0.02 & 1.00 & 26327.10 & 27034.56\\
		$\beta_{8}$ & -0.89 & -0.88 & 0.14 & 0.14 & -1.11 & -0.66 & 1.00 & 25888.35 & 27294.97\\
		$\beta_{9}$ & 0.14 & 0.14 & 0.04 & 0.04 & 0.07 & 0.21 & 1.00 & 24276.86 & 27184.25\\
		$\beta_{10}$ & -0.01 & -0.01 & 0.00 & 0.00 & -0.01 & -0.01 & 1.00 & 27362.14 & 27458.27\\
		$\beta_{11}$ & -0.61 & -0.61 & 0.08 & 0.08 & -0.73 & -0.48 & 1.00 & 32141.78 & 27873.60\\
		$\beta_{12}$ & 0.01 & 0.01 & 0.00 & 0.00 & 0.01 & 0.01 & 1.00 & 33646.78 & 28939.59\\
		$\beta_{13}$ & -0.28 & -0.28 & 0.04 & 0.04 & -0.35 & -0.21 & 1.00 & 23692.29 & 27306.82\\
		$w$ & 0.28 & 0.28 & 0.03 & 0.03 & 0.24 & 0.33 & 1.00 & 32636.21 & 29899.86\\
		$\sigma$ & 2.13 & 2.13 & 0.16 & 0.15 & 1.88 & 2.39 & 1.00 & 25937.22 & 26901.12\\
		$\delta$ & 2.24 & 2.21 & 0.31 & 0.29 & 1.79 & 2.79 & 1.00 & 28284.81 & 27224.91\\
		\bottomrule
	\end{tabular}
\end{table}


\begin{table}
	\centering
	\caption{Summary statistics and diagnostics for mean regression based on the normal likelihood in the serum data application.}
	\begin{tabular}[t]{lrrrrrrrrr}
		\toprule
		variable & mean & median & sd & mad & q5 & q95 & rhat & ess\_bulk & ess\_tail\\
		\midrule
		$\beta_{0}$ & 3.09 & 3.09 & 0.39 & 0.39 & 2.46 & 3.73 & 1.00 & 10579.95 & 14402.58\\
		$\beta_{1}$ & 0.96 & 0.97 & 0.31 & 0.31 & 0.44 & 1.47 & 1.00 & 9666.15 & 12217.35\\
		$\beta_{2}$ & -0.05 & -0.05 & 0.05 & 0.05 & -0.13 & 0.04 & 1.00 & 10065.81 & 13371.25\\
		$\sigma$ & 1.97 & 1.97 & 0.08 & 0.08 & 1.84 & 2.10 & 1.00 & 17463.42 & 19268.55\\
		\bottomrule
	\end{tabular}
\end{table}

\begin{table}
	\centering
	\caption{Summary statistics and diagnostics for median regression based on the ALD likelihood in the serum data application.}
	\begin{tabular}[t]{lrrrrrrrrr}
		\toprule
		variable & mean & median & sd & mad & q5 & q95 & rhat & ess\_bulk & ess\_tail\\
		\midrule
		$\beta_{0}$ & 2.81 & 2.80 & 0.44 & 0.45 & 2.12 & 3.55 & 1.00 & 8415.06 & 12185.42\\
		$\beta_{1}$ & 1.12 & 1.13 & 0.35 & 0.35 & 0.54 & 1.69 & 1.00 & 7983.74 & 11084.79\\
		$\beta_{2}$ & -0.07 & -0.07 & 0.06 & 0.06 & -0.16 & 0.03 & 1.00 & 8351.57 & 11698.45\\
		$\sigma$ & 0.78 & 0.77 & 0.05 & 0.04 & 0.70 & 0.85 & 1.00 & 13398.79 & 15917.55\\
		\bottomrule
	\end{tabular}
\end{table}

\begin{table}
	\centering
	\caption{Summary statistics and diagnostics for modal regression based on the FG likelihood in the serum data application.}
	\begin{tabular}[t]{lrrrrrrrrr}
		\toprule
		variable & mean & median & sd & mad & q5 & q95 & rhat & ess\_bulk & ess\_tail\\
		\midrule
		$\beta_{0}$ & 2.37 & 2.37 & 0.32 & 0.31 & 1.85 & 2.89 & 1.00 & 12344.81 & 17775.03\\
		$\beta_{1}$ & 1.15 & 1.15 & 0.26 & 0.26 & 0.72 & 1.59 & 1.00 & 10944.28 & 14809.64\\
		$\beta_{2}$ & -0.11 & -0.11 & 0.04 & 0.04 & -0.18 & -0.03 & 1.00 & 11537.75 & 16294.24\\
		$w$ & 0.06 & 0.04 & 0.07 & 0.04 & 0.00 & 0.19 & 1.00 & 11913.26 & 13666.64\\
		$\sigma_1$ & 1.82 & 1.30 & 6.88 & 0.56 & 0.46 & 3.48 & 1.00 & 14689.01 & 9593.03\\
		$\sigma_2$ & 1.68 & 1.68 & 0.09 & 0.09 & 1.54 & 1.82 & 1.00 & 20297.65 & 19120.22\\
		\bottomrule
	\end{tabular}
\end{table}


\begin{table}
	\centering
	\caption{Summary statistics and diagnostics for mean regression based on the normal likelihood in the left-skewed simulation study for one simulated data.}
	\begin{tabular}[t]{lrrrrrrrrr}
		\toprule
		variable & mean & median & sd & mad & q5 & q95 & rhat & ess\_bulk & ess\_tail\\
		\midrule
		$\beta_{0}$ & -2.33 & -2.33 & 2.36 & 2.30 & -6.18 & 1.54 & 1.00 & 34874.25 & 26421.25\\
		$\beta_{1}$ & 2.66 & 2.67 & 3.69 & 3.59 & -3.35 & 8.68 & 1.00 & 34343.20 & 26151.12\\
		$\sigma$ & 12.71 & 12.53 & 1.74 & 1.64 & 10.22 & 15.84 & 1.00 & 31547.10 & 26259.61\\
		\bottomrule
	\end{tabular}
\end{table}

\begin{table}
	\centering
	\caption{Summary statistics and diagnostics for mean regression based on the SNCP likelihood in the left-skewed simulation study for one simulated data.}
	\begin{tabular}[t]{lrrrrrrrrr}
		\toprule
		variable & mean & median & sd & mad & q5 & q95 & rhat & ess\_bulk & ess\_tail\\
		\midrule
		$\beta_{0}$ & -2.13 & -2.15 & 2.50 & 2.43 & -6.26 & 1.96 & 1.00 & 4748.24 & 5885.32\\
		$\beta_{1}$ & 2.88 & 2.91 & 3.82 & 3.72 & -3.41 & 9.11 & 1.00 & 5197.17 & 7262.71\\
		$\sigma$ & 13.34 & 13.14 & 1.91 & 1.83 & 10.63 & 16.81 & 1.00 & 3793.51 & 4521.24\\
		$\gamma_{1}$ & 0.08 & 0.06 & 0.08 & 0.06 & 0.00 & 0.24 & 1.00 & 3327.55 & 6325.35\\
		\bottomrule
	\end{tabular}
\end{table}

\begin{table}
	\centering
	\caption{Summary statistics and diagnostics for median regression based on the ALD likelihood in the left-skewed simulation study for one simulated data.}
	\begin{tabular}[t]{lrrrrrrrrr}
		\toprule
		variable & mean & median & sd & mad & q5 & q95 & rhat & ess\_bulk & ess\_tail\\
		\midrule
		$\beta_{0}$ & 0.65 & 0.65 & 0.48 & 0.46 & -0.14 & 1.45 & 1.00 & 32951.17 & 27214.69\\
		$\beta_{1}$ & 0.75 & 0.75 & 0.79 & 0.73 & -0.54 & 2.06 & 1.00 & 32143.98 & 24711.13\\
		$\sigma$ & 2.13 & 2.08 & 0.40 & 0.38 & 1.56 & 2.85 & 1.00 & 29347.16 & 26661.46\\
		\bottomrule
	\end{tabular}
\end{table}

\begin{table}
	\centering
	\caption{\label{tab:summary_simu_left_skewed_TPSC_one_data}Summary statistics and diagnostics for modal regression based on the TPSC-Student-$t$ likelihood in the left-skewed simulation study for one simulated data.}
	\begin{tabular}[t]{lrrrrrrrrr}
		\toprule
		variable & mean & median & sd & mad & q5 & q95 & rhat & ess\_bulk & ess\_tail\\
		\midrule
		$\beta_{0}$ & 0.85 & 0.82 & 0.39 & 0.36 & 0.26 & 1.51 & 1.00 & 19232.81 & 18628.48\\
		$\beta_{1}$ & 0.64 & 0.64 & 0.33 & 0.32 & 0.09 & 1.17 & 1.00 & 30263.89 & 25584.74\\
		$w$ & 0.55 & 0.55 & 0.12 & 0.12 & 0.35 & 0.74 & 1.00 & 19230.80 & 19503.33\\
		$\sigma$ & 0.80 & 0.78 & 0.18 & 0.17 & 0.54 & 1.12 & 1.00 & 30299.78 & 28861.64\\
		$\delta$ & 1.08 & 1.04 & 0.27 & 0.26 & 0.70 & 1.58 & 1.00 & 30531.88 & 28043.37\\
		\bottomrule
	\end{tabular}
\end{table}


\begin{figure}
	\centering
	\includegraphics[width=0.80\columnwidth]{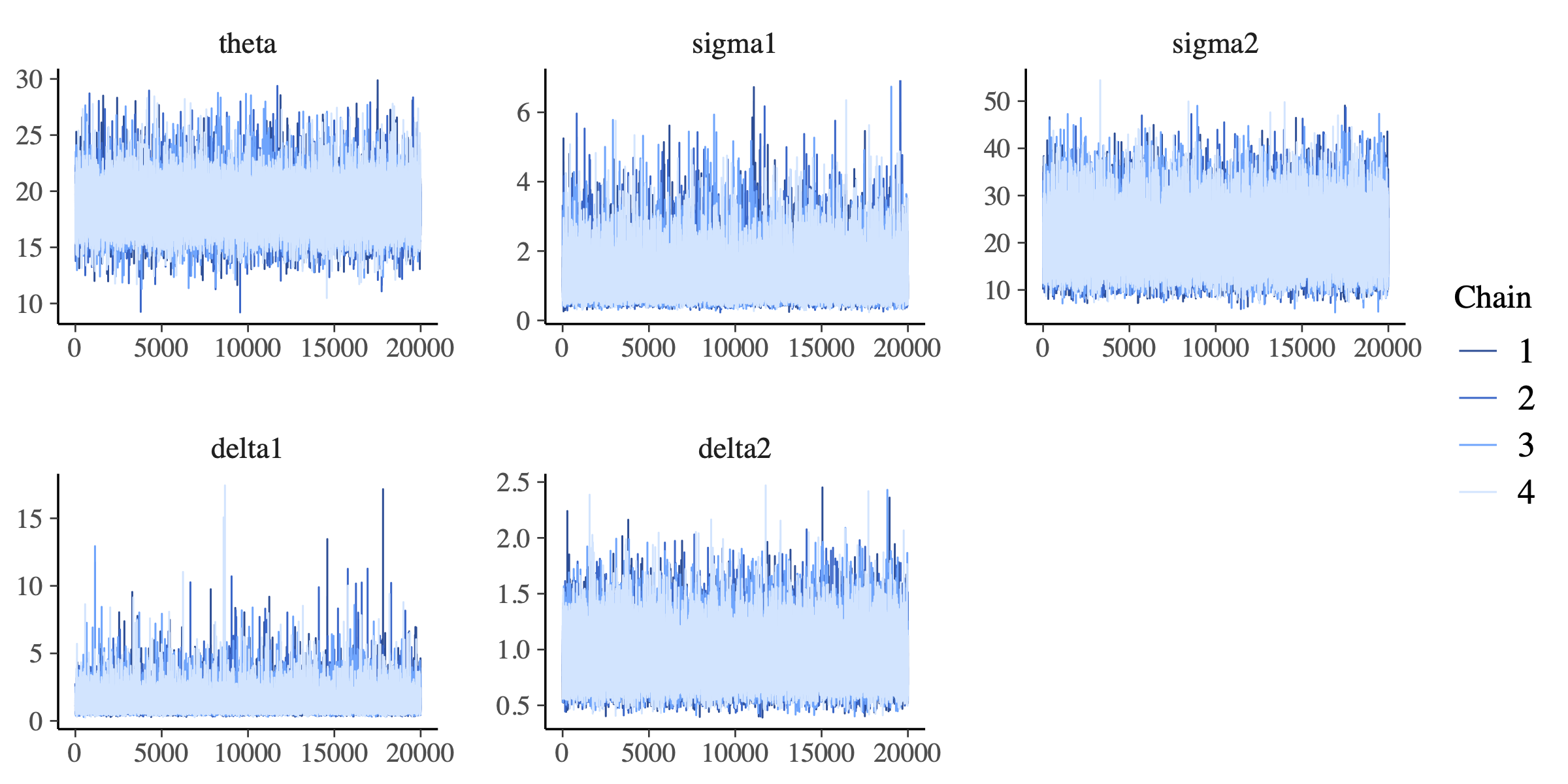}
	\caption{\label{fig:traceplot_bank} Traceplots for the modal regression model based on the DTP-Student-$t$ likelihood fit to the bank deposits data.}
\end{figure}


\begin{figure}
	\centering
	\begin{subfigure}{\columnwidth}
		\centering
		\includegraphics[width=0.80\columnwidth]{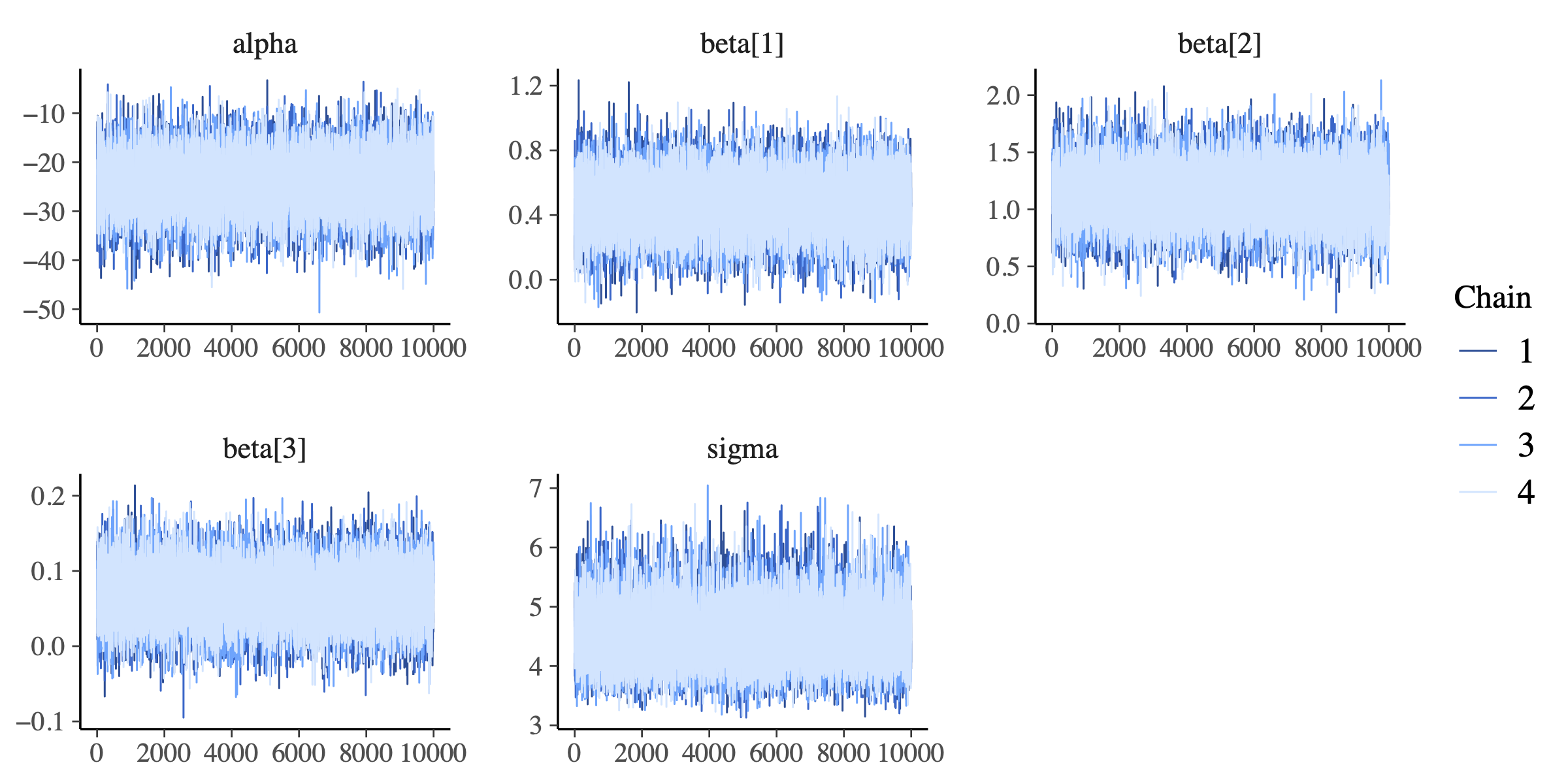}
		\caption{\label{fig:traceplot_crime_normal_included}Traceplots for the mean regression model based on the normal likelihood.}
	\end{subfigure}
	\begin{subfigure}{\columnwidth}
		\centering
		\includegraphics[width=0.80\columnwidth]{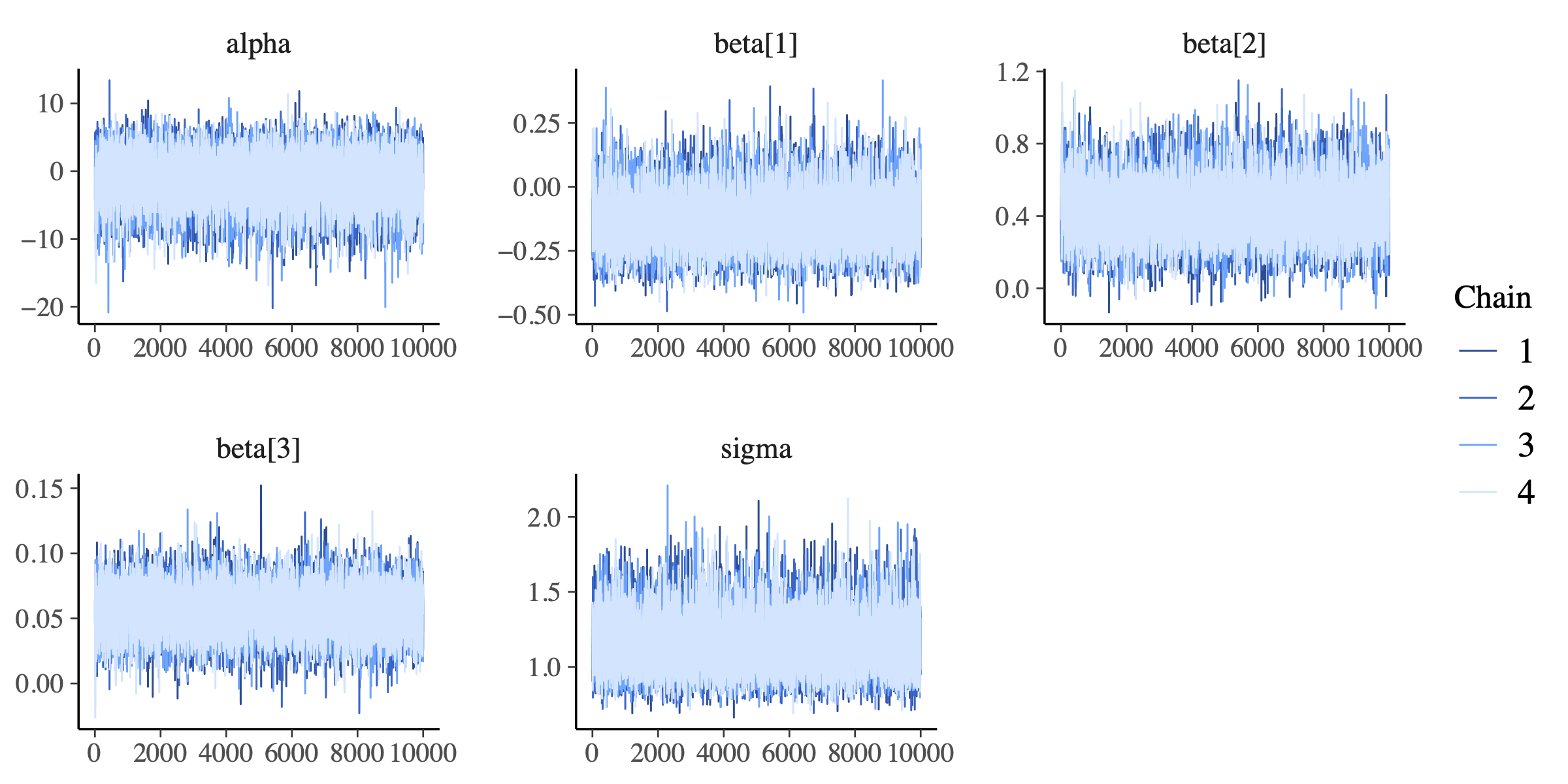}
		\caption{\label{fig:traceplot_crime_ALD_included}Traceplots for the median regression model based on the ALD likelihood.}
	\end{subfigure}
	\begin{subfigure}{\columnwidth}
		\centering
		\includegraphics[width=0.80\columnwidth]{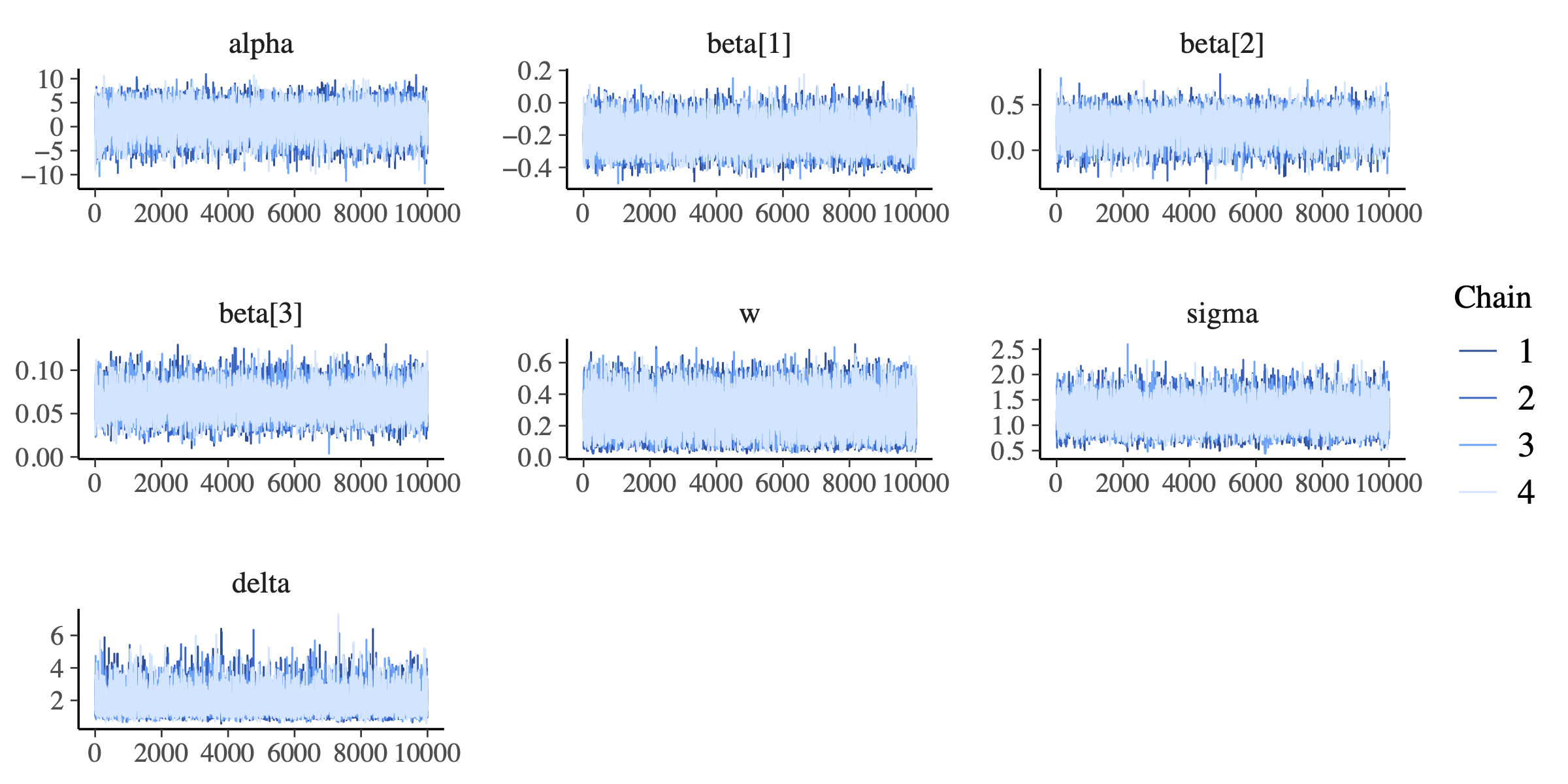}
		\caption{\label{fig:traceplot_crime_TPSC_included}Traceplots for the modal regression model based on the TPSC-Student-$t$ likelihood.}
	\end{subfigure}
	\caption{\label{fig:crime_dc_included} Traceplots for the mean/median/modal regression models fit to the crime data, including D.C.}
\end{figure}

\begin{figure}
	\centering
	\begin{subfigure}{\columnwidth}
		\centering
		\includegraphics[width=0.80\columnwidth]{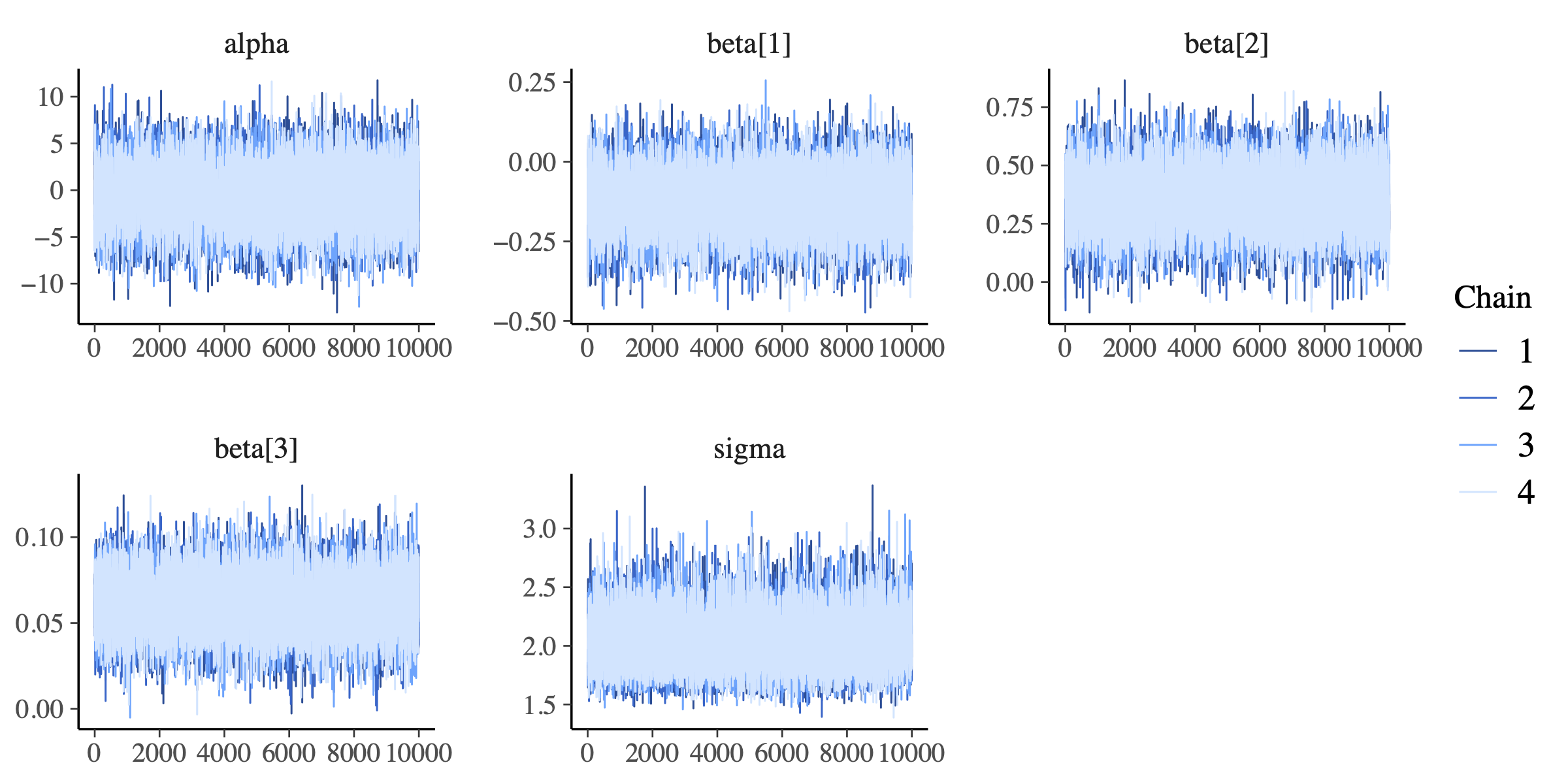}
		\caption{\label{fig:traceplot_crime_normal_dc_excluded}Traceplots for the mean regression model based on the normal likelihood.}
	\end{subfigure}
	\begin{subfigure}{\columnwidth}
		\centering
		\includegraphics[width=0.80\columnwidth]{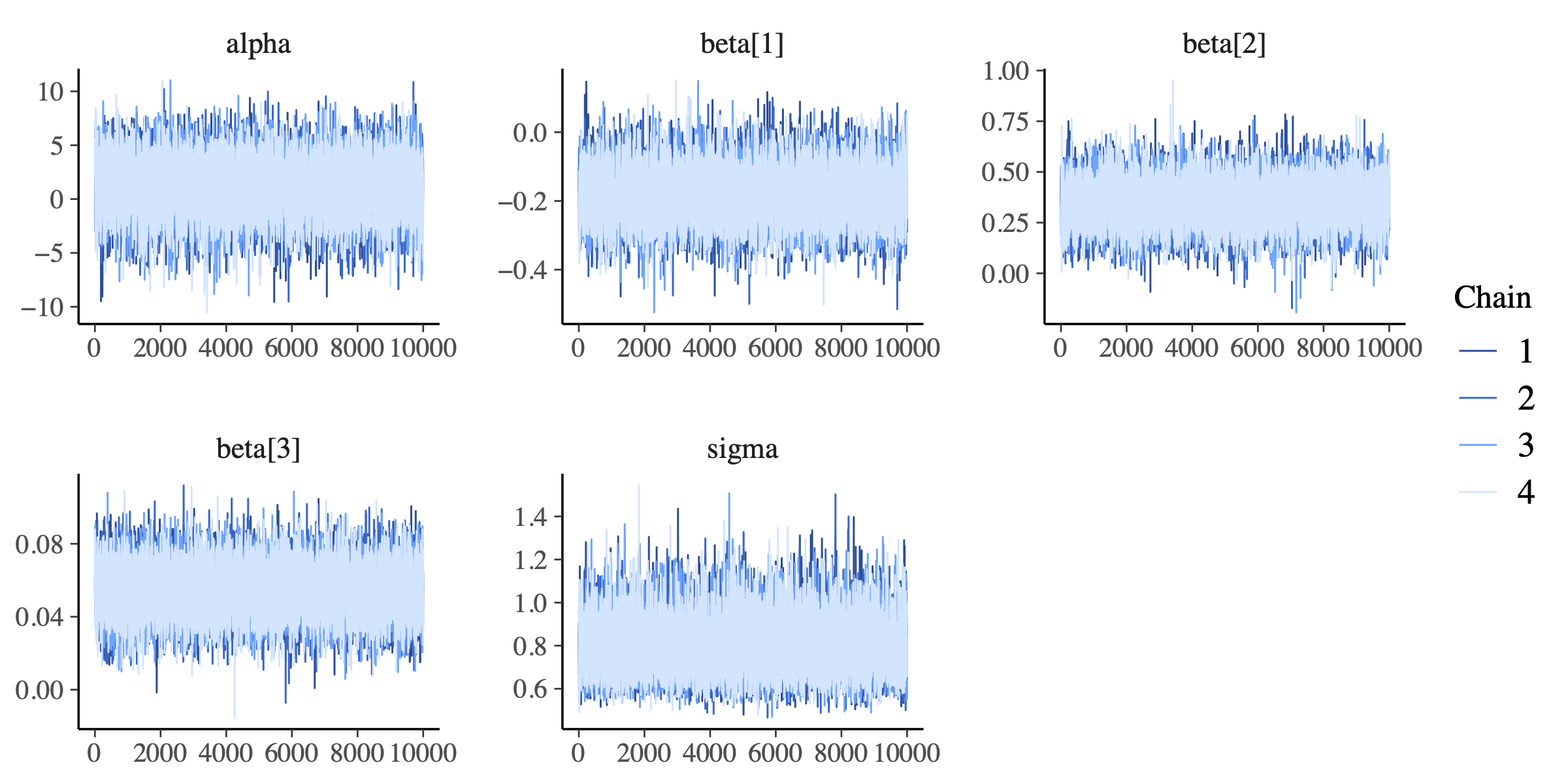}
		\caption{\label{fig:traceplot_crime_ALD_dc_excluded}Traceplots for the median regression model based on the ALD likelihood.}
	\end{subfigure}
	\begin{subfigure}{\columnwidth}
		\centering
		\includegraphics[width=0.80\columnwidth]{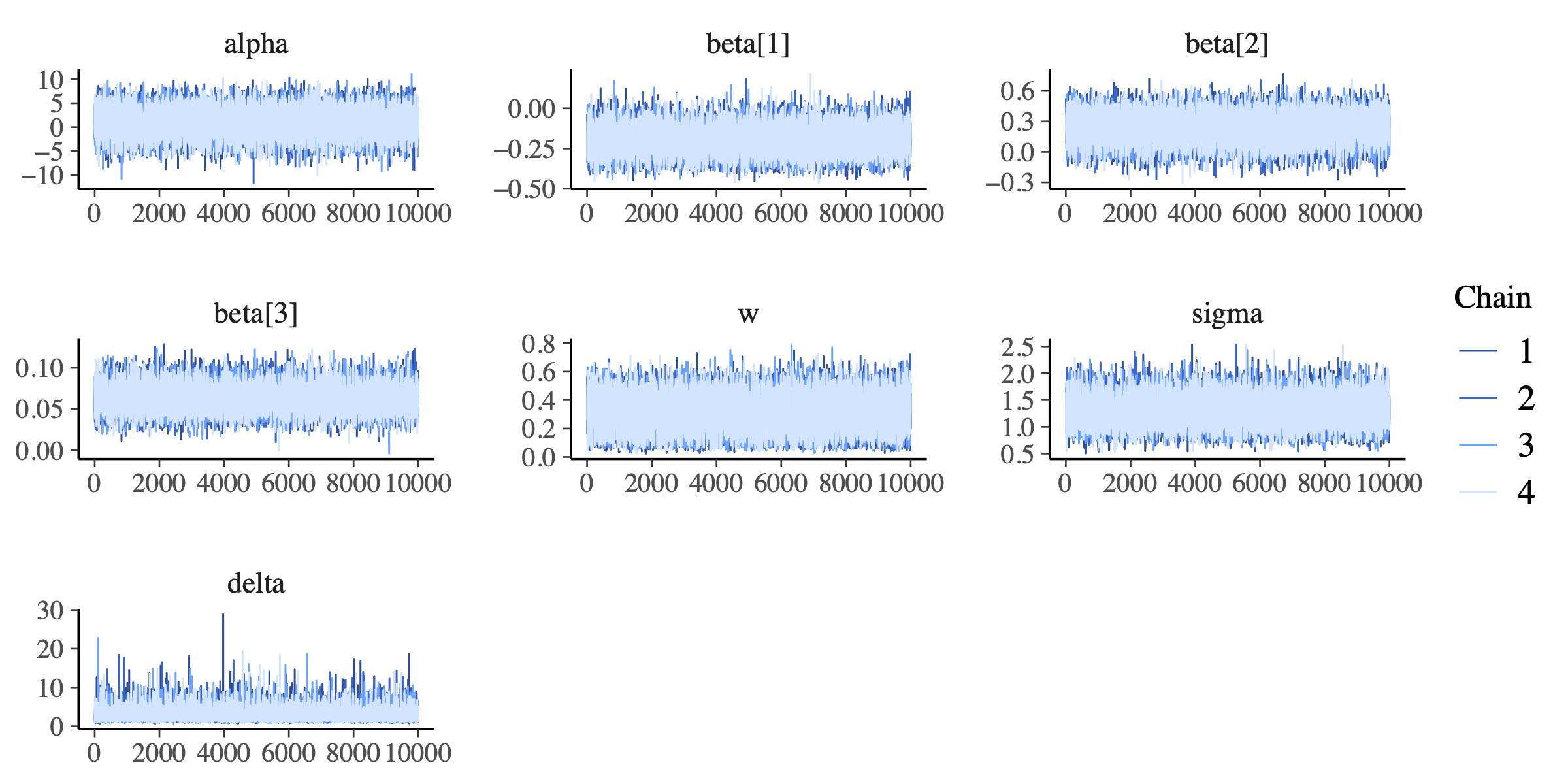}
		\caption{\label{fig:traceplot_crime_TPSC_dc_excluded}Traceplots for the modal regression model based on the TPSC-Student-$t$ likelihood.}
	\end{subfigure}
	\caption{\label{fig:crime_dc_excluded} Traceplots for the mean/median/modal regression models fit to the crime data, excluding D.C.}
\end{figure}


\begin{figure}
	\centering
	\includegraphics[width=0.80\columnwidth]{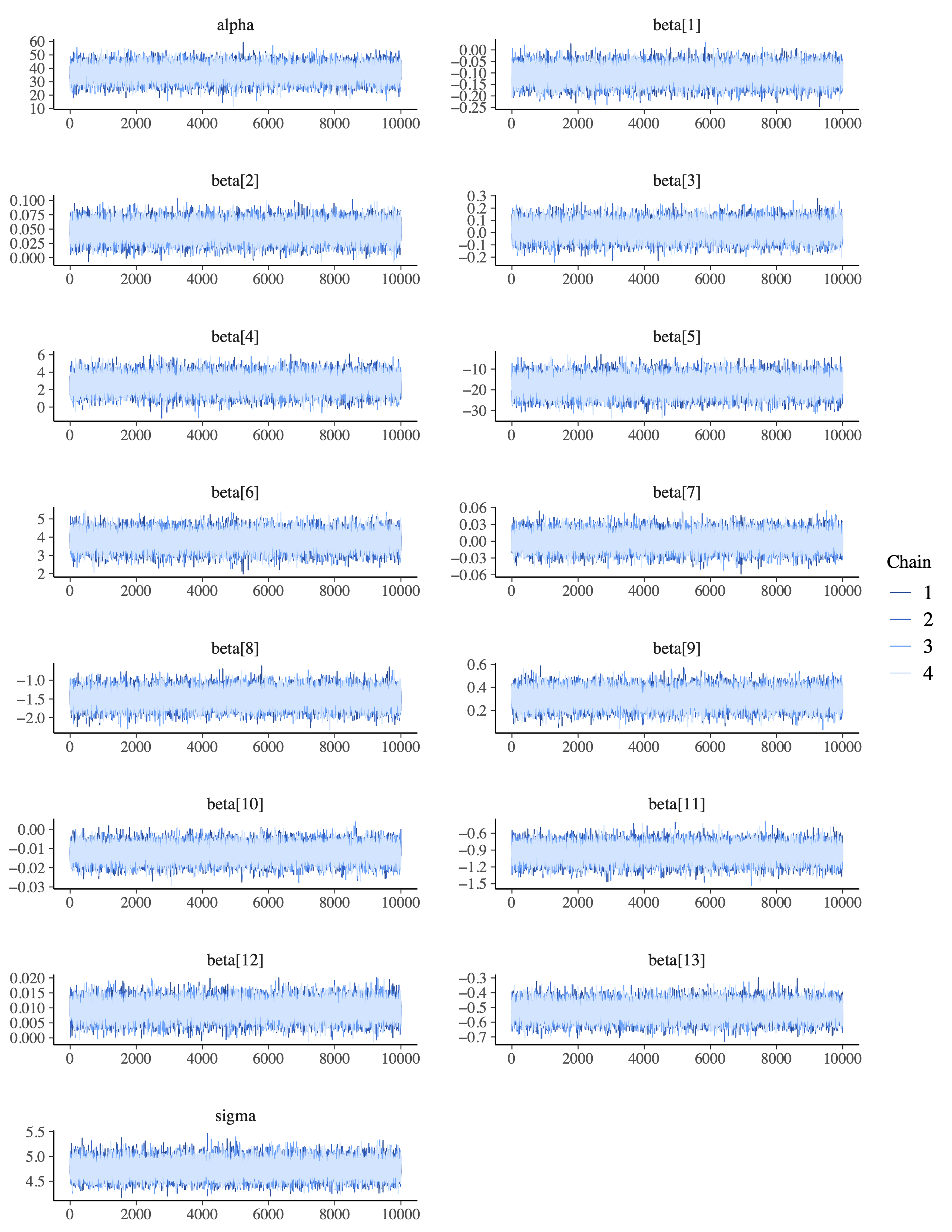}
	\caption{Traceplots for the mean regression model based on the normal likelihood fit to the Boston house price data.}
\end{figure}

\begin{figure}
	\centering
	\includegraphics[width=0.80\columnwidth]{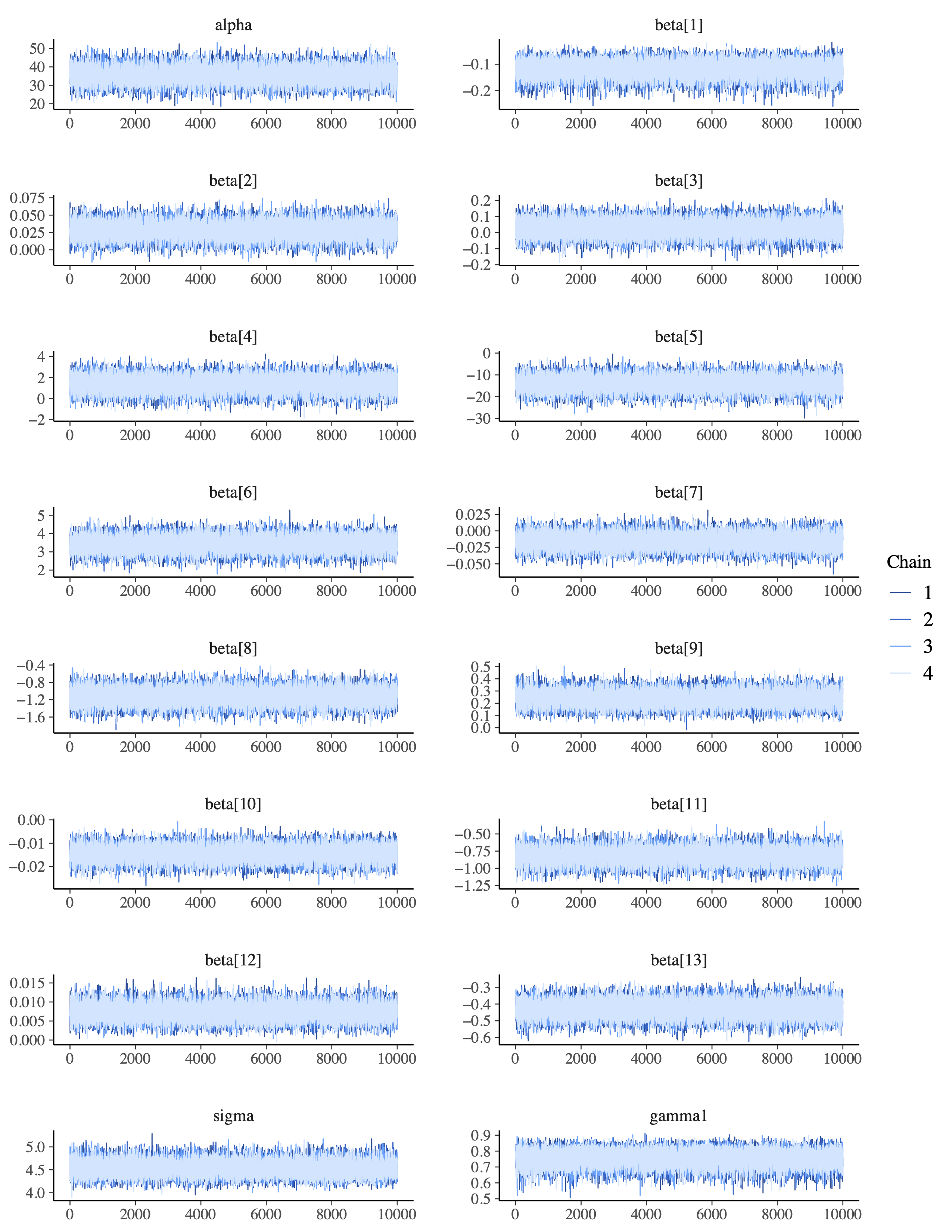}
	\caption{Traceplots for the mean regression model based on the SNCP likelihood fit to the Boston house price data.}
\end{figure}

\begin{figure}
	\centering
	\includegraphics[width=0.80\columnwidth]{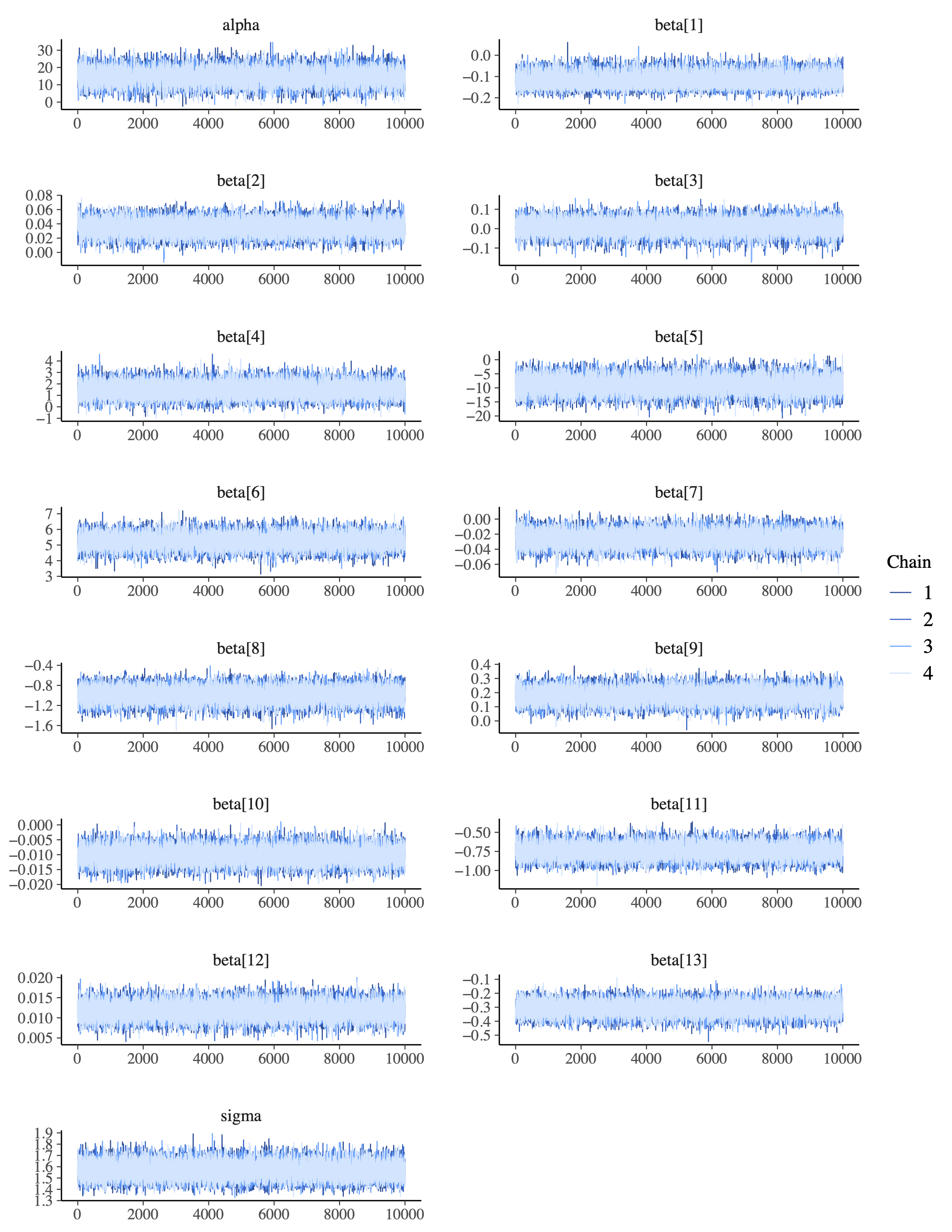}
	\caption{Traceplots for the median regression model based on the ALD likelihood fit to the Boston house price data.}
\end{figure}

\begin{figure}
	\centering
	\includegraphics[width=0.80\columnwidth]{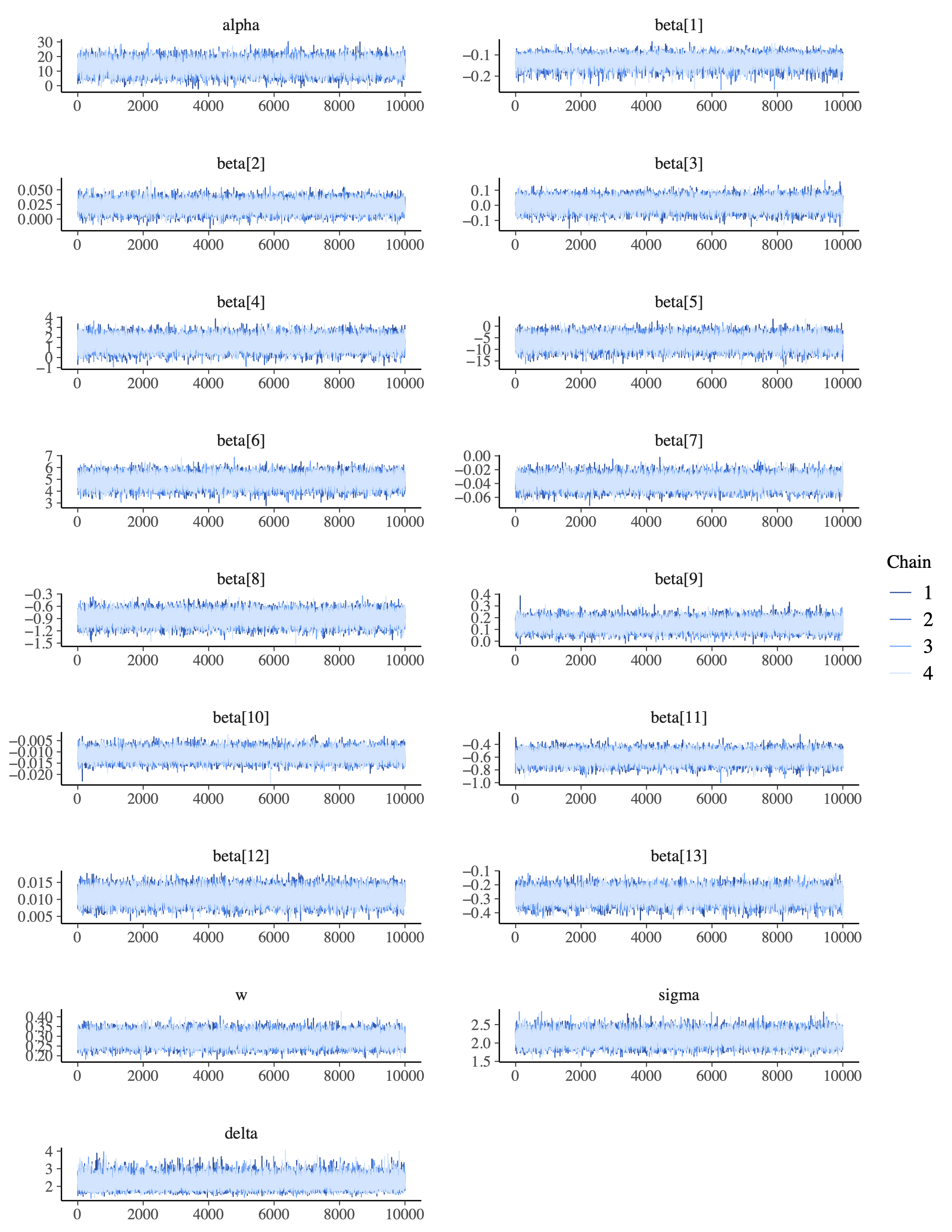}
	\caption{Traceplots for the modal regression model based on the TPSC-Student-$t$ likelihood fit to the Boston house price data.}
\end{figure}


\begin{figure}
	\centering
	\begin{subfigure}{\columnwidth}
		\centering
		\includegraphics[width=0.80\columnwidth]{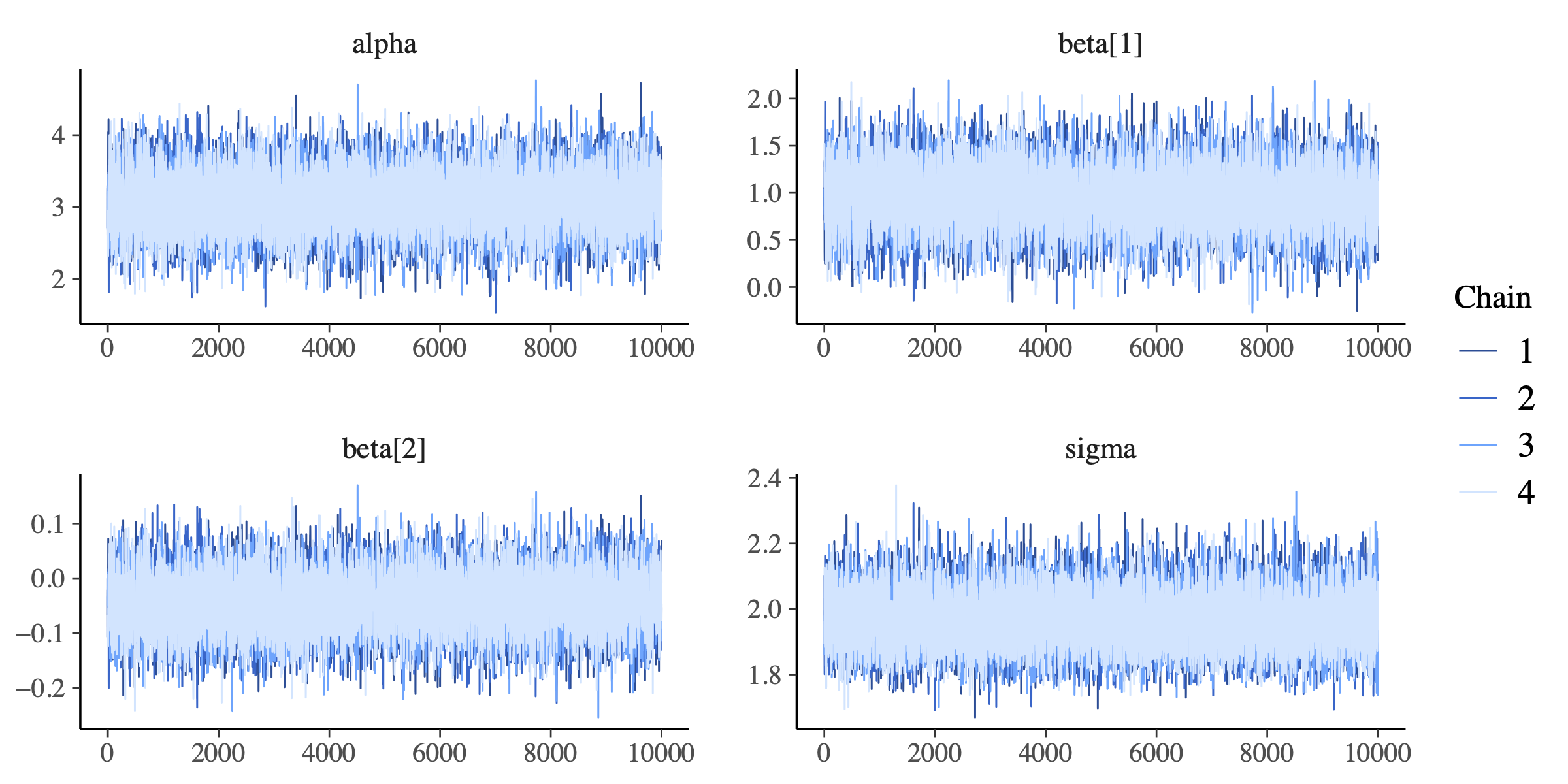}
		\caption{\label{fig:traceplot_serum_normal}Traceplots for the mean regression model}
	\end{subfigure}
	\begin{subfigure}{\columnwidth}
		\centering
		\includegraphics[width=0.80\columnwidth]{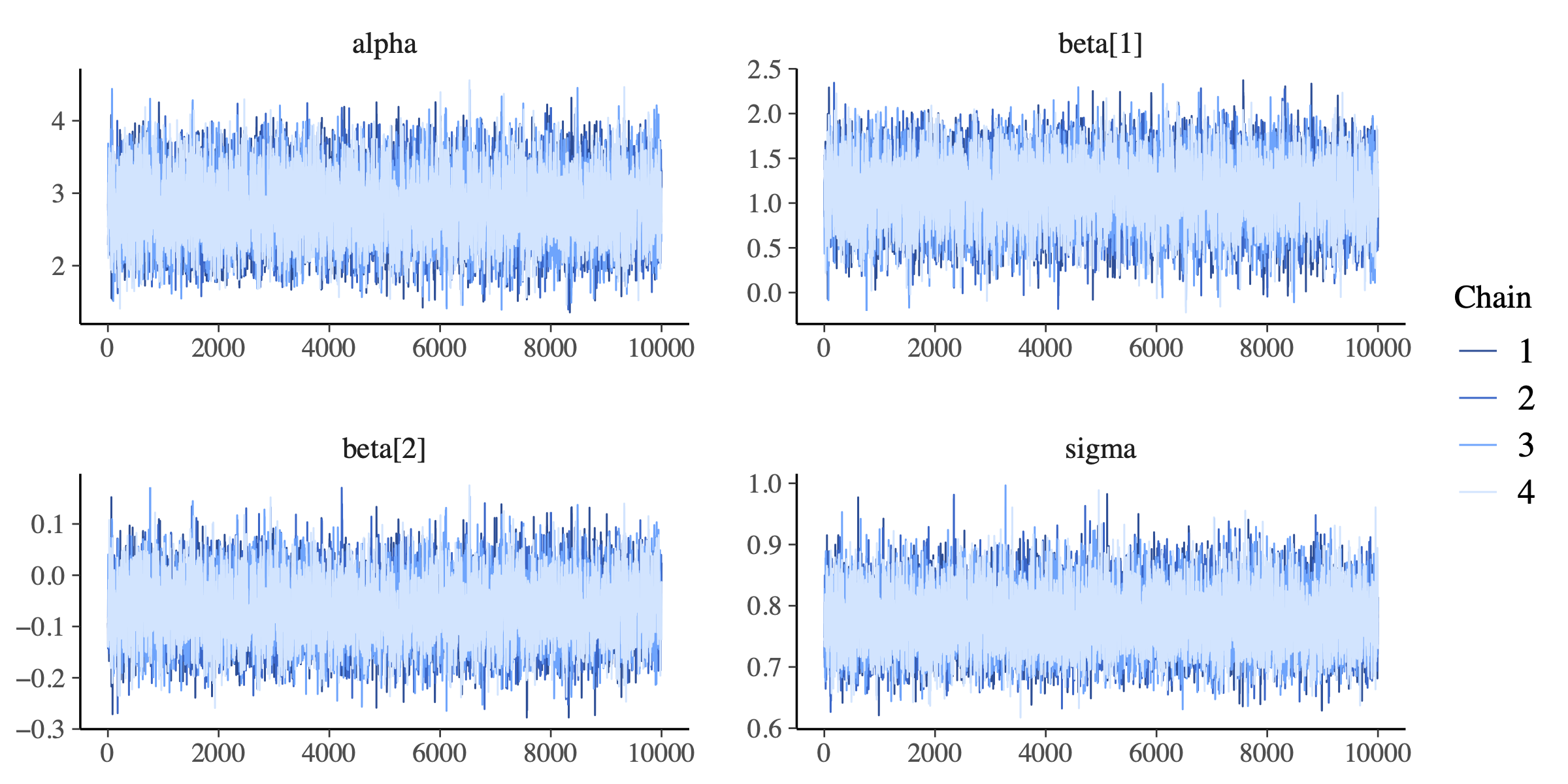}
		\caption{\label{fig:traceplot_serum_ALD}Traceplots for the median regression model}
	\end{subfigure}
	\begin{subfigure}{\columnwidth}
		\centering
		\includegraphics[width=0.80\columnwidth]{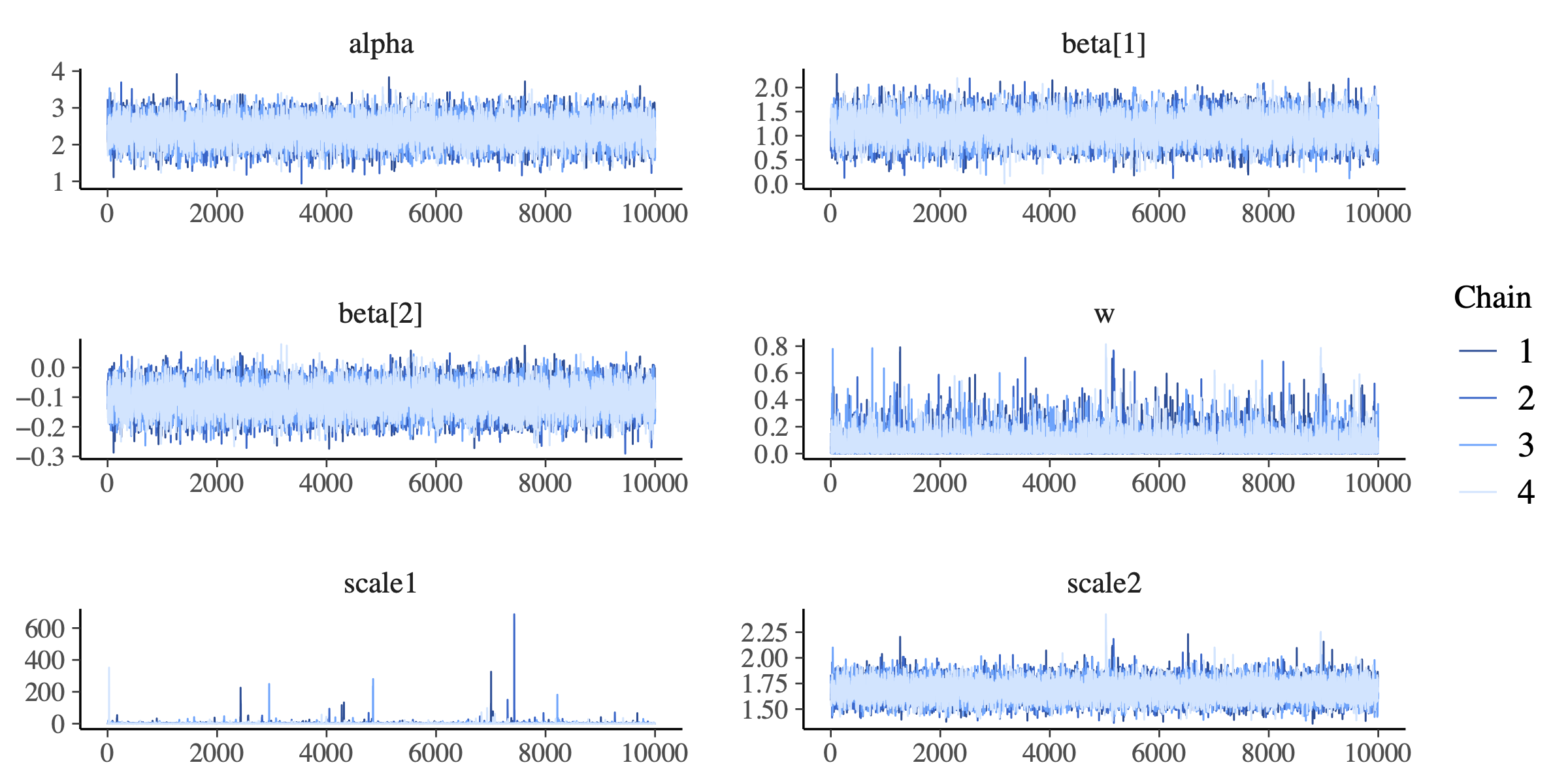}
		\caption{\label{fig:traceplot_serum_TPSC} Traceplots for the modal regression model based on the FG likelihood}
	\end{subfigure}
	\caption{\label{fig:serum} Traceplots for the mean/median/modal regression models fit to the serum data.}
\end{figure}


\begin{figure}
	\centering
	\begin{subfigure}{\columnwidth}
		\centering
		\includegraphics[width=0.78\columnwidth]{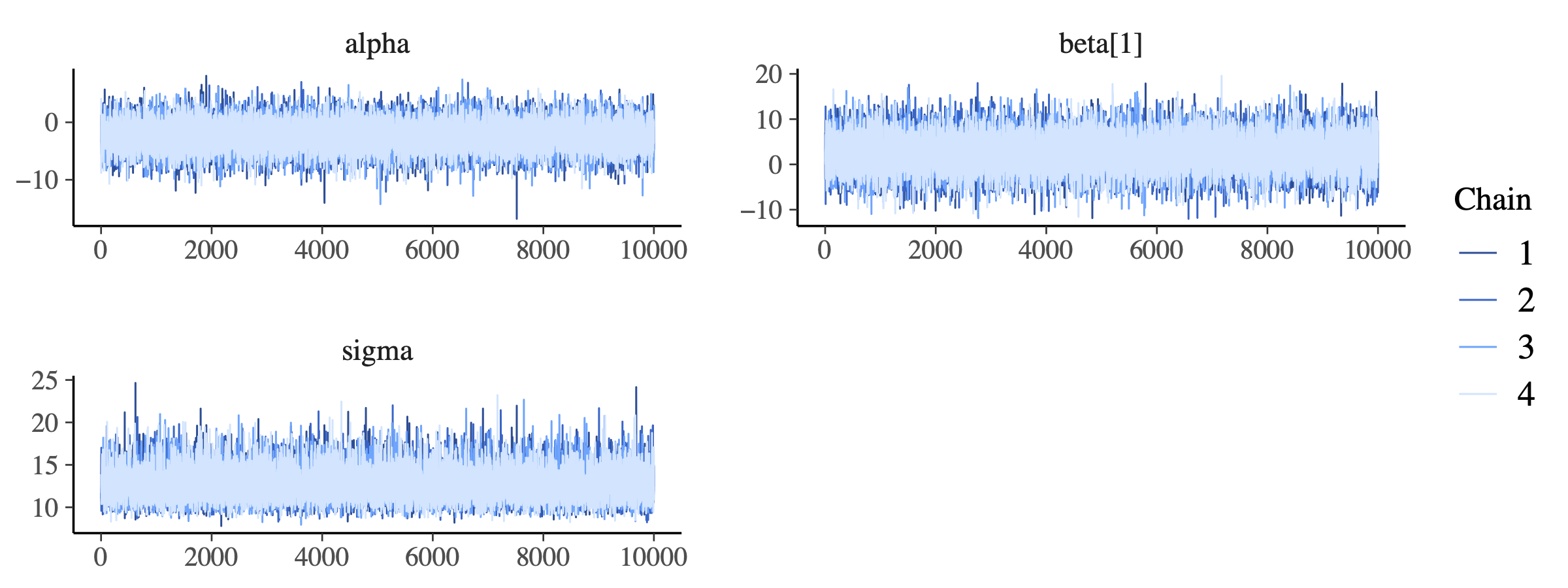}
		\caption{\label{fig:traceplot_left_simu_normal}Traceplots for the mean regression model based on the normal likelihood.}
	\end{subfigure}
	\begin{subfigure}{\columnwidth}
		\centering
		\includegraphics[width=0.78\columnwidth]{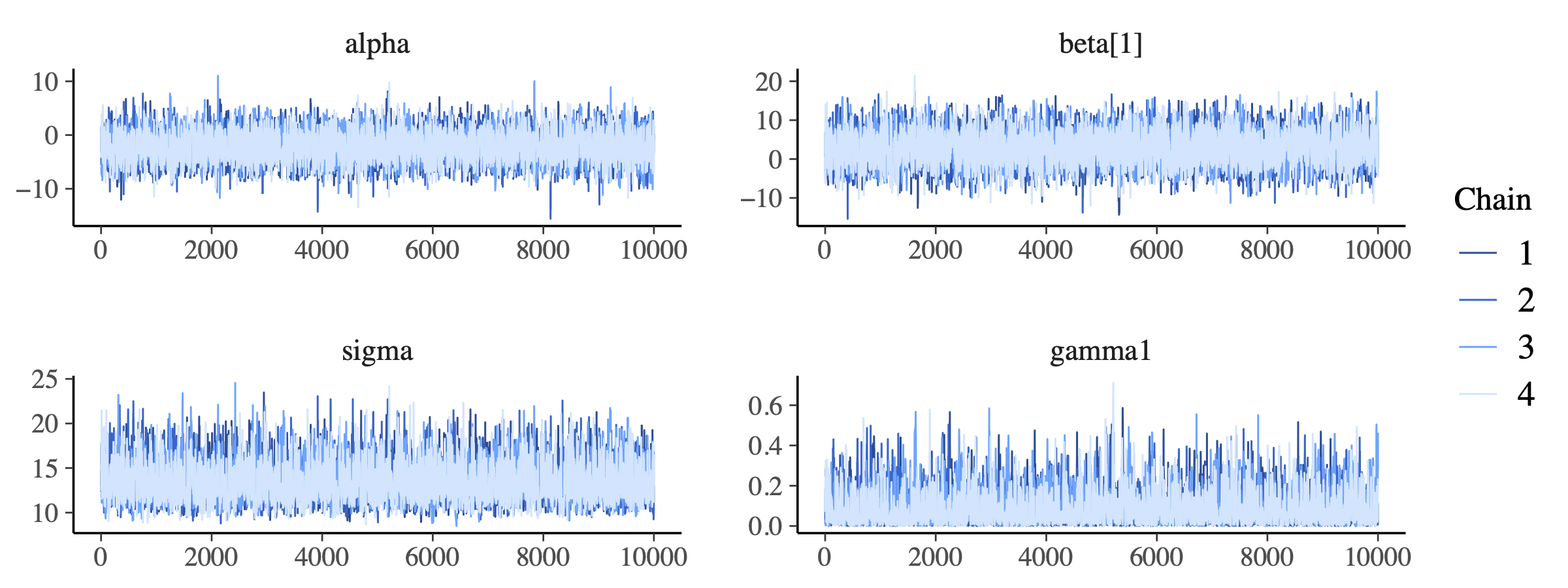}
		\caption{\label{fig:traceplot_left_simu_SNCP}Traceplots for the mean regression model based on the SNCP likelihood.}
	\end{subfigure}
	\begin{subfigure}{\columnwidth}
		\centering
		\includegraphics[width=0.78\columnwidth]{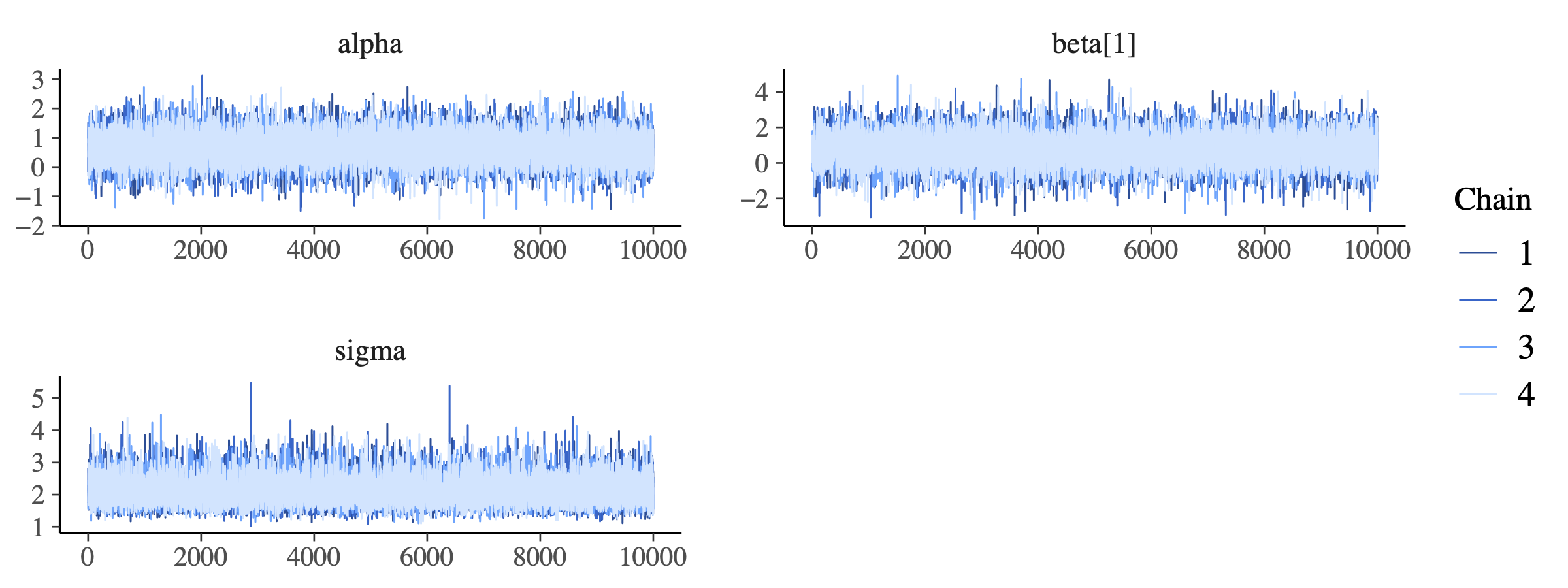}
		\caption{\label{fig:traceplot_left_simu_ALD}Traceplots for the median regression model based on the ALD likelihood.}
	\end{subfigure}
	\begin{subfigure}{\columnwidth}
		\centering
		\includegraphics[width=0.78\columnwidth]{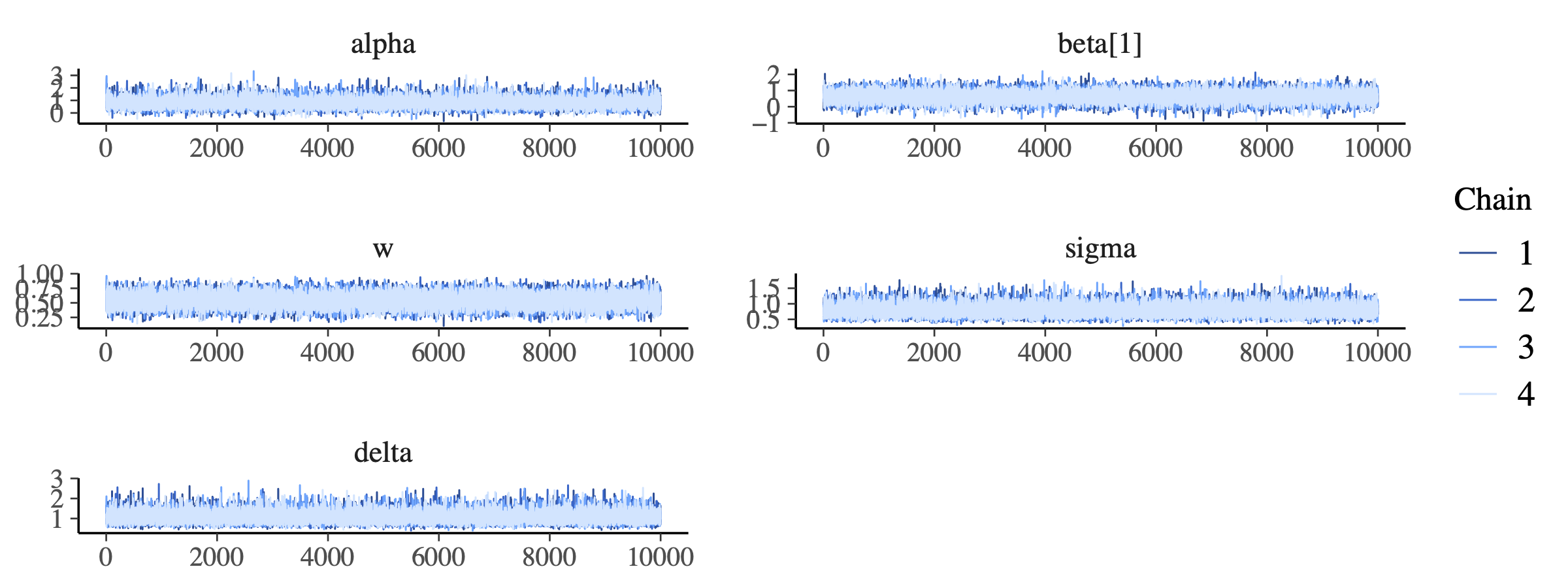}
		\caption{\label{fig:traceplot_left_simu_TPSC}Traceplots for the modal regression model based on the TPSC-Student-$t$ likelihood.}
	\end{subfigure}
	\caption{\label{fig:left_simu} Traceplots for the mean/median/modal regression models from the left-skewed simulation study for one simulated data.}
\end{figure}

\clearpage
\bibliographystyle{apalike}
\bibliography{bibliography}

\end{document}